%% file: paper.tex
\makeatletter\let\@twosidetrue\@twosidefalse\let\@mparswitchtrue\@mparswitchfalse\makeatother %
\documentclass[a4paper,runningheads]{llncs}
\usepackage[T1]{fontenc}

\usepackage{etoolbox}
\newtoggle{arxiv}
\toggletrue{arxiv} %
\iftoggle{arxiv}{
\pagestyle{headings}
\usepackage[paperheight=235mm, paperwidth=155mm,textwidth=12.2cm,textheight=19.3cm]{geometry}
}{}

\usepackage{lmodern}  %
\usepackage{fix-cm}\rmfamily %
\DeclareFontShape{T1}{lmr}{b}{sc}{<->ssub*cmr/bx/sc}{}
\DeclareFontShape{T1}{lmr}{bx}{sc}{<->ssub*cmr/bx/sc}{}
\usepackage[final]{microtype}

\RequirePackage{amsmath,amssymb}
\usepackage{calrsfs}
\DeclareMathAlphabet{\pazocal}{OMS}{zplm}{m}{n}
\usepackage[dvipsnames,table]{xcolor}
\usepackage{xfrac,xspace,xifthen}
\usepackage{array,booktabs,multirow}
\usepackage{mathtools}
\usepackage{bbm}
\usepackage{enumitem}
\usepackage{ragged2e}
\usepackage{pifont}  %
\usepackage{tikz}
\usepackage{pgfplots}\pgfplotsset{compat=1.18}
\usepackage{csquotes}
\usepackage{relsize}
\usepackage{cite}
\usepackage{wrapfig}
\usepackage{calc}  %
\usepackage{marginnote,tokcycle}
\RequirePackage[%
  bookmarks,
  unicode,
  colorlinks=true,
  allcolors=blue!65!black!90,
  breaklinks=true]{hyperref}
\usepackage{bbding}

\let\llncssubparagraph\subparagraph
\let\subparagraph\paragraph
\usepackage{titlesec}
\let\subparagraph\llncssubparagraph

\titlespacing*{\paragraph}{0pt}{0.75ex plus 0.5ex minus 0ex}{*1}

\setitemize{noitemsep,topsep=0pt,parsep=0pt,partopsep=0pt,leftmargin=12.0pt}
\setenumerate{noitemsep,topsep=0pt,parsep=0pt,partopsep=0pt}
\setdescription{noitemsep,topsep=2pt,parsep=0pt,partopsep=0pt,leftmargin=12.5pt}

\graphicspath{{figures/},}  %
\usetikzlibrary{positioning,shapes.geometric,decorations,decorations.footprints}  %
\usetikzlibrary{colorbrewer}
\usepgfplotslibrary{colorbrewer}

\makeatletter
\def\orcidID#1{\textsuperscript{\,\smash{\protect\raisebox{-1.25pt}{\href{http://orcid.org/#1}{\protect\includegraphics[scale=.8]{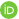}}}}}}
\makeatother

\newboolean{tosubmit}
\setboolean{tosubmit}{false}

\newboolean{showhelpers}
\setboolean{showhelpers}{true}

\definecolor{color1}{RGB}{55,126,184}
\definecolor{color2}{RGB}{228,26,28}
\definecolor{color3}{RGB}{77,175,74}
\definecolor{color4}{RGB}{152,78,163}
\definecolor{color5}{RGB}{255,127,0}
\colorlet{shade1}{black!11}
\colorlet{shade2}{white}
\colorlet{shade3}{Black!11}  %
\colorlet{shade4}{Black!6}  %
\colorlet{colcodekey}{PineGreen!60!green!70!black}
\colorlet{colcodeidentifier}{Sepia}
\colorlet{colcodenumeric}{MidnightBlue!75!black!85}
\colorlet{colcodesyncaction}{violet!90}
\colorlet{colcodecomment}{Black!60}
\colorlet{consolebkg}{Blue!13}
\colorlet{consoleborder}{Black!60}
\definecolor{darkred}{RGB}{150,0,0}
\definecolor{darkblue}{RGB}{0,0,128}
\definecolor{darkgreen}{RGB}{0,128,0}
\definecolor{darkpurple}{RGB}{128,0,128}

\usepackage[nameinlink,capitalize,english]{cleveref} %
\Crefname{figure}{Fig.}{Figs.}
\crefname{figure}{fig.}{figs.}
\Crefname{tabular}{Tab.}{Tabs.}
\crefname{tabular}{tab.}{tabs.}
\Crefname{section}{Sec.}{Sects.}
\crefname{section}{sec.}{sects.}
\Crefname{appendix}{App.}{Apps.}
\crefname{appendix}{app.}{apps.}
\Crefname{equation}{Eq.}{Eqs.}
\crefname{equation}{eq.}{eqs.}
\creflabelformat{equation}{#2#1#3}
\Crefname{example}{Ex.}{Exs.}
\crefname{example}{ex.}{exs.}

\newcolumntype{L}[1]{>{\RaggedRight\let\newline\\\arraybackslash}p{#1}}
\newcolumntype{C}[1]{>{\Centering\let\newline\\\arraybackslash}p{#1}}
\newcolumntype{R}[1]{>{\RaggedLeft\let\newline\\\arraybackslash}p{#1}}

\colorlet{hlboxcoldraw}{black!65}  %
\colorlet{hlboxcolfill}{black!6}   %
\newsavebox\BODYBOX  %
\def\BODY{\unhbox\BODYBOX}
\NewDocumentEnvironment{hlbox}{ O{} }
  {\vspace{-.5ex plus .3ex minus .2ex}%
	\begin{center}\begin{lrbox}{\BODYBOX}\begin{minipage}{.95\linewidth}%
	\ifthenelse{\equal{#1}{}}{}{\textbf{#1}. }}
  {\end{minipage}\end{lrbox}\begin{tikzpicture}%
	\node[rectangle,rounded corners=.3mm,inner sep=1.3ex,thick,
	      draw=hlboxcoldraw,fill=hlboxcolfill] {\BODY};
	\end{tikzpicture}\end{center}%
	\vspace{-.6ex plus .3ex minus .2ex}}

\renewcommand{\emptyset}{\varnothing}
\newcommand{\NN}{\ensuremath{\mathbb{N}}\xspace}  %
\newcommand{\RR}{\ensuremath{\mathbb{R}}\xspace}  %
\newcommand{\RRpos}{\ensuremath{\RR_{\geqslant0}}\xspace}  %

\newcommand{\tool}[1]{\textsc{#1}}
\newcommand{\lang}[1]{\textsc{#1}}
\newcommand{\modest}{\lang{Modest}\xspace}
\newcommand{\toolset}{\tool{Modest Toolset}\xspace}
\newcommand{\thetool}{\tool{modes}\xspace}
\newcommand{\FIG}{\tool{Fig}}
\newcommand{\jani}{\lang{Jani}\xspace}

\newcommand{\eg}{e.g.\ }
\newcommand{\ie}{i.e.\ }
\newcommand{\etal}{et al.\xspace}
\newcommand{\wrt}{w.r.t.\xspace}

\newcommand{\set}[1]{\ensuremath{\{\,#1\,\}}}
\newcommand{\tuple}[1]{\ensuremath{\langle #1 \rangle}}

\newcommand{\defeq}{\mathrel{\vbox{\offinterlineskip\ialign{\hfil##\hfil\cr{\tiny \rm def}\cr\noalign{\kern0.30ex}$=$\cr}}}}

\newcommand{\Dist}[1]{\ensuremath{\mathit{Dist}({#1})}\xspace}

\newcommand{\support}[1]{\ensuremath{\mathit{spt}({#1})}\xspace}

\newcommand{\totalreward}{\ensuremath{\mathit{PR}}}

\newcommand{\lambdafun}[1]{\ensuremath{\lambda\,{#1}\mathbf{.}\:}}

\DeclareMathAlphabet{\mathbbmsl}{U}{bbm}{m}{sl}

\DeclareMathAlphabet{\mathcal}{OMS}{cmsy}{m}{n}
\SetMathAlphabet{\mathcal}{bold}{OMS}{cmsy}{b}{n}

\makeatletter%
\g@addto@macro\normalsize{%
  \setlength\abovedisplayskip{3pt}%
  \setlength\belowdisplayskip{3pt}%
  \setlength\abovedisplayshortskip{-3pt}%
  \setlength\belowdisplayshortskip{3pt}%
}%
\makeatother

\usetikzlibrary{calc}
\pgfplotsset{
	table/col sep=comma,
	allplot/.style={
		x label style={at=(current axis.north),anchor=south,inner sep=3pt,font=\footnotesize},
		axis x line*=bottom,
		every tick label/.append style={font=\tiny},
	},
	confplot/.style={
		allplot,
		xtick={0,0.5,1}, xticklabels={0,0.5,1},
		xmin=0,xmax=1,ymin=0.85,ymax=1,
		ylabel=\empty,ytick=\empty,
	},
	confleftplot/.style={
		confplot,
		y label style={at=(current axis.west),anchor=south west,inner xsep=3pt, inner ysep=8pt},
		ylabel=\scriptsize $p_{cov}$,
		ytick={0.85,0.9,1}, yticklabels={0.85,0.9,1}
	},
	mainbodyplot/.style={
		width=4.95cm,height=3.5cm,
	},
	appendixplot/.style={
		width=6.5cm,height=5cm,
	}
}%

\begin{document}

\title{%
Sound Statistical Model Checking\\ for Probabilities and Expected Rewards%
\thanks{
This work was supported
by the DFG through the Cluster of Excellence EXC 2050/1 (CeTI, project ID 390696704, as part of Germany's Excellence Strategy)
and the TRR 248 (see \href{https://perspicuous-computing.science/}{perspicuous-computing.science}, project ID 389792660),
by the European Union's Horizon 2020 research and innovation programme under Marie Sk{\l}odowska-Curie grant agreements 101008233 (MISSION), 101034413 (IST-BRIDGE), and 101067199 (ProSVED),
by the EU under NextGenerationEU projects D53D23008400006 (Smartitude) under MUR PRIN 2022 and PE00000014 (SERICS) under MUR PNRR,
by the Interreg North Sea project STORM\_SAFE,
and
by NWO VIDI grant VI.Vidi.223.110 (TruSTy).
}\iftoggle{arxiv}{\\ (extended version)}{}
}
\titlerunning{Sound Statistical Model Checking for Probabilities and Expected Rewards}
\author{
\mbox{Carlos E.\ Budde\inst{1,2}\orcidID{0000-0001-8807-1548}}
\and
\mbox{Arnd Hartmanns\inst{3}\orcidID{0000-0003-3268-8674}}
\and
Tobias Meggendorfer\inst{4}\orcidID{0000-0002-1712-2165}\!
\and\\
Maximilian Weininger\inst{5}\orcidID{0000-0002-0163-2152}
\and
Patrick Wienh\"oft\inst{6,7}\orcidID{0000-0001-8047-4094}\,\smash{\raisebox{0.75pt}{\Envelope}}
}
\authorrunning{C.E. Budde, A. Hartmanns, T. Meggendorfer, M. Weininger, P. Wienhöft}
\institute{%
Technical University of Denmark, Lyngby, Denmark
\and
University of Trento, Trento, Italy
\and
University of Twente, Enschede, The Netherlands
\and
Lancaster University Leipzig, Germany
\and
Institute of Science and Technology Austria, Klosterneuburg, Austria
\and
Technical University Dresden, Germany
$\cdot$ \email{patrick.wienhoeft@tu-dresden.de}
\and
Centre for Tactile Internet with Human-in-the-Loop (CeTI), Dresden, Germany
}

\maketitle

\begin{abstract}
Statistical model checking estimates probabilities and expectations of interest in probabilistic system models by using random simulations.
Its results come with statistical guarantees.
However, many tools use \emph{unsound} statistical methods that produce incorrect results more often than they claim.
In this paper, we provide a comprehensive overview of tools and their correctness, as well as of sound methods available for estimating probabilities from the literature.
For expected rewards, we investigate how to bound the path reward distribution to apply sound statistical methods for bounded distributions, of which we recommend the Dvoretzky-Kiefer-Wolfowitz inequality that has not been used in SMC so far.
We prove that even reachability rewards can be bounded in theory, and formalise the concept of limit-PAC procedures for a practical solution.
The \thetool SMC tool implements our methods and recommendations, which we use to experimentally confirm our results.
\end{abstract}

\section{Introduction}
\label{sec:intro}

Statistical model checking (SMC)~\cite{YS02} estimates quantities of interest by sampling a large number $k$ of random runs from a compact executable model of a probabilistic system. %
Typical quantities of interest are reachability probabilities and expected rewards, to query for \eg reliability or performance measures~\cite{BHHK10}. %
A \emph{sound} statistical model checker delivers results guaranteed to be \emph{probably approximately correct} (PAC), \ie 
it returns %
a confidence interval $I$ with $|I| \leq 2\varepsilon$ (\enquote{$\varepsilon$-approximately correct}) such that the probability for $I$ to contain the (unknown) true value $x$ is higher than a given confidence level $\gamma$ (\enquote{probably correct}).

When applied judiciously, SMC can perform extremely well~\cite{YKNP06} and easily beat probabilistic model checking (PMC)~\cite{Bai16,BAFK18} tools that rely on exhaustive state space exploration as well as partial exploration tools~\cite{ABHK18,KM20,DBLP:conf/cav/MeggendorferW24} in competitions when PAC results are allowed~\cite{BHKKPQTZ20}.
It is widely implemented in tools such as \tool{C-smc}~\cite{CDGL21}, \tool{Cosmos}~\cite{BBD+15}, \FIG~\cite{Bud22a}, \tool{Hypeg}~\cite{PER17}, \tool{modes}~\cite{BDHS20} of the \tool{Modest Toolset}~\cite{HH14}, \tool{MultiVeStA}~\cite{SV13,GRV17}, \tool{Plasma Lab}~\cite{LST16}, \tool{Prism}~\cite{KNP11}, \tool{Sbip}~\cite{NMB+18}, or \tool{Uppaal Smc}~\cite{DLLMP15}.
It has been applied to case studies ranging from hardware~\cite{MTWPS23,RLHBRCZ21} over biology~\cite{Zul15} to cybersecurity~\cite{Bud22b,LMZ23}.
New SMC tools such as \tool{Smc Storm}~\cite{LKP24} are now being developed in industrial contexts.

However, as we detail in \Cref{tab:tools}, many existing and new SMC implementations either are unsound (\ie they do not deliver PAC guarantees), or inefficient (\ie they use statistical methods that need unnecessarily many samples).
The unsoundness is often due to computing confidence intervals via approaches that rely on the central limit theorem, %
while the inefficiency is notably due to the widespread use of the Okamoto bound~\cite{Oka59} for estimating probabilities. %

\paragraph{Our contribution} is a comprehensive treatment of the problem of efficiently obtaining sound SMC results when estimating probabilities as well as expected rewards.
We review the statistical methods available for probabilities %
in \Cref{sec:ForProbabilities}, which forms the basis for \Cref{tab:tools}.
For expected rewards, in \Cref{sec:ForRewards}, we provide a novel fully sound approach for the instantaneous and cumulative cases, prove that sound SMC is possible for reachability rewards in theory, and give a practically useful method.
We implemented our methods and recommendations in the \thetool SMC tool (\Cref{sec:Implementation}) to experimentally confirm our findings in \Cref{sec:Experiments}.

\medskip
SMC for reachability \textbf{probabilities} comes down to estimating binomial proportions,
a well-studied problem in statistics. %
Sound methods for \textbf{expected rewards}, on the other hand, have been an open problem.
Here, we need to estimate the mean $x$ of the path reward distribution $\mu$, whose shape is unknown and which can have unbounded support.
For this problem, no PAC statistical methods exist.
Thus, to obtain a sound SMC approach for expected rewards, we must (1)~use structural information to soundly reduce to case of bounded support $[a, b]$ %
to then (2)~employ an appropriate statistical method for this case.

We review the methods available for Step~2 in \Cref{sec:StatisticsAB}, recommending the use of the Dvoretzky-Kiefer-Wolfowitz inequality (DKW)~\cite{DKW56}.
This inequality provides a very strong and versatile result that allows the derivation of useful confidence intervals for the mean even for conservative values for $a$ and $b$, yet has been curiously ignored in the SMC community so far.

For Step~1, we distinguish two cases in \Cref{sec:BoundingRewards}:
First, for %
(step- or time-bounded) \textbf{cumulative} and \textbf{instantaneous rewards}, we can derive safe and practical values for $a$ and $b$ given an upper bound on $r_\mathit{\!max}$, the highest reward assigned to any state, which can typically be obtained from the model's syntax.
For (unbounded) \textbf{reachability rewards}, we introduce \emph{bounding sets} that provide a means to ignore very large path rewards while introducing an error of at most $\varepsilon' < \varepsilon$.
We prove that a bounding set can be obtained for every finite discrete-time Markov chain (DTMC),
given only bounds on certain parameters of the DTMC which can typically be derived syntactically, too.
Yet, the resulting bounding set is not practical, thus serving only as a proof of the possibility of sound SMC for reachability rewards. %
In practice, we propose to use the DKW to obtain guaranteed \emph{lower bounds} that provably converge to $x$ as $k \to \infty$. %

\paragraph{Our focus}
is on \emph{estimating} probabilities and \emph{undiscounted} expected rewards given either $k$ or an \emph{absolute} error $\varepsilon$.
We briefly comment on \emph{hypothesis testing} where appropriate, referring the reader to dedicated works on hypothesis testing in SMC like Reijsbergen \etal's~\cite{RBSH15} for more details, cautioning that they may not emphasise soundness.
Undiscounted rewards are standard in verification while discounting is ubiquitous in machine learning. %
It is easy to obtain good bounds $[a, b]$ on discounted rewards and thus apply the methods we review in \Cref{sec:StatisticsAB} efficiently.
For %
rare events~\cite{RT09} or very large expected rewards, one may want to specify a \emph{relative error} $\varepsilon \cdot x$; %
we mention some methods specific for this case.

We consider SMC as in the original papers by Younes and Simmons~\cite{YS02} and Hérault \etal\cite{HLMP04},
motivated by the state space explosion problem which PMC faces for finite-but-large models of realistic applications, and the lack of scalable PMC approaches for non-Markovian models like stochastic automata (IOSA)~\cite{DM18} or %
HPnGs~\cite{NPR20}.
Thus, we sample runs from a (mostly) black-box model using $\mathcal{O}(1)$ memory to estimate global quantities of interest.
This is in contrast to ``model-based SMC''~\cite{AKW19,AGKM22,MWW24}, which aims to apply PMC-like methods to black-box systems with simulation access by \emph{learning} a model, in particular its transition probabilities, requiring memory quadratic in the number of states. %

\paragraph{Related soundness.}
The formal methods community values trustworthy results with clear guarantees on possible analysis errors.
For example, after finding that the common stopping criterion of the value iteration algorithm %
can lead to arbitrarily wrong results~\cite{BCCFKKPU14,HM18}, the problem was addressed in many settings~\cite{BKLPW17,ACDKM17,QK18,HK20,EKKW22,AEKSW22,DBLP:conf/lics/KretinskyMW23}.
Yet, in SMC, the issue of soundness has received little attention. %
Only recently, a survey of sound and unsound methods for estimating probabilities appeared~\cite{MWW24}, which our recommendations in \Cref{sec:ForProbabilities} are based~on.

\section{Background}
\label{sec:Background}

We write $\mathbbm{1}_\mathit{pred}$ for the indicator function of $\mathit{pred}$: $\mathbbm{1}_\mathit{pred}(x) = 1\text{ if }\mathit{pred}(x)\text{ else }0$.
A \emph{probability distribution} over a countable set $S$ is a function $\mu\colon S \to [0, 1]$ such that $\sum_{s \in S} \mu(s) = 1$.
Its \emph{support} is $\support{\mu} \defeq \set{ s \in S \mid \mu(s) > 0 }$.

\subsubsection{Models.}
\label{sec:Models}

The assumption of SMC is that models are given in a higher-level formalism---like HPnGs~\cite{NPR20}, IOSA~\cite{DM18}, \jani~\cite{BDHHJT17}, \modest~\cite{BDHK06,HHHK13}, or the \tool{Prism} language~\cite{KNP11}---that allows behaviours to be randomly sampled without having to create an in-memory state space. %
Their semantics are some form of Markov process; %
we focus on the special case of DTMCs to simplify the presentation.
All our methods immediately apply in the general setting or can be extended.

\begin{definition}
A \emph{discrete-time Markov chain} (DTMC) is a tuple $\tuple{S, R, T, s_I}$ of a finite set of \emph{states} $S$, a \emph{reward} function $R\colon S \to \RRpos$, an \emph{initial state} $s_I \in S$, and a \emph{transition} function $T\colon S \to \Dist{S}$ mapping each state to a probability distribution over successor states.
A (finite) \emph{path} $\pi$ ($\pi_\mathit{fin}$) is (a prefix of) an infinite sequence $\pi = s_0\,s_1\ldots \in S^\omega$ such that $s_0 = s_I$ and $\forall i\colon T(s_i)(s_{i+1}) > 0$.
\end{definition}
A DTMC induces a probability measure $\mathbb{P}$ over sets of paths that, intuitively, corresponds to multiplying the probabilities along the path (see \eg\cite[Chp.~10]{BK08}).
Abusing notation, we also use $\pi$ or $\pi_\mathit{fin}$ to refer to the set of a path's states. %
We write $\pi[i]$ for $s_i$, the path's $(i+1)$-th state, and $\mathit{idx}(\pi, S') = \min\{i \in \NN \mid s_i \in S'\}$ for the index of the first state in $S' \subseteq S$ on $\pi$, with $\mathit{idx}(\pi, S') = \infty$ if $\pi \cap S' = \varnothing$.
We write %
$r_\mathit{\!max} \defeq \max\,\set{R(s) \mid s \in S}$ for a DTMC's maximum reward and $p_\mathit{min} \defeq \min(\set{T(s)(s') \mid s, s' \in S} \setminus \set{0})$ for its minimum probability.
We assume that, from only the higher-level formalism's syntax, we can efficiently obtain bounds $\overline{|S|} \geq |S|$, $\overline{r}_\mathit{\!max} \geq r_\mathit{\!max}$, and $0 < \underline{p}_\mathit{min} \leq p_\mathit{min}$.

\subsubsection{Properties.}
\label{sec:Properties}

Every property to be model-checked can be cast as the expected value $\mathbb{E}(X)$ \wrt $\mathbb{P}$ of a random variable $X$ that maps paths to values in $\RRpos$.
We consider the following kinds of properties, some of which take a step bound $c \in \NN$ or a set of \emph{goal states} $G \subseteq S$ specified as part of the model:
$$
\begin{aligned}
P_{\diamond\,G} &\defeq \mathbb{E}(\lambdafun{\pi} \mathbbm{1}_{G \cap \pi \neq \emptyset})
&\qquad\text{\it(reachability probability)}\phantom{_{\!\bot}}\\
P_{\diamond\,G}^{\leq c} &\defeq \mathbb{E}(\lambdafun{\pi} \mathbbm{1}_{\mathit{idx}(\pi, G) \leq c})
&\qquad\text{\it(bounded reach.\ probability)}_{\!\bot}\\
E_{\leq c} &\defeq \mathbb{E}(\lambdafun{\pi} \textstyle\sum_{i = 0}^c R(\pi[i]))
&\qquad\text{\it(cumulative reward)}_{\!\bot}\\
E_{\diamond\,G} &\defeq \mathbb{E}(\lambdafun{\pi} \textstyle\sum_{i = 0}^{\mathit{idx}(\pi, G)} R(\pi[i]))
&\qquad\text{\it(reachability reward)}\phantom{_{\!\bot}}\\
E_{\diamond\,G}^{\leq c} &\defeq \mathbb{E}(\lambdafun{\pi} \textstyle\sum_{i = 0}^{\min\{ \mathit{idx}(\pi, G),\, c \}} R(\pi[i]))
&\qquad\text{\it(bounded and reach.\ reward)}_{\!\bot}\\
E_{=c} &\defeq \mathbb{E}(\lambdafun{\pi} R(\pi[c]))
&\qquad\text{\it(instantaneous reward)}_{\!\bot}\\
E_{=G} &\defeq \mathbb{E}(\lambdafun{\pi} R(\pi[\mathit{idx}(\pi, G)]))
&\qquad\text{\it(reach-instant reward)}\phantom{_{\!\bot}}\\
\end{aligned}
$$
Rewards are obtained upon entering states.
$E_{\diamond\,G}$ and $E_{=G}$ are defined to be~$\infty$ if $\mathbb{P}(\set{ \pi \mid \mathit{idx}(\pi, G) = \infty }) > 0$~\cite{FKNP11}.
The properties marked $_{\bot}$ are \emph{bounded}; all others are \emph{unbounded}.
The former are typical for SMC,
required by \eg \tool{Plasma Lab}~\cite[Table~1]{LST16} and \tool{Smc Storm}~\cite[Sect.\ 2.1]{LKP24}.
Unbounded reachability probabilities and rewards, on the other hand, are standard in PMC and dominate model collections like the Quantitative Verification Benchmark Set (QVBS)~\cite{HKPQR19}.

\subsubsection{Statistical model checking.}
\label{sec:SMC}

At its core, SMC is Monte Carlo simulation~\cite{Kre16,AP18,LLTYSG19}:
randomly generate a predetermined number $k$ of paths, or \emph{simulation runs}%
, that give rise to samples $X_1, \ldots, X_k$ of the random variable $X$; then compute the empirical mean $\hat{X} \defeq \frac{1}{k}\sum_{i=1}^{k} X_i$, and perform a \emph{statistical evaluation} to obtain a confidence interval $I = [l, u] \ni \hat{X}$ at predetermined \emph{confidence level}~$\gamma$.

\medskip\noindent
\textit{Simulation.}
How to obtain simulation runs is specific to each higher-level formalism. %
We only assume that a method $\mathit{sample}(M, \mathit{prop})$ exists that, given a model $M$, implements the random variable $X$ of property $\mathit{prop}$, \ie that (pseudo\nobreakdash-)randomly generates a path $\pi_\mathit{fin}$ through $M$'s semantics according to $\mathbb{P}$ that is long enough to evaluate $X$ and returns $X(\pi_\mathit{fin})$.
For bounded properties, the ``long enough'' criterion is straightforward:
just generate paths of length $c$. %

To end a simulation run for $P_{\diamond\,G}$, we must determine whether it entered a bottom strongly connected component (BSCC) without goal states.
For $E_{\diamond\,G}$ and $E_{=G}$ to be finite, we must determine whether a non-goal BSCC \emph{exists}. %
BSCCs can be detected statistically by sampling given some structural information (such as $\underline{p}_\mathit{min}$)~\cite{ADKW20}.
Yet this requires storing a set of visited states that can be as large as $S$, breaking the $\mathcal{O}(1)$ memory property of SMC.
Additionally, some fraction of $1 - \gamma$ must be devoted to all these tests (see \eg\cite{DHKP17}).
However, most verification models---such as those in the QVBS---are structured so that
(i)~for reachability probabilities, all BSCCs contain only one state, and
(ii)~the goal state sets in reachability rewards are reached with probability~$1$,
which allows for an efficient but limited stopping criterion.
We follow this assumption in this paper.

\paragraph{Statistical evaluation.}
If we repeat the SMC procedure $m$ times to obtain con\-fi\-dence intervals $I_1, \ldots, I_m$, we find some of them might be incorrect, i.e.\ $\mathbb{E}(X) \notin I_i$ for some $i$;
occasionally obtaining an ``incorrect'' result is the nature of a statistical approach based on %
sampling.
The (a priori) probability for a correct result is the \emph{coverage probability} $p_\mathit{cov}(k) = \lim_{m \to \infty} \frac{\mathit{cov}_m}{m}$, where $\mathit{cov}_m$ denotes the amount of correct confidence intervals.
We call an SMC procedure \textbf{sound} if it is guaranteed to provide \emph{probably approximately correct} (PAC) results:
Given $k$ and confidence level $\gamma$, it has $p_\mathit{cov}(k) \geq \gamma$ while producing intervals of width $|I| \leq 2\varepsilon$.
Then, the midpoint of this interval is $\varepsilon$-close to the true mean of $X$ with high probability.

Sound SMC results are obtained by employing an appropriate \emph{statistical (evaluation) method} (SM) that relates $k$ (or the concrete $X_i$), $\gamma$, and $\varepsilon$ to ensure the PAC requirement, with two values given and the third under control of the SM.
We consider two settings:
The \textbf{fixed-}$\boldsymbol{k}$ setting, where $\gamma$ and $k$ are given while $\varepsilon$ is determined by the SM, and the \textbf{sequential} setting, where $\gamma$ and $\varepsilon$ are given so that the SM must determine~$k$.
In the latter, $k$ can precomputed from $\gamma$ and $\varepsilon$, or it can be determined by a \emph{truly sequential} SM that continuously checks whether enough samples have been gathered to be $\gamma$-confident of an interval $I$ with $|I| \leq 2\varepsilon$.
We always assume $\gamma$ to be given, a typical value being $\gamma = 0.95$.

\section{Sound SMC for Probabilities}
\label{sec:ForProbabilities}
\label{sec:Statistics01}

For probabilities, \ie $P_{\diamond\,G}$ and $P_{\diamond\,G}^{\leq c}$, each simulation run is a Bernoulli trial with outcome $X_i \in \set{0, 1}$.
Thus, the SM samples from a binomial distribution with success probability $p$.
After $k$ samples, it observed $k_s = \sum_{i=1}^k X_i$ successes and has empirical mean $\hat p = \frac{k_s}{k} = \hat{X}$. %
Constructing a $\gamma$-confidence interval around $\hat p$ is a well-studied problem in statistics, resulting in many SMs for this task.
We often abbreviate $\delta = 1 - \gamma$ for readability.

Meggendorfer~\etal~\cite[Sec.~3]{MWW24} survey SMs in the context of \enquote{model-based SMC} for Markov decision processes, where individual transition probabilities are estimated to \enquote{learn} the model. Methods for this specific case also apply to SMC for reachability (and other \enquote{qualitative} 0/1 properties) in DTMCs.
Hence, we recap their survey of SMs, extending it with examples and plots.
Moreover, \cite{MWW24} only considers the fixed-$k$ setting, whereas we also discuss the sequential setting and hypothesis testing.
Our survey is the basis for the tool comparison in \Cref{sec:Implementation}, where we show that existing tools use unsound and/or inefficient methods.

\subsection{Unsound Methods}
\label{sec:ForProbabilitiesUnsound}

Denote by $p_\mathit{cov}(k, p)$ the coverage probability that the SM at hand attains given success probability~$p$.
Many of the commonly used SMs for binomial proportions only guarantee an \emph{average} coverage probability of $\gamma$, \ie $\int_0^1 p_{cov}(k,p)\,\mathrm{d}p \geq \gamma$.
This is not in line with the frequentist definition of a confidence interval %
and \textcolor{BrickRed}{\bf not} sufficient for sound SMC%
, producing too many incorrect results for certain values of $p$.
We instead require that $\inf_{p=0}^1 p_{cov}(k,p) \geq \gamma$.

As per~\cite{MWW24}, unsound methods include those based on the central limit theorem (CLT)%
, in particular the textbook Wald interval, the Wilson score interval%
, the Agresti-Coull interval~\cite{AC98}, the Arcsine interval, and the Logit interval.

\paragraph{The Wilson score interval with continuity correction}
(Wilson/CC) \cite{New98} complements the CLT by adding adjustment terms to improve coverage.
However, Newcombe already observed slightly below-nominal coverage~\cite[Table II]{New98}, and \cite{MWW24} confirms the insufficient coverage for high confidence levels and $p$ close to $0$ or $1$. %

\paragraph{Sequential setting.}
Given $\varepsilon$ instead of $k$, Chow and Robbins~\cite{CR65} show that the coverage of %
constructing a Wald interval after every sample and terminating once this interval has half-width $\leq \varepsilon$ goes to $\gamma$ as $\varepsilon \to 0$.
For any concrete $\varepsilon > 0$, however, coverage $\geq \gamma$ may not be achieved and thus this procedure is not sound for SMC.
Reijsbergen \etal\cite{RBSH15} adapt it to perform hypothesis testing with some extra parameters that reduce, but do not eliminate, the chance for incorrect results.
The sequential methods proposed by Chen~\cite{Che15} are \emph{empirically} sound, \ie they appear to produce sound results in practice, but soundness is not proven for $p$ close to~$\frac{1}{2}$, and even for $p$ away from~$\frac{1}{2}$ only as a one-sided version.

\subsection{Sound Methods}
\label{sec:ForProbabilitiesSound}

\paragraph{Okamoto bound.}
In 1959, Okamoto~\cite{Oka59}
proved that, for binomial proportions,
$
\mathbb{P}(\hat p - p \geq \varepsilon) \leq \mathrm{e}^{-2k\varepsilon^2}.
$
We want\iftoggle{arxiv}{\footnote{The directions of the following inequality, and of the first one on the next line, are incorrect in the version of this paper published in the TACAS 2025 proceedings~\cite{publishedversion}.}}{} $\mathbb{P}(|\hat p - p| \geq \varepsilon) \leq \delta$, giving
$$
\textstyle
\frac{\delta}{2} \geq \mathrm{e}^{-2k\varepsilon^2}
\qquad\Leftrightarrow\qquad
\varepsilon \geq \sqrt{\frac{\ln \frac{2}{\delta}}{2k}}
\qquad\Leftrightarrow\qquad
k \geq \frac{\ln \frac{2}{\delta}}{2\varepsilon^2}
$$
by distributing $\delta$ symmetrically.
Thus the interval $I_\mathit{Oka} = [\hat p - \varepsilon, \hat p + \varepsilon]$ always has coverage $\geq \gamma$ when $\varepsilon$, $k$, and $\delta$ satisfy the above inequalities.
This bound is also referred to as Hoeffding bound~\cite{Hoe63} after his more general inequality, see \Cref{sec:StatisticsAB}.

\paragraph{Clopper-Pearson interval.}
The ``exact'' binomial interval by Clopper and Pearson~\cite{CP34} guarantees coverage $\geq \gamma$ for all $p$.
One of several ways to compute~it~is
$$
I_\mathit{CP} = [ B({\delta}/{2}, k_s, k - k_s + 1), B(1 - {\delta}/{2}, k_s + 1, k - k_s) ]
$$
where $B(p, \alpha, \beta)$ is the $p$-quantile of the Beta$(\alpha, \beta)$ distribution. %

\paragraph{Blyth-Still-Casella and Wang.}
The approaches of Blyth-Still-Casella~\cite{Cas86} and Wang~\cite{Wan14} are also sound and produce shortest intervals in a specific sense, but are intricate to implement and computationally very expensive.

\paragraph{Sequential setting.}
The Okamoto bound provides $\varepsilon$, $k$, or $\delta$ given the other two; thus it applies to the sequential setting, too.
For the Clopper-Pearson interval, we use the recent result that its number of required samples is maximal when $\hat p = \frac{1}{2}$~\cite{MWW24}.
Based on this worst case, we can precompute the smallest $k$ where the interval width is $\leq 2\varepsilon$ and perform a fixed-$k$ evaluation.

\subsection{Discussion}
\label{sec:ForProbabilitiesDiscussion}

\paragraph{Recipes for sequential SMs.}
The minimum $k$ may depend on $p$; \eg for Clopper-Pearson, lower $k$ suffice for $p$ close to $0$ or $1$.
A truly sequential method could exploit this.
Any fixed SM can be converted to truly sequential in the Chow-Robbins style by checking after every sample if half-width $\leq \varepsilon$ is met, resulting in methods like ``sequential Clopper-Pearson''.
They however are \textcolor{BrickRed}{\bf not sound} as in general sample mean $\hat{p}$ and precision $|I|$ are correlated~\cite{JSD19}.

We may first spend a fraction of the ``error budget'' $\delta$ to get a rough interval estimate of $p$, and then calculate the number of samples required (given the remaining part of $\delta$) based on the worst case in this interval, \eg the value closest to $\frac{1}{2}$ for Clopper-Pearson as in~\cite{BS24}.
Jégourel \etal's two-step approach~\cite{JSD19} similarly uses the Massart bound which improves on Okamoto's if $p$ is known to be away from $\frac{1}{2}$.
While \textcolor{OliveGreen}{\bf sound} and better than precomputation, these are two-step (generalisable to $n$-step), not truly sequential, approaches.

Frey~\cite{Fre10} calculates a $\delta^{*}$ a priori so that, if a confidence level of $\gamma^{*} = 1 - \delta^{*}$ is used in each iteration of a Chow-Robbins-style sequentialisation of a sound fixed SM, the overall coverage probability \textcolor{OliveGreen}{\bf sound}ly comes out to $\gamma$.
However, computing such a $\delta^{*}$ becomes very hard already for small $k$ (around $100$-$1000$).
The \textsc{AdaSelect}~\cite{DGW02} and \textsc{EBStop}~\cite{MSA08} algorithms are sequential methods for the relative error setting,
recently generalised by Parmentier and Legay~\cite{PL24}. %

\begin{figure}[t]
	\centering
	\pgfplotsset{every axis/.append style={mainbodyplot}}
	\begin{minipage}{0.36\textwidth}
	\centering
	\begin{tikzpicture}
		\begin{axis}[confleftplot,xlabel=Wald interval\strut]
			\addplot [no marks,draw=black] table [x index=0, y index=1] {experiments/binom_coverage/coverage_fixed_n_stdci_normal.csv};
			\addplot [no marks,draw=red,thick] coordinates {(0,0.9) (1,0.9)};
		\end{axis}
	\end{tikzpicture}
	\end{minipage}%
	\begin{minipage}{0.32\textwidth}
	\centering
	\begin{tikzpicture}
		\begin{axis}[confplot,xlabel=Okamoto\strut]
			\addplot [no marks,draw=black] table [x index=0, y index=1] {experiments/binom_coverage/coverage_fixed_n_okamoto.csv};
			\addplot [no marks,draw=red,thick] coordinates {(0,0.9) (1,0.9)};
		\end{axis}
	\end{tikzpicture}
	\end{minipage}%
	\begin{minipage}{0.32\textwidth}
	\centering
	\begin{tikzpicture}
		\begin{axis}[confplot,xlabel=Clopper-Pearson\strut]
			\addplot [no marks,draw=black] table [x index=0, y index=1] {experiments/binom_coverage/coverage_fixed_n_clopper_pearson.csv};
			\addplot [no marks,draw=red,thick] coordinates {(0,0.9) (1,0.9)};
		\end{axis}
	\end{tikzpicture}
	\end{minipage}

	\begin{minipage}{0.36\textwidth}
	\centering
	\begin{tikzpicture}
		\begin{axis}[confleftplot,xlabel=Chow-Robbins method\strut]
			\addplot [no marks,draw=black] table [x index=0, y index=1] {experiments/binom_coverage/coverage_fixed_eps_chow_robbins.csv};
			\addplot [no marks,draw=red,thick] coordinates {(0,0.9) (1,0.9)};
		\end{axis}
	\end{tikzpicture}
	\end{minipage}%
	\begin{minipage}{0.32\textwidth}
	\centering
	\begin{tikzpicture}
		\begin{axis}[confplot,xlabel=seq.\ Clopper-Pearson\strut]
			\addplot [no marks,draw=black] table [x index=0, y index=1] {experiments/binom_coverage/coverage_fixed_eps_sequential_clopper_pearson.csv};
			\addplot [no marks,draw=red,thick] coordinates {(0,0.9) (1,0.9)};
		\end{axis}
	\end{tikzpicture}
	\end{minipage}%
	\begin{minipage}{0.32\textwidth}
	\centering
	\begin{tikzpicture}
		\begin{axis}[confplot,xlabel=Chen's method\strut]
			\addplot [no marks,draw=black] table [x index=0, y index=1] {experiments/binom_coverage/coverage_fixed_eps_chen.csv};
			\addplot [no marks,draw=red,thick] coordinates {(0,0.9) (1,0.9)};
		\end{axis}
	\end{tikzpicture}
	\end{minipage}
	\caption{Coverage for fixed (top, $k=50$) and sequential methods (bottom, $\varepsilon=0.05$).}
	\label{fig:CoverageProbabilities}
\end{figure}

\begin{example}[Soundness]
\label{ex:SoundnessNumerically}
To evaluate SMs for probabilities, %
we can directly compute their coverage probabilities for binomial distributions (see \iftoggle{arxiv}{\Cref{app:addCoverage}}{\cite[App.~C]{BHMWW24-arxiv}}). %
To give a visual comparison of coverage probabilities highlighting the concern for soundness, we fix confidence level $\gamma = 0.9$ and calculate the coverage probabilities for various methods in \Cref{fig:CoverageProbabilities}.
The top row shows $p_{cov}(50,p)$ as achieved by the unsound Wald, the sound-but-inefficient Okamoto, and the sound-and-recommended Clopper-Pearson interval.
Indeed, the Wald interval does not attain coverage $\geq 0.9$ for many values of $p$, while the others do.
Similarly, the bottom row concerns the truly sequential setting with $\varepsilon = 0.05$ and shows that the coverage probabilities for the unsound Chow-Robbins and sequential Clopper-Pearson methods are sometimes below~$\gamma$.
For Chen's method, we confirm its empirical but unproven soundness.
\end{example}

\begin{example}[Sample Efficiency]
We also observe that Okamoto significantly overshoots the desired confidence, which increases the number of samples it requires.
Indeed, $I_\mathit{\!CP}$ always needs fewer samples than $I_\mathit{Oka}$ experimentally~\cite[Sec.~3.3]{MWW24}.
For example, with $\gamma = 0.95$ and $\varepsilon = 0.01$,
we get a minimum $k$ of $18\,445$ for Okamoto, independent of $p$ or~$\hat p$.
For Clopper-Pearson, we get a worst-case $k$ of $9\,701$; for $k$ given and $p$ closer to $0$ or $1$, we would in turn get much smaller $\varepsilon$.
\end{example}

\paragraph{Hypothesis testing.}
All methods that produce sound confidence intervals $I = [l, u]$ can be turned into sound hypothesis tests for deciding whether $p \sim p_t$ for a threshold $p_t$ and $\sim\; \in \set{ \leq, \geq }$:
assuming $\sim$ is $\leq$, answer \textit{yes} if $u \leq p_t$, \textit{no} if $l \geq p_t$, and \textit{unknown} otherwise.
A dedicated and efficient method for hypothesis testing is the sequential probability ratio test (SPRT)~\cite{Wal45}.
It is \textcolor{OliveGreen}{\textbf{sound}} if we consider its \emph{indifference region} (an interval $[p - \varepsilon_i, p + \varepsilon_i]$ where the SPRT is allowed to give wrong answers) to fulfil the role of the $\varepsilon$ error in our PAC requirement.

\subsubsection{Our recommendation} is to implement the Clopper-Pearson interval for $P_{\diamond\,G}$ and $P_{\diamond\,G}^{\leq c}$ in the fixed and sequential settings as it is proven sound \emph{and} sample-efficient.
In the sequential setting, a two-step approach can be considered.
The Okamoto bound, employed by \emph{most} tools using a sound method (\Cref{sec:Implementation}), needs too many samples and produces overly conservative intervals:
it should not be used for estimating probabilities.
We highlight that our recommendations are independent of the underlying system dynamics and thus apply to SMC~in~general.

\section{Sound SMC for Expected Rewards}
\label{sec:ForRewards}

For unbounded expected rewards, %
each simulation run is a sample from an unknown \emph{path reward distribution} $\mu$ with outcomes in $[0, \infty)$. %
Given $k$ samples, we want a PAC guarantee for the expected value $E_{\diamond G}$.
In general, \emph{no} SM can guarantee coverage for unknown distributions with unbounded support.
Intuitively, $k$ gives the SM an indication of how likely it is to have missed some paths; with bounded support (e.g.\ for probabilities), this allows to quantify the uncertainty and thus~$\varepsilon$.
With unbounded support, outcomes with extremely low probability can dominate the expectation if they are even more extremely large.
We discuss the general case in \iftoggle{arxiv}{\Cref{app:bounds_impossibility}}{\cite[App.~A]{BHMWW24-arxiv}} and here provide an illustrative example.

\begin{example}
	The DTMC in \Cref{fig:ExampleForUnsound} %
	has $E_{=\{t_1,t_2\}} = c$.
	If $k$ is significantly lower than $n$, however, the SM likely only sees paths to $t_2$ and---if it is not sound---returns a confidence interval with an upper bound far below $c$.
\end{example}
\begin{figure}[t]
\input{figures/dtmc-example-tr}
\end{figure}
When $\support{\mu} \subseteq [a, b]$, a number of sound SMs exists, which we survey below. %
Then, in \Cref{sec:BoundingRewards}, we avoid the general impossibility of the unbounded case in two ways.
First, we exploit structural information about the DTMC to reduce to the bounded case.
Second, we introduce a novel perspective by considering a new notion of statistically converging lower bounds.

\subsection{Statistical Methods for Bounded Distributions}
\label{sec:StatisticsAB}

The textbook confidence interval for the mean of an unknown distribution is
the \emph{normal interval}:
$I_\mathit{Norm} = \hat X \pm z_\delta \hat \sigma / \sqrt{k}$,
where $\hat \sigma$ is the empirical standard deviation and $z_\delta$ the $(1 - \frac{\delta}{2})$-quantile of the standard normal distribution.
(Obtaining $z_\delta$ via the Student's-$t$ distribution with $k - 1$ degrees of freedom instead may work better for smaller~$k$.)
While asymptotic statement of Chow and Robbins about the normal interval also holds in the general, non-binomial setting, these methods are \textcolor{BrickRed}{\bf not sound}.
For example, on the DTMC of \Cref{fig:ExampleForUnsound} with $n=1000$, $c=1$ and $k = 500$, we experimentally found $I_\mathit{Norm}$ for $\gamma = 0.95$ to have a coverage probability of only~$\approx 0.39 \ll 0.95$.
Knowing bounds $[a, b]$ on the distribution's support, however, the \textcolor{OliveGreen}{\bf sound} methods we list below in this section are available.

\paragraph{Hoeffding's inequality.}
The Okamoto bound of \Cref{sec:ForProbabilitiesSound} is a special case of Hoeffding's inequality, which actually bounds the sum of independent (not necessarily i.d.) bounded random variables~\cite{Hoe63}.
It states that
\begin{equation*}
\mathbb{P}(\hat X - \mathbb{E}(X) \geq \varepsilon)
\leq
\mathrm{e}^{-\frac{2k\varepsilon^2}{(b - a)^2}}
\end{equation*}
and accordingly
$
\varepsilon \geq (b - a)\sqrt{(\ln \sfrac{2}{\delta}) / 2k}
$
for the two-sided case by distributing $\delta$ equally.
Note that Chernoff bounds~\cite{Che52} can be used to derive this inequality.

\paragraph{Bennett's and Bernstein's inequalities.}
Bennett's inequality can provide tighter bounds on a sum of random variables than Hoeffding's by %
taking the variance $\sigma^2$ into account~\cite{Ben62}.
However, not knowing the distribution, we do not know $\sigma^2$ either.
We could insert bounds%
, a simple one being $\sigma^2 \leq \frac14 (b-a)^2$.
Then, however, Bennet's inequality is strictly worse than Hoeffding's~\cite[App.~B]{MWW24}.
Bernstein's inequality~\cite{Ber24,Ber34} is a relaxation of Bennet's that is easier to compute, but yields even wider intervals.
Thus, in our setting, Hoeffding's inequality is preferable. %

\paragraph{Dvoretzky-Kiefer-Wolfowitz(-Massart) inequality (DKW).}
The DKW~\cite{DKW56,Mas90} relates the cumulative distribution function (cdf)
$F(x)=\mathbb{P}(X \leq x)$
of the unknown distribution $\mu$
to the empirical cdf
$\hat{F}(x)=\frac{1}{k} \sum_{i=1}^k\mathbbm{1}_{X_i\leq x}$ as follows:
\[
\mathbb{P}\Bigl({\textstyle\sup_{x\in\mathbb R}}\,|\hat{F}(x) - F(x)| > \varepsilon \Bigr) \le 2e^{-2k\varepsilon^2}.
\]
DKW is about thresholds, \ie in our setting the probability of exceeding a certain reward.
It characterizes a confidence band in which the real cdf lies with high probability.
This can be used to derive bounds on the expected value by computing the expected values of the best- and worst-case cdfs within the confidence band.
Formally, %
let $C$ be a confidence band containing an uncountable set of cdfs;
then with probability at least $1-2e^{-2k\varepsilon^2}$ we have
\[
\min_{\underline{F}\in C}\mathbb{E}(Y\mid Y\sim \underline{F}) \leq \mathbb{E}(X) \leq \max_{\overline{F}\in C}\mathbb{E}(Y\mid Y\sim \overline{F}).
\]
The cdfs minimising or maximising the expectation can be easily computed, as they are the upper and lower bound of the confidence band, respectively:
\begin{align*}
	\underline{F}(x)=\min\left\{ 1, \hat{F}(x) + \sqrt{\tfrac{1}{2k} \ln \tfrac{2}{\delta}} \right\} \quad \quad \overline{F}(x)=\max\left\{ 0, \hat{F}(x) - \sqrt{\tfrac{1}{2k}\ln \tfrac{2}{\delta}} \right\}
\end{align*}
\begin{wrapfigure}[10]{r}{0.33\linewidth}
	\vspace{-0.5cm}
	\includegraphics[width=\linewidth]{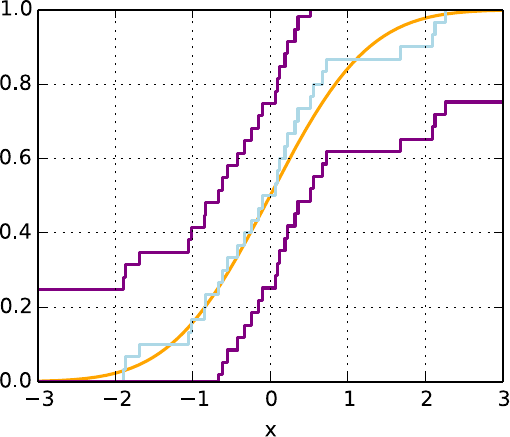}
	\vspace{-0.75cm}
	\caption{The DKW cdfs}
	\label{fig:DKW}
\end{wrapfigure}
\Cref{fig:DKW} illustrates the DKW\footnote{\Cref{fig:DKW} is based on file \href{https://commons.wikimedia.org/wiki/File:DKW_bounds.svg}{commons.wikimedia.org/wiki/File:DKW\_bounds.svg} (CC0).}
for $[a, b] = [-3, 3]$, with $F$ the smooth orange line, $\hat{F}$ the light blue step function in the center, and the outer purple step functions being $\underline{F}$ (to the left, with a higher probability for smaller outcomes) and $\overline{F}$ (to the right).
All steps of $\underline{F}$ are the same as those of $\hat{F}$ except that we \enquote{map} the largest $\sqrt{(\ln \sfrac{2}{\delta})/{2k}}$ fraction of steps to the lower bound (\ie at $x=a$). %
Similarly, $\overline{F}$ shifts probability mass into the upper bound~$b$.
In the worst case, the expectations of $\underline{F}$ or $\overline{F}$ coincide with Hoeffding, but provide a tighter confidence interval when few samples have an extremal value of $a$ or $b$.
Applying DKW is therefore especially advantageous when one of the a priori bounds $a$ and $b$ is very loose and samples are far above/below it.
The best case is obtained when all samples coincide, where the width of confidence interval halves when using DKW as opposed to Hoeffding.
In any case, the DKW interval is always contained in the Hoeffding interval.

\begin{proposition}\label{lem:dkw-v-hoeffding}
For given confidence level $\gamma$ and set of samples from a distribution with bounds $[a,b]$, let $[l_d,u_d]$ and $[l_h,u_h]$ be confidence intervals given by DKW and Hoeffing's inequality, respectively. 
Then $l_h\leq l_d < u_d \leq u_h$ and~$\frac{u_h-l_h}{u_d-l_d} \leq 2$.
\end{proposition}

\noindent DKW is also used by Phan \etal\cite{PTL21,LT19} with a view towards machine learning applications, who attribute its application to expected rewards to Anderson~\cite{And69}.

\paragraph{Sequential setting and hypothesis testing.}
Hoeffding's inequality applies in the sequential setting in the same way as the Okamoto bound.
For DKW, we could precompute $k$ based on the worst case, but this coincides with Hoeffding.
As mentioned, the Chow-Robbins scheme remains applicable and unsound.
The \textsc{AdaSelect} algorithm we mentioned in \Cref{sec:ForProbabilitiesDiscussion} also works soundly for bounded distributions when a relative error is desired.
For hypothesis testing, the SPRT in principle also applies to bounded distributions, and will in general perform better than the DKW for testing a single threshold.
The latter's advantage is that it provides an entire confidence \emph{band} around the cdf and thereby allows deriving the expected reward as well as probability bounds on \emph{all} reward thresholds \emph{at once}.
In this way, the DKW can also be used to tackle quantile problems.

\subsubsection{Our recommendation}\!%
for estimating a bounded distribution
is to use DKW in the fixed-$k$ setting and resort to Hoeffding's inequality when given $\varepsilon$. %

\subsection{Bounding Expected Rewards}
\label{sec:BoundingRewards}

For a full sound SMC procedure for expected rewards, it remains to find the bounds $[a, b]$ on the path rewards.
As rewards are non-negative, $a = 0$ is a safe lower bound (though larger $a$ may give lower $\varepsilon$ or $k$), leaving $b$ to be determined.

\subsubsection{Instantaneous and cumulative rewards.}
\label{sec:BoundingInstCumulRewards}
We know $\overline{r}_\mathit{\!max}$, an upper bound on the maximum state reward (\Cref{sec:Background}).
For instantaneous reward properties $E_{=c}$ and $E_{=G}$, a path's reward is at most $\max_{s\in S}R(s)$, making $b = \overline{r}_\mathit{\!max}$ the tightest safe upper bound we can give.
For step-bounded reward properties $E_{\leq c}$ and $E_{\diamond\,G}^{\leq c}$, we can upper-bound the reward of a path $\pi = s_0 \ldots$ by  $b = (c + 1) \cdot \overline{r}_\mathit{\!max} \geq \sum_{i=0}^{c}R(s_i)$.
Exploiting specific structures in higher-level languages may yield tighter bounds.

\subsubsection{Reachability rewards.}
For unbounded reachability rewards $E_{\diamond\,G}$, we need to bound the accumulated path reward until visiting a state in $G$, i.e.\ $\totalreward(\pi) = \sum_{i=0}^{\mathit{idx}(\pi, G)}R(\pi[i])$.
We assume $\mathbb{E}(\totalreward)$ to be finite,
but $\totalreward$ can still be unbounded:

\begin{example}\label{ex:unboundedPathReward}
	Consider the very simple DTMC in \Cref{fig:ExampleForUnbounded}.
	We have $E_{\diamond\,G} = \sum_{i=1}^\infty i \cdot (\sfrac{1}{2})^i = 2$, but the reward of a single path is unbounded, since every reward $v\in\mathbb{N}$ is obtained with positive probability $(\sfrac{1}{2})^v$.
\end{example}
This is not a degenerate case, but occurs whenever there exists a cycle with non-zero rewards.
Consequently, we cannot directly apply the SMs from \Cref{sec:StatisticsAB}.
Nevertheless, we can give meaningful estimations, by requiring additional knowledge or by relaxing the constraints on the result.

\paragraph{Bounding large values.}
As we cannot bound $\totalreward$, we aim to bound the effect that large values have on $\mathbb{E}(\totalreward)$.
Let us work in a general setting, as follows:

\begin{definition}
\label{def:BoundingSet}
Let $(\Omega, \mathcal{F}, \mathbb{P})$ be a probability space, $X\colon \Omega \to \RRpos$ a random variable with finite expectation, and $B \in \mathcal{F}$.
We call $B$ an $\varepsilon$-\emph{bounding set} if (i)~$\int_{\overline{B}} X \,\mathrm{d}\mathbb{P} < \varepsilon$ (with $\overline{B} = \Omega \setminus B$), and (ii)~$\forall\,\omega \in B\colon X(\omega) \in [a, b]$.
\end{definition}
(Concretely, $X$ could be $\totalreward$.)
If $B$ is a bounding set, then we can rewrite
\begin{equation*}
	\mathbb{E}(X) \,= \int_{\Omega} \!\!X \,\mathrm{d}\mathbb{P} \,= \int_{B} \!\!X \,\mathrm{d}\mathbb{P} +\!\! \int_{\overline{B}} \!\!X \,\mathrm{d}\mathbb{P} \,< \int_{B} \!\!X \,\mathrm{d}\mathbb{P} + \varepsilon.
\end{equation*}
Observe that $\int_{B} X \,\mathrm{d}\mathbb{P} = \mathbb{E}(X^B)$, where $X^B(\omega) = X(\omega)\text{ if }X \in B\text{ else }0$.
Since we chose $B$ such that $X(\omega) \in [a, b]$ for all $\omega \in B$, $X^B$ clearly is bounded and we can apply our previous methods to obtain a statistical estimate of $X^B$.
Inserting in the above equation then yields a bound on the overall expectation of $X$.

One possible choice for $B$ would be $\set{\omega \mid |X(\omega)| < t}$ for a sufficiently large $t$.
Such a $t$ exists for any random variable with finite expectation by positivity and additivity of $\mathbb{P}$ (when $\mathbb{E}(X) = \int_\Omega X \,\mathrm{d}\mathbb{P} < \infty$ we necessarily have $\lim_{t \to \infty} \int_{\{|X| > t\}} X \,\mathrm{d}\mathbb{P} = 0$).
However, %
without further assumptions we cannot derive such a $t$ or any other kind of bounding set just by sampling.
Thus, in the following we exploit that DTMCs give some structure to the random variable.

\paragraph{Geometric path lengths.}
While the value of $\totalreward$ in \Cref{ex:unboundedPathReward} can be arbitrarily large, as long as the expectation $\mathbb{E}(\totalreward)$ is finite, this only happens with vanishingly low probabilities.
This is the case for DTMCs (and many other Markov systems, see \Cref{rem:mixing-time}) in general:
Intuitively, the number of steps until $G$ is reached is (roughly) geometrically distributed.
\begin{lemma}\label{lem:bounding-set}
	Let $\varepsilon > 0$. %
	Choose $q$ such that
	$$
	\overline{|S|}\cdot \overline{r}_\mathit{\!max} \cdot (1 - (\underline{p}_\mathit{min})^{\overline{|S|}})^q \cdot (q - q (\underline{p}_\mathit{min})^{\overline{|S|}}+ 1) \cdot (\underline{p}_\mathit{min})^{-\overline{|S|}} \,<\, \varepsilon.
	$$
	Then $B = \lozenge^{\leq q \cdot \overline{|S|}} G = \set{ \pi \mid \mathit{idx}(\pi, G) \leq q \cdot \overline{|S|} }$ is a bounding set.
\end{lemma}
\begin{proof}[Sketch]
	Every state $s$ has a path of length at most $|S|$ to the goal $G$ by assumption.
	Such a path has probability at least $(p_{\min})^{|S|}$ and reward at most $|S|\cdot r_{\max}$.
	Considering $|S|$ steps as an \enquote{episode}, we can geometrically lower bound the probability to reach $G$ after $q$ episodes, and use this to upper bound the reward.
	See \iftoggle{arxiv}{\Cref{app:bounding-set}}{\cite[App.~B.1]{BHMWW24-arxiv}} for the full proof.
\end{proof}

\begin{corollary} \label{stm:structural_total_reward_bound}
	Given bounds $\overline{|S|}\geq |S|$, $0 < \underline{p}_\mathit{min} \leq p_{\min}$, and $\overline{r}_\mathit{\!max} \geq r_{\max}$, we can give PAC guarantees on $E_{\diamond\,G}$.
\end{corollary}
\begin{remark}\label{rem:mixing-time}
	Due to the worst-case over-approximation involved, $q$ is extremely large even for very small DTMCs and is thus not a practical solution.
	For example, a DTMC with $5$ states, $p_\mathit{min}=0.05$, and $r_\mathit{max}=1$ with a desired bound $\varepsilon=1$ requires $q>10^8$.
	The value $(p_{\min})^{-|S|}$ is closely related to the \emph{mixing time} of a DTMC, see e.g.~\cite{LP17}, which, in our setting, intuitively upper bounds the time until a goal state is reached with high probability.	
	While often a coarse bound, there exists DTMCs for which it is tight~\cite[Fig. 3]{HM18}.
	Determining better bounds on the mixing time (and thus $q$) requires knowledge of the DTMC's state space and transitions, which SMC explicitly does not have access to.
	For Markovian systems other than DTMCs, the geometric path lengths construction can also be used, provided we can obtain similar bounds; we conjecture that a sufficient condition is that the system is finite and goal states are reached almost surely.
\end{remark}

\paragraph*{Lower bounds.}
The inability to practically bound $b$ means that we cannot gain confidence in an \emph{upper} bound on $E_{\diamond\,G}$.
However, since rewards are non-negative, we cannot miss any extreme \enquote{negative} events.
Thus we at least want to derive meaningful \emph{lower bounds}.
Since $0$ trivially is a correct lower bound, %
we need a definition of bounds being \enquote{close} to the true value.
We propose the novel definition of \emph{limit-PAC} lower bounds:
they are not only (i)~sound, but additionally require that (ii)~given enough samples, they (unknowingly) become $\varepsilon$-close.

\begin{definition}
	Let $X$ be a random variable.
	A procedure $\mathcal{A}$ yields \emph{limit-PAC lower bounds} on $\mathbb{E}(X)$ if, for any confidence $\gamma$, the following two conditions hold:
	(i)~For a collection of independent samples $\Xi$ drawn from $X$, we have $\mathbb{P}(\mathcal{A}(\Xi, \gamma) \leq \mathbb{E}(X)) \geq \gamma$.
	(ii)~For any precision $\varepsilon > 0$, there exists a threshold $k_0$ such that for a collection of independent samples $\Xi$ drawn from $X$ with $|\Xi| \geq k_0$, we have $\mathbb{P}(\mathbb{E}(X) - \varepsilon \leq \mathcal{A}(\Xi, \gamma) \leq \mathbb{E}(X)) \geq \gamma$. %
\end{definition}

\begin{remark}
Classical procedures such as normal intervals do not provide limit-PAC bounds.
While they may satisfy condition (ii) and provide enough coverage to satisfy condition (i) in the limit, they can be unsound for many sample sets $\Xi$ that are not \enquote{sufficiently close to the limit.}
\end{remark}
We describe a procedure \enquote{DKW-$\mathbb{E}$-Lower} which provides limit-PAC lower bounds:
For a given set of samples $\Xi$ with $k = |\Xi|$, set $\chi_k = \sqrt{(\ln\sfrac{2}{1-\gamma})/{2k}}$ (using $\chi$ instead of $\varepsilon$ to avoid a clash of notation) and compute the empirical average over~$\Xi$, however setting the largest $\chi_k$ fraction of samples to $0$. %
This is equivalent to computing the expectation of the minimising cdf $\underline{F}(x)$ provided by the DKW (with width $\chi_k$), as explained in \Cref{sec:StatisticsAB}.

\begin{theorem}
\label{thm:DkwLowerPAC}
For any non-negative, finite-expectation random variable $X$, DKW-$\mathbbmsl{E}$-Lower gives limit-PAC lower bounds on $\mathbb{E}(X)$.
\end{theorem}

\begin{proof}[sketch, full proof in \iftoggle{arxiv}{\Cref{app:DkwLowerPAC}}{\cite[App.~B.2]{BHMWW24-arxiv}}]
Condition (i) holds by the DKW with coverage $\geq\gamma$ due to our choice of $\chi_k$.
For condition (ii), note that for every $\varepsilon/2$, we can find some bounding set $\set{X> t}$.
Then for large enough $\Xi$, the difference between the actual expected value and the output of DKW-$\mathbb{E}$-Lower on $[0,t]$ can be bounded by $\varepsilon/2$ as in the bounded case (see \Cref{sec:StatisticsAB}), and on $[t,\infty)$ the difference is also bounded by $\varepsilon/2$ by definition of bounding sets.
\end{proof}

\begin{corollary}
	DKW-$\mathbbmsl{E}$-Lower gives limit-PAC lower bounds on~$E_{\diamond\,G}$.
\end{corollary}
Note that \Cref{thm:DkwLowerPAC} directly extends to any random variable with a known lower bound, i.e.\ $X \in [a, \infty)$, by considering $X' = X - a$.
Then $X' \geq 0$ and any limit-PAC estimation of $X'$ also yields one for $X$, as $\mathbb{E}(X) = \mathbb{E}(X') + a$.
Similarly, for $X \in (-\infty, a]$ we can give limit-PAC \emph{upper} bounds.

\section{Implementation}
\label{sec:Implementation}

\newcommand{\Normal}{Normal}
\newcommand{\Student}{Student's-$t$}
\newcommand{\SeqStudent}{seq.\,Student's-$t$}
\newcommand{\ClopperPearson}{\textbf{Clopper-Pearson}}
\newcommand{\Wilson}{Wilson/CC}
\newcommand{\WilsonNonCC}{Wilson w/o CC}
\newcommand{\SeqClopperPearson}{seq.\,Clopper-P.}
\newcommand{\Chen}{Chen}
\newcommand{\Okamoto}{\textbf{Okamoto}}
\newcommand{\Hoeffding}{\textbf{Hoeffding}}
\newcommand{\ChowRobbins}{Chow-Robbins}
\newcommand{\UNSOUND}{}%
\newcommand{\UnknownSound}{}%
\newcommand{\WTF}{\textsuperscript{\bfseries\color{red}?}}

\begin{table}[t]
	\centering
	\caption{Default statistical methods used in state-of-the-art SMC tools}
	\label{tab:tools}
	\input{table_tools}
\end{table}

\paragraph{State of the art.}
In \Cref{tab:tools}, we collect the results of an extensive survey of the SMs used by default in all current SMC tools we are aware of. 
It is based on the information available in their tool papers (column \textit{Ref.}); for those marked ``*'', we also tested a current version or consulted its documentation%
\footnote{We used
\tool{Fig} 1.3,
the previous version of \tool{modes} from the \toolset v3.1.265,
\tool{Prism} 4.8.1,
\tool{Uppaal Smc} 5.0.0 with its \href{https://docs.uppaal.org/language-reference/query-semantics/smc_queries/ci_estimation/}{online documentation} as of 2024-10-09,
and
the \tool{Plasma Lab} 1.4.4 \href{https://web.archive.org/web/20191101134206/http://plasma-lab.gforge.inria.fr/plasma_lab_doc/1.4.4/html/index.html}{documentation from the Web Archive} as of 2019-11-01.
}
for more accurate information.
The ``seq.'' prefix for a method indicates a Chow-Robbins-like procedure using an interval different from Wald's/the normal approximation-based one.
We highlight the provably sound methods in boldface.
Entries ``---'' indicate that the tool does not appear to support that setting. %

We see that, in the fixed setting for probabilities, 5 of 8 tools choose a sound method, although three of those use the inefficient Okamoto bound;
in the sequential setting, only 2 of 10 tools use a sound (but inefficient) method.
For expected rewards, \tool{Cosmos} and \tool{Plasma Lab} apply Hoeffding's inequality when there is an obvious upper bound $b$, with \tool{Cosmos} using information from its higher-level formalism (\eg a system's finite capacity bound when estimating the average number of clients) for this purpose.
In the general setting, in particular for their respective variants of reachability rewards, \tool{Cosmos} will build a normal interval instead while \tool{Plasma Lab} will return the estimate $\hat X$ only, without error bounds.
Overall, \emph{no tool} implements a sound \emph{and} efficient method for probabilities in the sequential setting, nor for rewards in the fixed setting; those tools that use Hoeffding for rewards only do so for very specific cases.

\paragraph{Sound SMC in \thetool.}
We have implemented the recommendations we make \wrt SMs for probabilities and the new methods we propose for soundly handling expected rewards in the newest version of the \thetool statistical model checker as shown in the last row of \Cref{tab:tools}.
In particular, \thetool uses the $k$-precomputation based on the Clopper-Pearson interval in the sequential setting for probabilities, and the DKW in the fixed setting for expected rewards, improving upon the state of the art in soundness and sample efficiency.
\thetool supports $P_{\diamond\,G}^{\leq c}$, $P_{\diamond\,G}$, $E_{\diamond\,G}^{\leq c}$, and $E_{\diamond\,G}$ properties.
For $E_{\diamond\,G}^{\leq c}$ properties, it computes the upper bound as $b = (c + 1) \cdot \overline{r}_\mathit{\!max}$.
For $E_{\diamond\,G}$, it uses our new DKW-$\mathbb{E}$-Lower method by default.
\thetool also implements the Wilson/CC, Wald/normal, and Student's-$t$ intervals, the Okamoto bound, Chen's methods, the Chow-Robbins approach, and the SPRT.
Via a command-line parameter, the user can provide a preference list of these methods; for each property being analysed, \thetool chooses the first in the list that can be applied to it.
By default, it prefers sound over unsound and then efficient over less efficient methods, resulting in the first choices as in \Cref{tab:tools}.

\section{Experimental Evaluation}
\label{sec:Experiments}
\label{sec:ExperimentsForRewards}
\vspace{-.5ex}

To evaluate SMs for probabilities, we can directly work with the binomial distribution as in \Cref{ex:SoundnessNumerically}.
With expected rewards, however, the shape of the (unknown) reward distribution matters.
We thus use our implementation in \thetool on models from the QVBS~\cite{HKPQR19} to evaluate the coverage probability, performance, and effectiveness of the methods we propose in \Cref{sec:ForRewards} in a realistic setting.
The code and scripts for reproduction are available online~\cite{BHMWW24-artifact}.

\paragraph{Experimental setup.}
We used \thetool version 3.1.273\footnote{\thetool 3.1.273 implements all methods as described in \Cref{sec:Implementation}, but uses Wilson/CC for probabilities by default. Version 3.1.281 defaults to Clopper-Pearson as in \Cref{tab:tools}.}. %
We chose all DTMC and Markov decision process (MDP) models from the QVBS that contain an expected-reward property (except those that just ask for an expected number of transitions), excluding only the artificial \textit{haddad-monmege} model plus \textit{bluetooth} and \textit{oscillators}, which \thetool cannot handle for technical reasons (the former having multiple initial states, which \thetool does not support\footnote{Supporting multiple initial states, while natural for a PCTL model checker like \tool{Prism}, would require an SMC tool to perform a separate analysis starting from each initial state and thus defeat the scalability of SMC.}, and the latter's syntax being too large to parse and compile\footnote{The \textit{oscillators} model explicitly encodes a large flat state space in the higher-level \tool{Prism} language's syntax, overwhelming \thetool' parser that assumes its input to \emph{compactly} encode a potentially large state space for SMC to sample runs from.}).
We turn the MDP into DTMC by applying the \tool{Prism} language's DTMC semantics, which resolves all nondeterministic choices uniformly at random.
The models are parametrised; we use up to four parameter valuations each, including
the smallest and largest ones included in QVBS. %
A triple $\tuple{\mathit{model}, \mathit{parameter\,values}, \mathit{property}}$ is a \emph{benchmark instance}.

We consider all $E_{\diamond\,G}^{\leq c}$ and $E_{\diamond\,G}$ properties included with the models, the only
$E_{\diamond\,G}^{\leq c}$ property being in the \textit{resource-gathering} model.
To be able to study the DKW and Hoeffding methods, we manually turn all $E_{\diamond\,G}$ into $E_{\diamond\,G}^{\leq c}$ by
experimentally determining a small $c$ that does not change the value of the property up to the third significant digit.%
\footnote{We obtain reference results from the QVBS, if available, or via SMC with $k = 5\cdot 10^6$. We determine a $c$ as follows: Use SMC with $k=10^6$ and some step bound $c$ to obtain a value. If the three most significant digits of this value are equal to the reference result, stop and report $c$. Otherwise, increase $c$ and repeat. All choices of $c$ greater than the output of this procedure allow to approximate the result with precision $10^{-3}$, in particular since in all QVBS models the target state is reached with probability 1.} 
This in essence constitutes a manually-derived bounding set except we truncate rewards instead of setting them to~$0$.
\iftoggle{arxiv}{\Cref{sec:Appendix:ExperimentalSetup}}{\cite[App.~D.1]{BHMWW24-arxiv}} lists the resulting 44 instances (including the values of $c$).

\paragraph{Coverage probabilities.}
While the normal interval and Chow-Robbins are unsound, it was not clear if this manifests on real models under typical $k$ and~$\varepsilon$.
To investigate this, we implemented an empirical coverage test inside \thetool:
Given a benchmark instance, a confidence $\gamma$, step bound $k$, and a number $m$, it executes SMC with fixed $k$ for $m$ times, each time computing a $\gamma$-confidence interval.
It counts $w$, the number of times the reference result was \emph{w}rong, \ie not in the computed confidence interval.
Thus, we obtain the empirical coverage probability as $p_\mathit{cov} = \frac w m$, and compute a \enquote{meta} confidence interval $[l_\mathit{cov}, u_\mathit{cov}]$ around it using Clopper-Pearson.

\begin{wraptable}[8]{r}{0.53\linewidth}
	\vspace{-0.9cm}
	\centering
	\caption{Coverage over 44 instances ($\gamma = 0.95$)}
	\label{tab:WrongCoverageProbs}
	\vspace{4pt}
	\setlength{\tabcolsep}{4pt}
	\scriptsize
	\begin{tabular}{lcccc}
		\toprule
		SM & $u_\mathit{cov} {<} \gamma$ & $p_\mathit{cov} {<} \gamma$ & $\min p_\mathit{cov}$ & $\varnothing\,p_\mathit{cov}$ \\
		\midrule
		\Normal            & 10 & 31 & $0.908$ & $0.946$ \\
		\Student           &  9 & 32 & $0.902$ & $0.947$ \\
		Hoeffding ($k$)    &  0 &  0 & $1\phantom{.000}$ & $1\phantom{.000}$ \\
		DKW                &  0 &  0 & $0.999$ & $1.000$ \\
		\midrule
		\ChowRobbins\!\!\! & 16 & 24 & $0.723$ & $0.937$ \\
		Hoeffding ($\varepsilon$) &  0 &  0 & $1\phantom{.000}$ & $1\phantom{.000}$ \\
		\bottomrule
	\end{tabular}
\end{wraptable}
In \Cref{tab:WrongCoverageProbs}, we report the result of using this empirical coverage test choosing $\gamma = 0.95$, $k=1000$, and $m=5000$.
We report the number of benchmark instances where the respective SM attained insufficient coverage ($p_\mathit{cov} < \gamma$) in a statistically significant way ($u_\mathit{cov} < \gamma$), as well as the minimum $\min\,p_\mathit{cov}$ and average $\varnothing\,p_\mathit{cov}$ of the 44 coverage probabilities.
Detailed results are in \iftoggle{arxiv}{\Cref{sec:Appendix:CoverageTests}}{\cite[App.~D.2]{BHMWW24-arxiv}}.
Hoeffding's inequality and the DKW produced only sound results as expected, although Hoeffding timed out (${>}\,10$ minutes) 39 times in the sequential setting.
The unsound methods produce incorrect results much more often than they claim, with the insufficiency even being statistically significant in almost a quarter of the benchmark instances.

\def\figSoundPerformanceScatterHEIGHT{13}
\begin{wrapfigure}[\figSoundPerformanceScatterHEIGHT]{r}{0.44\linewidth}
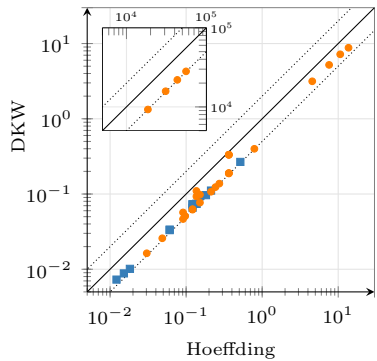

\vspace{1.325cm}
\include{experiments/plots/scatter-sound-performance.tex}
\vspace{{.725ex*\figSoundPerformanceScatterHEIGHT}}\vspace{-1ex}
\caption{Sound $\varepsilon$ given $k$}
\label{fig:SoundPerformanceScatter}
\end{wrapfigure}

\paragraph{Performance.}
We next evaluate the performance of the two sound SMs available when $[a, b]$ is known.
The runtime spent on the calculations involved with the SMs that we consider is negligible compared to that for generating sample paths.
We thus compare the performance of Hoeffding's inequality and the DKW via the half-width $\varepsilon$ of the interval returned given fixed $k = 500\,000$. %
The results are shown as a scatter plot in \Cref{fig:SoundPerformanceScatter},
where every point $\tuple{x, y}$ (blue for DTMCs, orange for MDPs) is the result of one benchmark instance, stating that using Hoeffding's inequality resulted in $\varepsilon = x$ while the DKW gave $\varepsilon = y$.
Note the logarithmic scale; points on the %
dotted diagonals mark $2\times\!$ differences.
We see that, as expected, the DKW consistently produces smaller intervals; %
the geometric mean of the ratios $\frac{\varepsilon\text{ for Hoeffding}}{\varepsilon\text{ for DKW}}$ over all 44 instances is~$1.72$, close to the theoretical maximum of $2$ (see \Cref{lem:dkw-v-hoeffding}).
Our upper bounds $b$ computed as per \Cref{sec:BoundingInstCumulRewards} are rather loose, benefiting DKW and resulting in very asymmetric DKW intervals with a lower bound close to the true value and an upper bound similar to Hoeffding's.

\begin{wrapfigure}[11]{r}{0.4\linewidth}
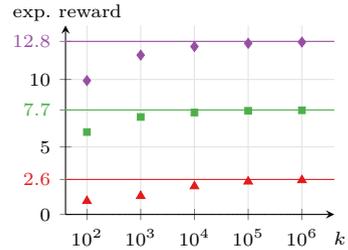

\vspace{0.6cm}
\include{experiments/plots/scatter-DKW-lower-bound.tex}
\vspace{0.21\linewidth}
\caption{DKW-$\mathbb{E}$-Lower}
\label{fig:DKWELowerExamples}
\end{wrapfigure}

\paragraph{Effectiveness.}
For unknown $b$ as in $E_{\diamond\,G}$ properties, we test how quickly our novel DKW-$\mathbb{E}$-Lower method converges to the (usually unknown) true value by applying it to our 44 benchmark instances for $k = 10^i$ with $i \in \set{2, 3, 4, 5, 6}$.
All results are in \iftoggle{arxiv}{\Cref{sec:Appendix:EffectivenessDKW}}{\cite[App.~D.3]{BHMWW24-arxiv}} and show behaviour similar to the three benchmark instances of \Cref{fig:DKWELowerExamples} (from top to bottom:
$\langle\textit{coupon},$ $\tuple{15,4,5},$ $\textit{exp\_draws}\,\rangle$,
$\langle\textit{resource-gathering},$ $\tuple{1300,100,100},$ $\allowbreak \textit{expgold}\,\rangle$, and
$\langle\textit{egl},$ $\tuple{5, 8},$ $\textit{messagesB}\,\rangle$%
).
The distance between the lower bound and the true value in most cases decreases by a factor between $2$ and $3$ on each step.
Over all instances and steps, on average (geometric mean) the distance decreases by a factor of $2.6$, which gives an indication of the practical convergence rate of the DKW-$\mathbb{E}$-Lower method.

\vspace{-.5ex}

\section{Conclusion}
\label{sec:Conclusion}
\vspace{-1.5ex}

We raise attention to the issue of soundness in SMC given a state of the art where many tools use unsound statistical methods.
For estimating probabilities, several sound methods exist, which have recently been compared in~\cite{MWW24}.
We summarised them as a reference for the SMC practitioner, and expanded upon~\cite{MWW24} by looking into the sequential setting as well as adding coverage probability plots that highlight the level of (un)soundness at a glance and providing an overview of the methods employed by tools.
For expected-reward properties, only two tools had (ad-hoc and inefficient) sound methods so far; we contribute a recommendation for the---apparently little-known---DKW and a thorough treatment of the problem of bounding the path reward distribution.
While our proof that sound SMC is possible for reachability rewards is currently of theoretical use only, we expect our notion of bounding sets to be crucial for future practical solutions based on the identification of specific structural features of a model's state space or higher-level description.
On the practical side, we formalised the notion of \emph{limit-PAC} procedures, which we instantiate by the DKW-$\mathbb{E}$-Lower method that we show to give close bounds in practice.
As immediate future work, our results can be extended to estimating rare event probabilities, where samples are in $[0, 1]$ or potentially unbounded depending on the rare event simulation method used.
Our contributions should transfer to continuous-time Markov chains straightforwardly.

\paragraph{Data availability statement.}
The models, tools, and scripts to reproduce our experimental evaluation are archived and available at DOI \href{https://doi.org/10.5281/zenodo.14602067}{10.5281/zenodo.14602067}~\cite{BHMWW24-artifact}.

\bibliographystyle{splncs04}
\bibliography{paper}

\iftoggle{arxiv}{
\clearpage
\appendix

\crefalias{section}{appendix} %
\crefalias{subsection}{appendix} %

\section{Impossibility of Confidence Intervals for Means of Unbounded Distributions}\label{app:bounds_impossibility}

	The impossibility of providing sound confidence intervals for an unbounded random variable is not an issue of our approach or of SMC, but rather a fundamental problem of statistics.
	Without additional assumptions, it is simply impossible to give meaningful upper bounds on the expectation of an unbounded (non-negative, finite support) random variable.
	For illustration, fix a discrete probability space $(\Omega, 2^{\Omega}, \mathbb{P})$ and suppose that $X$ is a random variable with an event $\omega \in \Omega$ where $\mathbb{P}[\omega] = \tau$ and $X(\omega) = v \frac{1}{\tau}$ for an arbitrary $v \in \mathbb{N}$.
	Then, $\mathbb{E}(X) = \tau \cdot v \frac{1}{\tau} + \sum_{\omega' \neq \omega} \mathbb{P}[\omega'] \cdot X(\omega') \geq v$.
	However, without additional knowledge, we cannot give a bound on $v$ without encountering $\omega$, nor can we prove by sampling that such an event does not exist.
	(In other words, given any number of samples, we can never be sure that we have seen all events possible under $\mathbb{P}$.)
	To conclude, observe that for any such random variable, we can directly construct a (finite state) Markov chain such that $\totalreward = X$.

\section{Proofs}

\subsection{\Cref{lem:dkw-v-hoeffding}}\label{app:dkw-v-hoeffding}

\begin{proof}
	Let $X_{(1)},X_{(2)}\dots X_{(k)}$ be the order statistics of the samples.
	Define $\varepsilon = (b - a)\sqrt{(\ln \sfrac{2}{\delta}) / 2k}$.
	Then, applying Hoeffding's inequality yields $[l_h,u_h] = \frac{1}{k} \sum_{i=1}^k X_{(i)} \pm \varepsilon$.
	For DKW, the confidence interval is constructed by replacing the largest (resp. smallest) $\varepsilon k$ samples with $a$ (resp. $b$) and averaging.
	Formally, to handle the case where $\varepsilon k$ is not a whole number, define $m=\lfloor \varepsilon k \rfloor$.
	Then
	\begin{align*}
	l_d =& \frac{1}{k} \left( \sum_{i=1}^{k-m-1} X_{(i)} + (m+1-\varepsilon k)X_{(k-m)} + \varepsilon k a \right) \\	
	u_d =& \frac{1}{k} \left( \varepsilon k b + (m+1-\varepsilon k)X_{(m+1)} + \sum_{i=m+2}^{k} X_{(i)} \right)
	\end{align*}
	We now show that $l_h \leq l_d$.
	\begin{align*}
	l_d =& \frac{1}{k} \left( \sum_{i=1}^{k-m-1} X_{(i)} + (m+1-\varepsilon k)X_{(k-m)} + \varepsilon k a \right) \\
	=& \frac{1}{k} \left( \sum_{i=1}^{k} X_{(i)} + (\varepsilon k-m)X_{(k-m)} + (\varepsilon k-m)a + ma - \sum_{i=k-m}^k X_{(i)} \right) \\	
	=& \frac{1}{k} \left( \sum_{i=1}^{k} X_{(i)} - (\varepsilon k-m)(X_{(k-m)}-a) - \sum_{i=k-m}^k (X_{(i)}-a) \right) \\	
	\geq & \frac{1}{k} \left( \sum_{i=1}^{k} X_{(i)} - (\varepsilon k-m)(b-a) - \sum_{i=k-m}^k (b-a) \right) \\
	=& \frac{1}{k} \left( \sum_{i=1}^{k} X_{(i)} - \varepsilon k(b-a) \right) \\
	=& \frac{1}{k} \left( \sum_{i=1}^{k} X_{(i)} \right) - \varepsilon (b-a) \\
	=& l_h
	\end{align*}
	Proving $u_d\leq u_h$ is analogous, and $l_d \leq u_d$ follows from definition.
	
	Finally, to show  $\frac{u_h-l_h}{u_d-l_d} \leq 2$, notice that $u_h-l_h=2\varepsilon$, regardless of any $X_{(i)}$.
	Thus, $\frac{u_h-l_h}{u_d-l_d}$ is maximal when $u_d-l_d$ is minimal.
	We have 
	\[u_d-l_d = \frac{1}{k} \left( \varepsilon k (b-a) + (m+1-\varepsilon k)(X_{(m+1)}-X_{(k-m)}) + \sum_{i=1}^{k-m-1} (X_{(k-i+1)} - X_{(i)}) \right) \]
	For $\varepsilon\leq\frac{1}{2}$, all differences of $X_{i}$ are non-negative.
	If $\varepsilon\leq\frac{1}{2}$, we can cancel opposite terms inside the sum such that only non-negative terms remain.
	Clearly then, each $X_{(k-i+1)} - X_{(i))}$ is minimal if it is $0$, i.e. when $X_{(k-i+1)} = X_{(i)}$.
	Since it is an order statistic, this implies $X_{j} = X_{(i))}$ for all $1\leq i,j \leq k$.
	
	Lastly, if $X_{(i)}=x$ for all $i$ then $l_d=x-\varepsilon (x-a)$ and $u_d=x-\varepsilon (b-x)$ showing that $u_d-l_d$ is always at least $\varepsilon(b-a)$, and thus $\frac{u_h-l_h}{u_d-l_d} \leq 2$.
\end{proof}

\subsection{\Cref{lem:bounding-set}}\label{app:bounding-set}

\begin{proof}
	From any state $s$ there exists a simple finite path $\pi_\mathit{fin}$ to the goal $G$ by assumption.
	It has non-zero probability $p_s$, and $|\pi_\mathit{fin}| \leq |S|$.
	We clearly have that $p_s \geq (p_{\min})^{|S|}$.
	Since the DTMC has finitely many states, there is a minimum path probability $\underline{p} = \min_{s \in S} p_s \geq (p_{\min})^{|S|}$.
	Consequently, at any time, there is a probability of at least $\underline{p}$ of reaching the goal within $|S|$ steps.
	Thus the probability of reaching the goal in at most $q \cdot |S|$ steps is at least $1 - (1 - \underline{p})^q$.
	Further, every such path obtains at most $q \cdot |S| \cdot r_{\max}$ reward with that probability.
	Hence, for any $q$ the total weight of all paths with length $\geq q \cdot |S|$ can be bounded by
	\begin{align*}
		{\sum}_{j = q}^\infty ((j + 1) \cdot |S| \cdot r_{\max}) \cdot (1 - \underline{p})^j \underline{p}
		&= |S| \cdot r_{\max} \cdot \underline{p} \cdot {\sum}_{j = q}^\infty (j+1) (1 - \underline{p})^j\\
		&= |S| \cdot r_{\max} \cdot (1 - \underline{p})^q \cdot (q \underline{p} + 1) \cdot \underline{p}^{-1}.
	\end{align*}
	Observe that the right-hand side goes to zero for $q \to \infty$, hence for any $\varepsilon > 0$ we can find a $q$ such that this expression is less than $\varepsilon$.
	Moreover, any path with length less than $q \cdot |S|$ steps can only obtain a reward of at most $q \cdot |S| \cdot r_{\max}$.
	Together, we get for $B = \lozenge^{\leq q \cdot |S|} T$ that $\int_{\overline{B}} \totalreward \:\mathrm{d}\mathbb{P} < \varepsilon$ for sufficiently large $q$, but also $0 \leq \totalreward(\pi) \leq q \cdot |S| \cdot r_{\max}$ for all $\pi \in B$.
	\qed
\end{proof}

\subsection{\Cref{thm:DkwLowerPAC}}\label{app:DkwLowerPAC}

\begin{proof}
	Fix the random variable $X$ and confidence level $\gamma$.
	We write $F(x) = \mathbb{P}(X \leq x)$ for the cdf of $X$ and, for a fixed set $\Xi$ of $k$ independent samples, $F_\Xi(x)$ for the empirical distribution function over $\Xi$, indicating the fraction of samples obtaining a reward $\leq x$, and $E_\Xi$ the expectation of $F_\Xi$.
	Recall that the function $F_\Xi$ is a product of a random process (drawing the samples $\Xi$ from $X$), so, for example, $E_\Xi$ is a random variable depending on $\Xi$. %
	For a random set of samples $\Xi$ with cardinality $k = |\Xi|$ we have by the DKW-inequality that
	\begin{equation}\label{eq:dkw}
		|F(x) - F_\Xi(x)| \leq \chi_k
	\end{equation} 
	with confidence $1 - 2 \mathrm{e}^{-2 k \chi_k^2}=\gamma$.
	We define $F'_\Xi(x) \defeq \min\,\set{F_\Xi(x) + \chi_k, 1}$ (noting the relation to the Kolmogorov-Smirnov confidence band) and $E_\Xi'$ as the expectation of $F_\Xi'$ (\ie the empirical average where the largest $\chi_k$ fraction of samples is mapped to $0$).
	
	\paragraph{Proof of (i).} 
	From \Cref{eq:dkw} it follows that $F(x) \leq F_\Xi(x) + \chi_k$, \ie $F$ stochastically dominates $F'_\Xi(x)$.
	In particular, this implies the desired property $E_\Xi' \leq \mathbb{E}(X)$ holds with high confidence.
	
	\paragraph{Proof of (ii).} Additionally fix a precision requirement $\varepsilon$.
	Recall that, by the DKW-inequality, for a sufficiently large number of samples the \emph{probability} inaccuracy (difference between $F_\Xi'$ and $F$) is arbitrarily small with high confidence.
	However, it remains to bound the error in \emph{expectation}, i.e.\ the difference between $E_\Xi'$ and $\mathbb{E}(X)$.
	To this end, we show that for a random set of samples $\Xi$ with $|\Xi| \geq k_0$ (where the concrete $k_0$ is given below) we obtain the desired precision whenever \Cref{eq:dkw} holds, i.e.\ overall $\mathbb{E}(X) - \varepsilon \leq E'_\Xi \leq \mathbb{E}(X)$ with confidence $\gamma$.

	To fix $k_0$, observe that $X$ has finite expectation, thus there exists $t$ such that $\set{X \leq t}$ is a bounding set for $\sfrac{\varepsilon}{2}$; fix such a $t$.
	(Observe that $\set{X \leq t}$ might be the whole set of outcomes $\Omega$ when $X$ is bounded.)
	As $\lim_{k\to\infty}\chi_k = 0$ for, there is $k_0$ such that for all $k\geq k_0$ we have $2 \cdot t \cdot \chi_k \leq \sfrac{\varepsilon}{2}$.

	Now, fix a set of samples $\Xi$ with $k = |\Xi| \geq k_0$ and assume that \Cref{eq:dkw} holds.
	On the one hand, we get from the above arguments that $E_\Xi'$ is a lower bound, as required.
	On the other hand, since $X$ is a non-negative random variable, we have that $\mathbb{E}(X) = \int_{0}^\infty 1 - F(x) \, dx$.
	The same holds for the empirical estimate, i.e.\ $E_\Xi = \int_0^\infty 1 - F_\Xi(x) \, dx$.
	Together, we get
	\begin{equation*}
		\mathbb{E}(X) - E_\Xi' = \int_0^\infty 1 - F(x) - (1 - F'_\Xi(x)) \,dx = \int_0^\infty F_\Xi'(x) - F(x) \,dx
	\end{equation*}
	Splitting the above integral, we get
	\begin{equation*}
		\mathbb{E}(X) - E_\Xi' = \int_{[0, t]} F_\Xi'(x) - F(x) \, dx + \int_{(t, \infty)} F_\Xi'(x) - F(x) \, dx.
	\end{equation*}
	For the first part, observe that $0 \leq \int_0^t F'_\Xi(x) - F_\Xi(x) \, dx \leq t \cdot \chi_k$ by definition of $F_\Xi'$.
	Moreover, by \Cref{eq:dkw}, we similarly get
	\begin{equation*}
		\left| \int_0^t F_\Xi(x) - F(x) \, dx \right| \leq \int_0^t |F_\Xi(x) - F(x)| \, dx \leq \int_0^t \chi_k \,dx \leq t \cdot \chi_k.
	\end{equation*}
	Thus,
	\begin{equation*}
		\int_0^t F_\Xi'(x) - F(x) \, dx = \int_0^t F_\Xi'(x) - F_\Xi(x) \, dx + \int_0^t F_\Xi(x) - F(x) \, dx \leq 2 \cdot t \cdot \chi_k \leq \frac{\varepsilon}{2}.
	\end{equation*}
	For the second part, observe that since $F'_\Xi \leq 1$ we have
	\begin{equation*}
		\int_{(t, \infty)} F_\Xi'(x) - F(x) \,dx = \int_{(t, \infty)} 1 - F(x) \, dx - \int_{(t, \infty)} 1 - F_\Xi'(x) \, dx \leq  \int_{(t, \infty)} 1 - F(x) \, dx.
	\end{equation*}
	By choice of $t$ and definition of bounding sets we have
	\begin{equation*}
		\int_{\set{X > t}} X \, d\mathbb{P} = \int_{(t, \infty)} 1 - F(x) \, dx \leq \frac{\varepsilon}{2},
	\end{equation*}
	and thus $\int_{(t, \infty)} F_\Xi'(x) - F(x) \,dx \leq \sfrac{\varepsilon}{2}$.
	All together, we have
	\begin{equation*}
		\mathbb{E}(X) - E'_\Xi = \int_0^t F_\Xi'(x) - F(x) \, dx + \int_t^\infty F_\Xi'(x) - F(x) \, dx \leq \frac{\varepsilon}{2} + \frac{\varepsilon}{2} = \varepsilon.
	\end{equation*}
	Thus, whenever \cref{eq:dkw} holds for a set of samples large enough, our estimate is a lower bound and close to the correct value. \qed
\end{proof}

\section{Additional Discussion on Coverage Probability}\label{app:addCoverage}

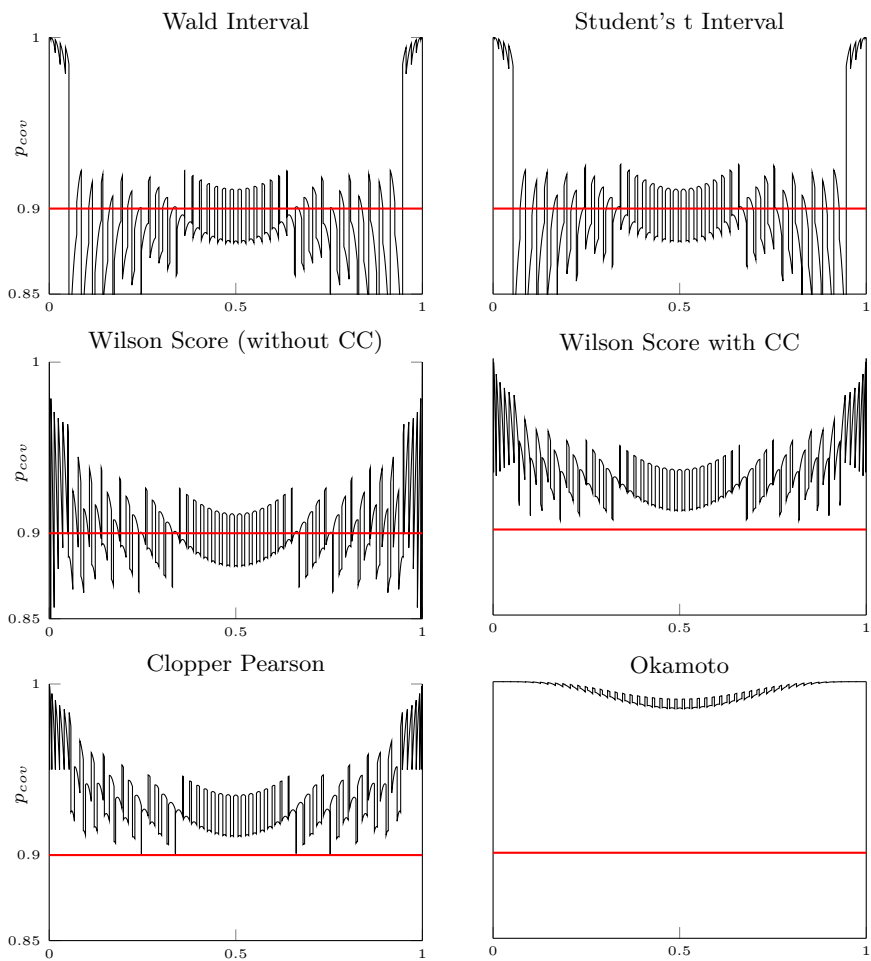
\begin{figure}[tp]
\centering
	\pgfplotsset{every axis/.append style={appendixplot}}
	\begin{minipage}{0.55\textwidth}
	\centering
	\begin{tikzpicture}
		\begin{axis}[confleftplot,xlabel=Wald Interval]
			\addplot [no marks,draw=black] table [x index=0, y index=1] {experiments/binom_coverage/coverage_fixed_n_stdci_normal.csv};
			\addplot [no marks,draw=red,thick] coordinates {(0,0.9) (1,0.9)};
		\end{axis}
	\end{tikzpicture}
	\end{minipage}%
	\begin{minipage}{0.45\textwidth}
	\centering
	\begin{tikzpicture}
		\begin{axis}[confplot,xlabel=Student's t Interval]
			\addplot [no marks,draw=black] table [x index=0, y index=1] {experiments/binom_coverage/coverage_fixed_n_stdci_student.csv};
			\addplot [no marks,draw=red,thick] coordinates {(0,0.9) (1,0.9)};
		\end{axis}
	\end{tikzpicture}
	\end{minipage}%

	\begin{minipage}{0.55\textwidth}
	\centering
	\begin{tikzpicture}
		\begin{axis}[confleftplot,xlabel=Wilson Score (without CC)]
			\addplot [no marks,draw=black] table [x index=0, y index=1] {experiments/binom_coverage/coverage_fixed_n_wilson.csv};
			\addplot [no marks,draw=red,thick] coordinates {(0,0.9) (1,0.9)};
		\end{axis}
	\end{tikzpicture}
	\end{minipage}%
	\begin{minipage}{0.45\textwidth}
	\centering
	\begin{tikzpicture}
		\begin{axis}[confplot,xlabel=Wilson Score with CC]
			\addplot [no marks,draw=black] table [x index=0, y index=1] {experiments/binom_coverage/coverage_fixed_n_wilson_cc.csv};
			\addplot [no marks,draw=red,thick] coordinates {(0,0.9) (1,0.9)};
		\end{axis}
	\end{tikzpicture}
	\end{minipage}%

	\begin{minipage}{0.55\textwidth}
	\centering
	\begin{tikzpicture}
		\begin{axis}[confleftplot,xlabel=Clopper Pearson]
			\addplot [no marks,draw=black] table [x index=0, y index=1] {experiments/binom_coverage/coverage_fixed_n_clopper_pearson.csv};
			\addplot [no marks,draw=red,thick] coordinates {(0,0.9) (1,0.9)};
		\end{axis}
	\end{tikzpicture}
	\end{minipage}%
	\begin{minipage}{0.45\textwidth}
	\centering
	\begin{tikzpicture}
		\begin{axis}[confplot,xlabel=Okamoto]
			\addplot [no marks,draw=black] table [x index=0, y index=1] {experiments/binom_coverage/coverage_fixed_n_okamoto.csv};
			\addplot [no marks,draw=red,thick] coordinates {(0,0.9) (1,0.9)};
		\end{axis}
	\end{tikzpicture}
	\end{minipage}%
\caption{Coverage probabilities over $p$ for $k=50$ and $\gamma=0.9$ for different methods.}
\label{fig:FixedNCoverageProbabilities}
\end{figure}

\begin{figure}[tp]
\centering
	\pgfplotsset{
		cycle list/Dark2,
		cycle multiindex* list={
			linestyles*\nextlist
			Dark2\nextlist
		},
		every axis plot/.append style={line width=1.25pt}
	}
	\begin{tikzpicture}
		\begin{axis}[allplot,xlabel=$p$,ylabel=expected interval width,width=7cm,height=6cm,
				xmin=0,xmax=1,ymin=0,ymax=0.4,
				legend style={draw=none},
				legend cell align=left,
				legend pos=outer north east,
		]
			\addplot+ [no marks] table [x index=0, y index=1] {experiments/binom_coverage/expected_interval_width_clopper_pearson.csv};
			\addlegendentry{Clopper Pearson}
			\addplot+ [no marks] table [x index=0, y index=1] {experiments/binom_coverage/expected_interval_width_wilson_cc.csv};
			\addlegendentry{Wilson Score with CC}
			\addplot+ [no marks] table [x index=0, y index=1] {experiments/binom_coverage/expected_interval_width_okamoto.csv};
			\addlegendentry{Okamoto}
		\end{axis}
	\end{tikzpicture}
\caption{Expected interval width over $p$ for $k=50$ and $\gamma=0.9$ for sound methods.}
\label{fig:ExpectedIntervalLength}
\end{figure}
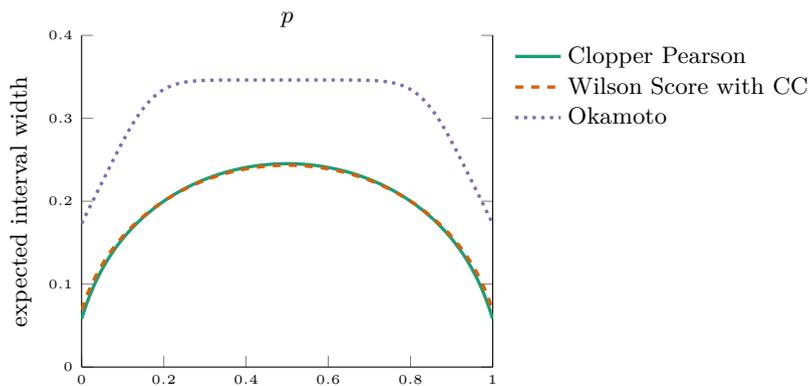

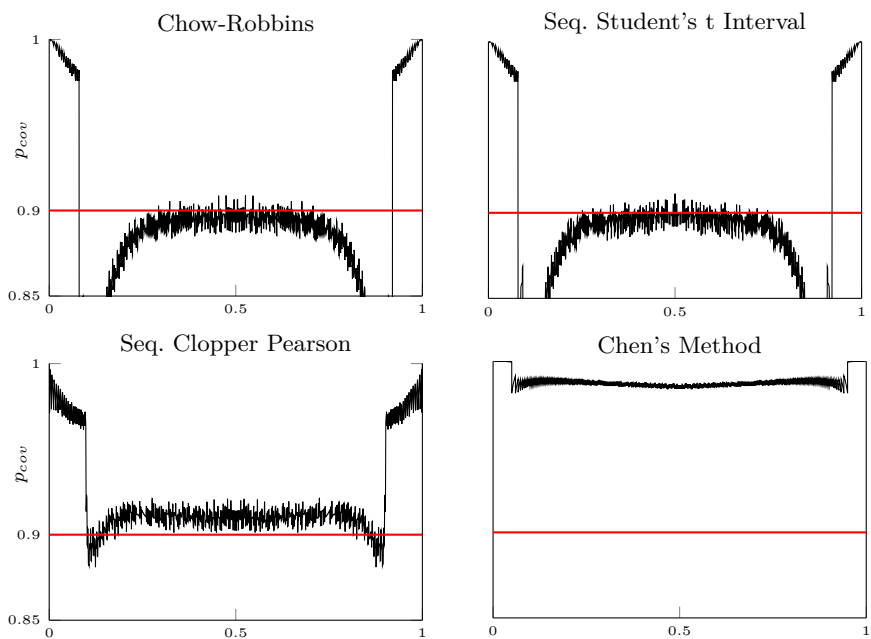
\begin{figure}[tp]
\centering
	\pgfplotsset{every axis/.append style={appendixplot}}
	\begin{minipage}{0.54\textwidth}
	\centering
	\begin{tikzpicture}
		\begin{axis}[confleftplot,xlabel=Chow-Robbins]
			\addplot [no marks,draw=black] table [x index=0, y index=1] {experiments/binom_coverage/coverage_fixed_eps_chow_robbins.csv};
			\addplot [no marks,draw=red,thick] coordinates {(0,0.9) (1,0.9)};
		\end{axis}
	\end{tikzpicture}
	\end{minipage}%
	\begin{minipage}{0.45\textwidth}
	\centering
	\begin{tikzpicture}
		\begin{axis}[confplot,xlabel=Seq.\ Student's t Interval]
			\addplot [no marks,draw=black] table [x index=0, y index=1] {experiments/binom_coverage/coverage_fixed_eps_sequential_student.csv};
			\addplot [no marks,draw=red,thick] coordinates {(0,0.9) (1,0.9)};
		\end{axis}
	\end{tikzpicture}
	\end{minipage}%

	\begin{minipage}{0.55\textwidth}
	\centering
	\begin{tikzpicture}
		\begin{axis}[confleftplot,xlabel=Seq.\ Clopper Pearson]
			\addplot [no marks,draw=black] table [x index=0, y index=1] {experiments/binom_coverage/coverage_fixed_eps_sequential_clopper_pearson.csv};
			\addplot [no marks,draw=red,thick] coordinates {(0,0.9) (1,0.9)};
		\end{axis}
	\end{tikzpicture}
	\end{minipage}%
	\begin{minipage}{0.45\textwidth}
	\centering
	\begin{tikzpicture}
		\begin{axis}[confplot,xlabel=Chen's Method]
			\addplot [no marks,draw=black] table [x index=0, y index=1] {experiments/binom_coverage/coverage_fixed_eps_chen.csv};
			\addplot [no marks,draw=red,thick] coordinates {(0,0.9) (1,0.9)};
		\end{axis}
	\end{tikzpicture}
	\end{minipage}%
\caption{Coverage probabilities over $p$ for $\varepsilon=0.05$ and $\gamma=0.9$ for sequential methods.}
\label{fig:FixedEpsCoverageProbabilities}
\end{figure}

\begin{figure}[tp]
\centering
	\begin{tikzpicture}
		\pgfplotsset{
			cycle list/Dark2,
			cycle multiindex* list={
				linestyles*\nextlist
				Dark2\nextlist
			},
			every axis plot/.append style={line width=1.25pt}
		}
		\begin{axis}[allplot,xlabel=$p$,ylabel=expected number of samples,width=9cm,height=6cm,
				xmin=0,xmax=1,ymin=0,ymax=650,
				legend style={draw=none},
				legend cell align=left,
				legend pos=outer north east
		]
			\addplot+ [no marks] table [x index=0, y index=1] {experiments/binom_coverage/expected_num_samples_chen.csv};
			\addlegendentry{Chen's Method}
			\addplot+ [no marks] table [x index=0, y index=1] {experiments/binom_coverage/expected_num_samples_chow_robbins.csv};
			\addlegendentry{Chow-Robbins}
			\addplot+ [no marks] table [x index=0, y index=1] {experiments/binom_coverage/expected_num_samples_sequential_clopper_pearson.csv};
			\addlegendentry{Seq.\ Clopper Pearson}
			\addplot+ [no marks] table [x index=0, y index=1] {experiments/binom_coverage/expected_num_samples_sequential_student.csv};
			\addlegendentry{Seq.\ Student's t Interval}
			\addplot+ [no marks] table [x index=0, y index=1] {experiments/binom_coverage/expected_num_samples_clopper_pearson.csv};
			\addlegendentry{Clopper Pearson}
			\addplot+ [no marks] table [x index=0, y index=1] {experiments/binom_coverage/expected_num_samples_wilsoncc.csv};
			\addlegendentry{Wilson Score with CC}
		\end{axis}
	\end{tikzpicture}
\caption{Expected number of samples over $p$ for $\varepsilon=0.05$ and $\gamma=0.9$ for sequential methods.}
\label{fig:ExpectedNumSamples}
\end{figure}
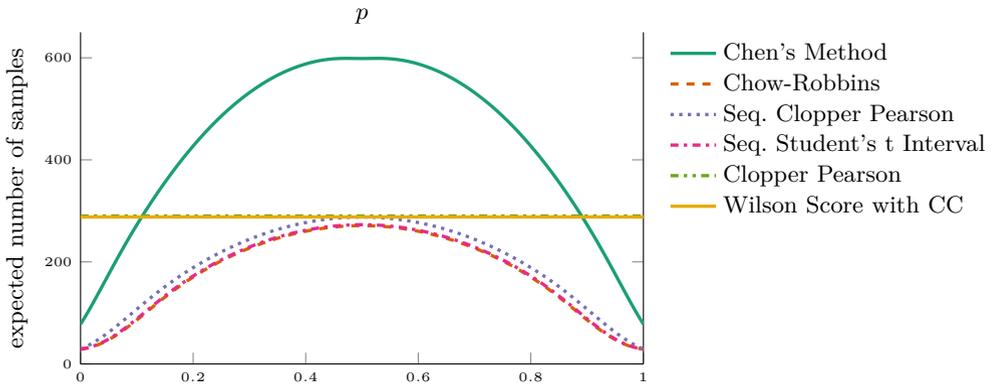

In this section we provide further details on \Cref{sec:ForProbabilities}.
We first outline how we obtained the exact coverage probability plots, before evaluating the soundness and efficiency of the statistical methods used in SMC tools (see \Cref{tab:tools}) by analysing their coverage probabilities.

\paragraph{Computation.}
Given some SM, a sample size $k$, confidence level $\gamma$, and a fixed success probability $p$,  we can compute the coverage probabilities by considering all $k+1$ possible outcomes of the sampled binomial distribution.
For each number of successes $k_s\in\{ 0,\dots,k \}$ we can easily compute the probability of observing exactly $k_s$ successes via the pmf of the binomial distribution.
Additionally, for each $k_s$ we can also compute the confidence interval returned and determine whether it contains the true $p$.
The coverage probability then is simply the sum of the probabilities of observing exactly $k_s$ successes for those $k_s$ for which the confidence interval contained the true $p$.

For the sequential algorithms, where a desired precision $\varepsilon$ is given instead of $k$, the core idea remains the same; however computation is slightly more involved:
Here, we need to compute the set of \emph{stopping points}, i.e.\ pairs $(k_s,k)$ where the SM would stop sampling after $k_s$ successes out of $k$ trials; and the respective probability to reach each stopping point.
For a formal description of the problem and details on this computation, see \cite{Fre10}.

Since the computation of the coverage probability can become rather computationally expensive for large $k$ and small $\varepsilon$, respectively, we chose moderate values of $k=50$ and $\varepsilon=0.05$ for all our examples in this section.
Additionally, note that for practical purposes, with larger $k$ and smaller $\varepsilon$ floating point errors will become an issue.

\paragraph{Fixed sample size.} We first focus on the setting with a fixed sample size $k$.
In \Cref{fig:FixedNCoverageProbabilities} we plot the coverage probability over $p$ for all methods implemented in some SMC tool (see \Cref{tab:tools}). 
As discussed, the Wald/normal invertal, Student's-$t$ interval and Wilson score are \textcolor{BrickRed}{\bf not sound} and have coverage probability less than $\gamma=0.9$ for some $p$.
While the Wilson score with CC is also \textcolor{BrickRed}{\bf not sound} in general, we observe that it is sound in our example.
This is in particular since we chose $\gamma=0.9$ and unsound behaviour is empirically only observed for larger $\gamma\approx0.95$ and above \cite{MWW24}.
On the other hand, we can see that the Okamoto bound is overly conservative, consistently achieving coverage probability over $0.98$ at the cost of increased sample complexity.
The difference between the Wilson/CC and Clopper-Pearson interval is only minimal.

For the two generally \textcolor{OliveGreen}{\bf sound} intervals, as well as Wilson score with CC which is sound for the chose parameters, we also plot the expected with of the confidence interval for different $p$ in \Cref{fig:ExpectedIntervalLength}.
Again, we see that both Wilson score with CC and Clopper-Pearson perform very similar while clearly outperforming Okamoto's bound.
Note that the computed confidence half-width $\varepsilon$ in Okamoto's bound does not directly depend on the sample success probability $\hat{p}$ and hence the width of $I_\mathit{Oka}$ is constant.
However as $\hat{p}\pm\varepsilon$ may be outside of the interval $[0,1]$ we can clamp the bounds of $I_\mathit{Oka}$ to $0$ and $1$, respectively, leading to confidence intervals smaller than $2\varepsilon$.

\paragraph{Sequential setting.} For the setting with a fixed precision $\varepsilon$, \Cref{fig:FixedEpsCoverageProbabilities} shows that all truly sequential algorithms except for Chen's method are \textcolor{BrickRed}{\bf not sound}, again attaining coverage probability $<\gamma=0.9$ for some $p$.
This is not surprising for the Chow-Robbins method and sequential Student's-$t$ interval since their underlying methods for fixed $k$ are unsound already (recall that the Chow-Robbins method uses the Wald/normal interval).
For the sequential Clopper-Pearson method we lose the soundness property since the sequential method is biased towards early termination \cite{JSD19}.

For Chen's method we do not observe unsound behaviour, matching the experimental soundness of Chen's evaluation \cite{Che15}.
In fact, we actually observe that the method is overly conservative.
This is not too surprising, since at its core it is based on a refinement of Okamoto's bound for $p$ away from $\frac{1}{2}$.
Nonetheless the soundness of Chen's method remains unproven.

In \Cref{fig:ExpectedNumSamples} we plot the expected number of samples collected by each of the four methods until they stop for different $p$.
This number was calculated as follows: There is a finite set of pairs $(k_s,k)$ of successes $k_s$ and tries $k$ where the algorithm stops. For every $p$, we can compute the probability of the algorithm reaching every such pair $(k_s,k)$.
Then, the expected number of samples is the weighted average of these $k$, with the weight being the sum of all the probabilities of the corresponding pairs $(k_s,k)$.
Additionally, we plot the a priori bound for the Wilson/CC and the Clopper-Pearson interval which is obtained by computing the largest number of samples $k$ for which each method outputs a confidence interval of size $\leq\varepsilon$, regardless of the number of successes $k_s$.
This is done by considering the worst-case scenario where $k_s = k/2$ \cite{MWW24}.
Hence this value is constant and not affected by the true probability $p$.
Note that for both methods yield very similar bounds (288 for Wilson/CC and 290 for Clopper-Pearson) which lie on top of each other in \Cref{fig:ExpectedNumSamples}.
While these a priori bound are naturally higher than the sample complexity of the unsound methods, they are significantly lower than the sample complexity of Chen's method for many values of $p$, with Chen's method only being favourable for $p$ close to $0$ or~$1$.

As these a priori bounds guarantee sound confidence intervals and require less samples than the only truly sequential method that is \emph{conjectured} to be sound, we recommend to use them and implement them in our tool.

\section{Additional Experimental Results}

In this section, we include additional plots and data for the experiments that we performed but cannot include in \Cref{sec:Experiments} in this conference version of our~paper.

\subsection{Experimental Setup}
\label{sec:Appendix:ExperimentalSetup}

\begin{table}[tp]
	\centering
	\caption{Models and properties used in experimentation}
	\label{tab:Models}
	\smaller
	\input{table_models}
\end{table}

We list the benchmark instances that we use in \Cref{tab:Models}.
The values for $\overline{|S|}$, $\overline{r}_\mathit{\!max}$, and $\underline{p}_\mathit{min}$ were determined by \thetool from the \jani syntax of the models, while the value of $|S|$ was determined by a PMC tool (\tool{mcsta} of the \toolset or \tool{Storm}~\cite{HJKQV22} using BDDs).
Note in particular that $\overline{|S|}$, based on model variable ranges, is close for some but far off for many instances.

\subsection{Coverage tests}
\label{sec:Appendix:CoverageTests}

\begingroup
\newlength{\scatterPlotsSize}
\setlength{\scatterPlotsSize}{.35\linewidth}
\def\largeplotMIN{0.01}
\def\largeplotMAX{2000}
\def\smallplotMIN{115}
\def\smallplotMAX{9000000}
\def\scatterPlot#1#2{\node at (#1) {\parbox{\linewidth}{\include{experiments/plots/#2}}};}
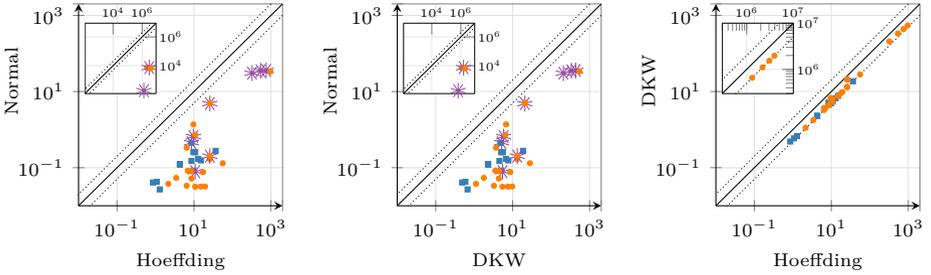
\begin{figure}[tp]
	\vspace{7ex}
	\hspace*{-.15\linewidth}%
	\begin{tikzpicture}
		\node at (0.0,0) {\parbox{\linewidth}{\include{experiments/plots/scatter-coverage-normal-hoeffding.tex}}};
		\node at (4.2,0) {\parbox{\linewidth}{\include{experiments/plots/scatter-coverage-normal-dkw.tex}}};
		\node at (8.4,0) {\parbox{\linewidth}{\include{experiments/plots/scatter-coverage-hoeffding-dkw.tex}}};
	\end{tikzpicture}
	\caption{Widths comparison of DKW, Hoeffding, and Normal 95\% intervals}
	\label{fig:CoverageTests}
\end{figure}
\endgroup

\Cref{fig:CoverageTests} shows log-log scatter plots for the coverage tests from \Cref{sec:ExperimentsForRewards}.
Each plot compares the half-width of 95\% confidence intervals of two methods, built by \thetool with $k=1000$ for the 44 instances from \Cref{tab:Models}.
Per plot, the method that builds the narrower intervals is placed in the $y$-axis: the distance of a mark from the main diagonal indicates how much narrower is that interval.
For instance, a mark at coordinate $x=10^1$ and $y=10^{-1}$ in the left plot of \Cref{fig:CoverageTests} indicates that the Normal interval for that instance is 100$\times$ narrower than the Hoeffding interval.
The right plot in \Cref{fig:CoverageTests} is equivalent to \Cref{fig:SoundPerformanceScatter} (page~\pageref{fig:SoundPerformanceScatter}), with wider intervals since \Cref{fig:SoundPerformanceScatter} was built with $k=500000$.

In \Cref{fig:CoverageTests}, we mark DTMC instances as~\tikz{\node[color=color1,mark size=1.45]{\pgfuseplotmark{square*}}}, MDP instances as~\tikz{\node[color=color5,mark size=1.55]{\pgfuseplotmark{*}}}, and the 10 Normal intervals that did not reach 95\% coverage as~\tikz{\node[color=color4,mark size=3]{\pgfuseplotmark{10-pointed star}}}\,.
\Cref{tab:CoverageTests,tab:CoverageTestsRelative} show the exact widths of all intervals, highlighting those violated coverage.
These include tests for Student's $t$ intervals, which are not shown in \Cref{fig:CoverageTests} due to their similarity to Normal intervals, to which they converge in the limit.

\begin{table}[tp]
	\centering
	\caption{CI half-widths of Normal, Student's $t$, Hoeffding, and DKW intervals}
	\label{tab:CoverageTests}
	\includegraphics[width=.95\linewidth]{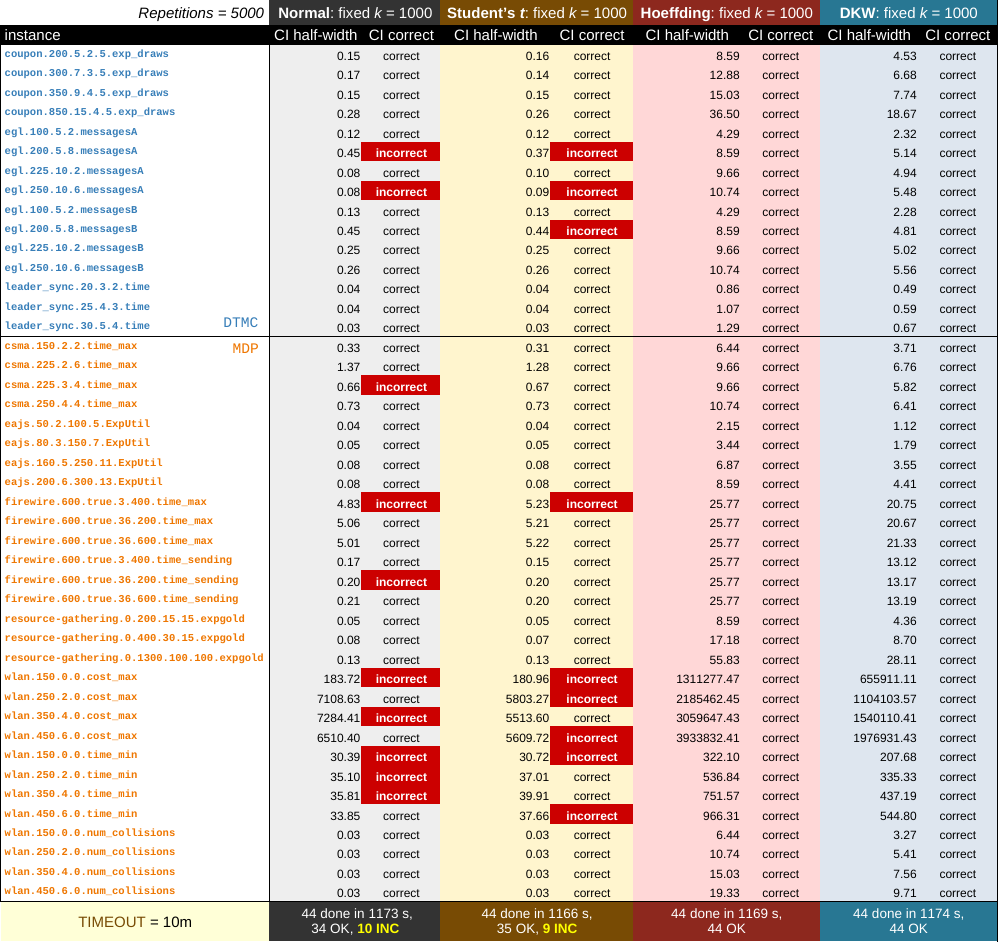}
\end{table}

\begin{table}[tp]
	\centering
	\caption{CI half-widths of Chow-Robbins and Hoeffding given $\varepsilon$}
	\label{tab:CoverageTestsRelative}
	\includegraphics[width=.75\linewidth]{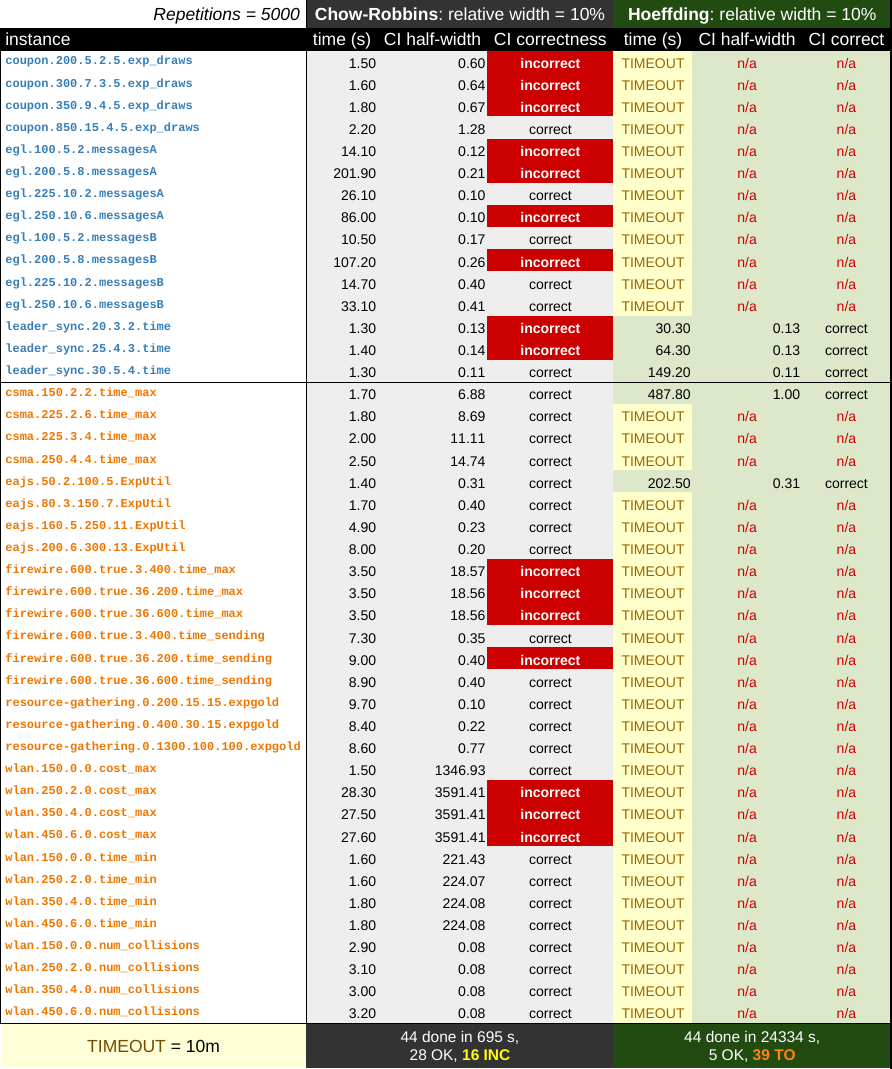}
\end{table}

\subsection{Effectiveness of DKW}
\label{sec:Appendix:EffectivenessDKW}

\Cref{tab:EffectivenessDKW} shows the lower-bound of DKW 95\% intervals for our 44 instances, ran for increasing values of $k\in\{10^i\}_{i=2}^6$.
The four rightmost columns show the ratio at which the relative error of the lower bound decreases every time $k$ increases an order of magnitude, and its bottom row shows the corresponding geometric mean.
On average, \emph{the lower-bound of the DKW interval gets $2.6\times$ closer to the expected reward for every step}.

The instances shown in \Cref{fig:DKWELowerExamples} on page~\pageref{fig:DKWELowerExamples} are
\texttt{\smaller\color{color1}coupon.850.15.4.5.exp\_draws},
\texttt{\smaller\color{color1}egl.200.5.8.messagesB}, and
\texttt{\smaller\color{color5}resource-gathering.0.1300.100.100.expgold},
with expected rewards
$E^{\leq 850}_{\diamond\,G}=12.8184$,
$E^{\leq 200}_{\diamond\,G}=2.5908$, and
$E^{\leq 1300}_{\diamond\,G}=7.7370$ resp.

\begin{table}[tp]
	\centering
	\caption{Approximation of the lower-bound of DKW intervals to the expected reward}
	\label{tab:EffectivenessDKW}
	\includegraphics[angle=90,origin=c,width=0.95\linewidth]{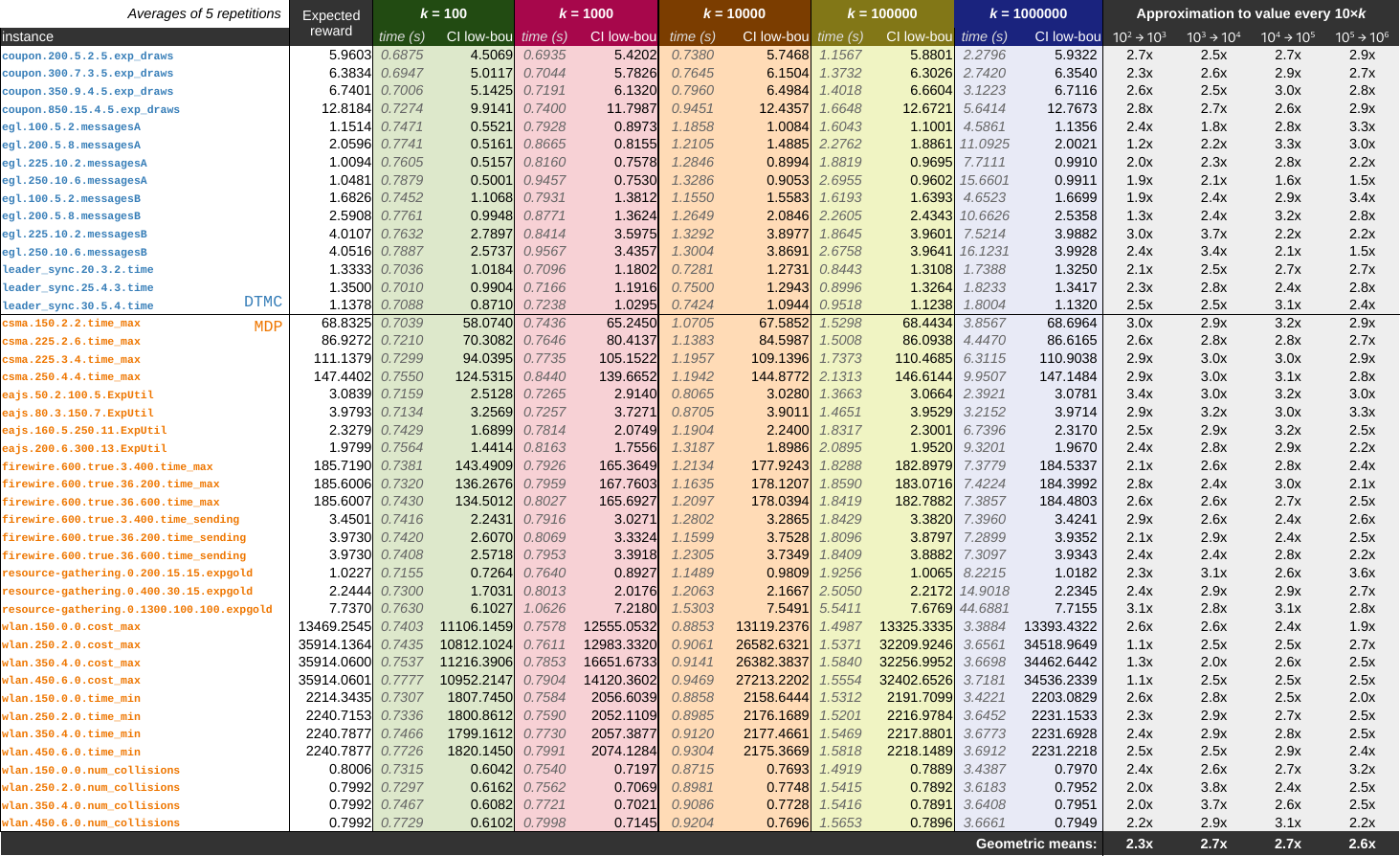}
\end{table}
}{} %

\end{document}

%% file: figures/dtmc-example-tr.tex
\newcommand{\drawcirc}{\node[draw,circle,minimum size=.7cm, outer sep=1pt]}
\newcommand{\drawdiamond}{\node[draw,diamond,minimum size=.7cm, outer sep=1pt]}
\newcommand{\drawbox}{\node[draw,rectangle,minimum size=.7cm, outer sep=1pt]}
\newcommand{\drawdummy}{\node[minimum size=0,inner sep=0]}
\newcommand{\drawnum}{\node[minimum size=.4cm,inner sep=1pt]}

\centering
\begin{minipage}[b]{0.5\linewidth}
\centering
\begin{tikzpicture}
	\drawcirc (s) at (0,0) {$s$};
	\drawdummy (mid1) at (1.7,0) {};
	\drawcirc (t1) at (2,0.5) {$t_1$};
	\drawdummy (t1r) at (2,1.05) {$+cn$};
	\drawcirc (t2) at (2,-0.5) {$t_2$};
	\drawdummy (t2r) at (2,-1.05) {$+0$};
	\draw[->] (-0.75,0) to (s);
	\draw[->] [out=10,in=-180] (s) to node [above,pos=0.5,inner sep=0.5mm] {\scriptsize$\frac{1}{n}~~~$} (t1);
	\draw[->] [out=-10,in=-180] (s) to node [below,pos=0.5,inner sep=1mm] {\scriptsize$\frac{n-1}{n}~~~$} (t2);
	\draw[->] [loop,out=-20,in=20,looseness=4] (t1) to node [right,inner sep=1mm] {\scriptsize$1$} (t1);
	\draw[->] [loop,out=-20,in=20,looseness=4] (t2) to node [right,inner sep=1mm] {\scriptsize$1$} (t2);
\end{tikzpicture}
\caption{High reward with low probability}
\label{fig:ExampleForUnsound}
\end{minipage}%
\begin{minipage}[b]{0.5\linewidth}
\centering
\begin{tikzpicture}
	\drawcirc (s) at (0,0) {$s$};
	\drawdummy (sr) at (0,-0.55) {$+1$};
	\drawcirc (t) at (2,0) {$t$};
	\draw[->] (-0.75,0) to (s);
	\draw[->] [out=-10,in=-180] (s) to node [above,pos=0.5,inner sep=1mm] {\scriptsize$\frac{1}{2}$} (t);
	\draw[->] [loop,out=10,in=80,looseness=4] (s) to node [above right,inner sep=0.33mm] {\scriptsize$\frac{1}{2}$} (s);
	\draw[->] [loop,out=-20,in=20,looseness=4] (t) to node [right,inner sep=1mm] {\scriptsize$1$} (t);
\end{tikzpicture}
\caption{Unbounded path rewards}
\label{fig:ExampleForUnbounded}
\end{minipage}
\vspace{-12pt}

%% file: table_tools.tex
\begingroup
\newlength{\colwidth}
\setlength{\colwidth}{.11\linewidth}
\setlength{\tabcolsep}{3.5pt}
\setlength{\extrarowheight}{.3ex}
\rowcolors{6}{shade4}{shade2}
\smaller[.5]
\begin{tabular}{%
  >{}L{.131\linewidth}l  
  >{\,}cc@{~~}  
  >{\,}cc  
	}%
	\toprule
	\multicolumn{2}{c}{\bfseries Tool}
	& \multicolumn{2}{c}{\bfseries For probabilities $p\in[0,1]$}
	& \multicolumn{2}{c}{\bfseries For rewards $r\in[a,b]$}
	\\
	\cmidrule(l{1.1\tabcolsep}r{1.2\tabcolsep}){1-2}
	\cmidrule(l{1.1\tabcolsep}r{1.2\tabcolsep}){3-4}
	\cmidrule(l{1.1\tabcolsep}r{1.2\tabcolsep}){5-6}
	Name & Ref.
	& \larger fixed $k$ & \larger seq.\ $\varepsilon$  
	& \larger fixed $k$ & \larger seq.\ $\varepsilon$  
    \\
	\midrule     
	\tool{C-SMC} & \cite{CDGL21}
	& \Okamoto            & ---
	& ---                 & ---
	\\
	\tool{Cosmos} & \cite{BBD+15}
	& \ClopperPearson     & \ChowRobbins\UNSOUND
	& \Hoeffding          & \ChowRobbins\UNSOUND
	\\
	\tool{Fig} & \cite{Bud22a}*
	& \WilsonNonCC        & \SeqStudent\UNSOUND
	& ---                 & ---
	\\
	\tool{Hypeg} & \cite{PR17}
	& ---                 & \SeqStudent\UNSOUND
	& ---                 & ---
	\\
	\tool{modes} {\tiny (prev.)\!\!\!\!} & \cite{BDHS20}*
	& \Okamoto            & Chen\UnknownSound
	& \Normal\UNSOUND & \ChowRobbins\UNSOUND
	\\
	\tool{MultiVeSta} & \cite{GRV17}
	& ---                 &  \ChowRobbins\UNSOUND
	& ---                 &  \ChowRobbins\UNSOUND
	\\
	\tool{Plasma~Lab} & \cite{LST16}*
	& \Okamoto            & \Okamoto
	& \Hoeffding          & \Hoeffding
	\\
	\tool{Prism} & \cite{KNP11}*
	& \Student\UNSOUND    & \SeqStudent\UNSOUND
	& \Student\UNSOUND    & \SeqStudent\UNSOUND
	\\
	\tool{Sbip} & \cite{NMB+18}
	& \Normal\UNSOUND     & \Okamoto
	& ---                 & ---
	\\
	\tool{Smc\,Storm}\!\!\! & \cite{LKP24}
	& ---                 & Chen\UnknownSound
	& ---                 & \ChowRobbins\UNSOUND
	\\
	\tool{Uppaal\,Smc} & \cite{DLL+15}*
	& \ClopperPearson     & \SeqClopperPearson\UNSOUND
	& \Student\UNSOUND    & ---
	\\
	\midrule
	\multicolumn{2}{l}{\thetool v3.1.273}
	& \Wilson             & \Wilson
	& \textbf{DKW}        & \Hoeffding
	\\
	\bottomrule
\end{tabular}
\endgroup

%% file: experiments/plots/scatter-sound-performance.tex
\newlength{\plotsizeSoundPerformace}
\setlength{\plotsizeSoundPerformace}{\linewidth}

\definecolor{color1}{RGB}{55,126,184}
\definecolor{color2}{RGB}{228,26,28}
\definecolor{color3}{RGB}{77,175,74}
\definecolor{color4}{RGB}{152,78,163}
\definecolor{color5}{RGB}{255,127,0}

\pgfplotstableread[col sep=comma]{
category,               instance,                x,                  y,
dtmc, coupon.200.5.2.5.exp_draws,                0.121472292381661,  0.0655950577314795,  
dtmc, coupon.300.7.3.5.exp_draws,                0.182208438572491,  0.0965705123549139,
dtmc, coupon.350.9.4.5.exp_draws,                0.212576511667907,  0.111213347501585,
dtmc, coupon.850.15.4.5.exp_draws,               0.516257242622059,  0.267251323495673,
dtmc, egl.100.5.2.messagesA,                     0.0607361461908305, 0.0337085610758908,
dtmc, egl.200.5.8.messagesA,                     0.121472292381661,  0.0731870561745688,
dtmc, egl.225.10.2.messagesA,                    0.136656328929369,  0.0741212524623308,
dtmc, egl.250.10.6.messagesA,                    0.151840365477076,  0.0938422999586155,
dtmc, egl.100.5.2.messagesB,                     0.0607361461908305, 0.0331011996392432,
dtmc, egl.200.5.8.messagesB,                     0.121472292381661,  0.0725796946417249,
dtmc, egl.225.10.2.messagesB,                    0.136656328929369,  0.0740380375128772,
dtmc, egl.250.10.6.messagesB,                    0.151840365477076,  0.0937782650274824,
dtmc, leader_sync.20.3.2.time,                   0.0121472292381661, 0.00728833745889894,
dtmc, leader_sync.25.4.3.time,                   0.0151840365477076, 0.00880674110879243,
dtmc, leader_sync.30.5.4.time,                   0.0182208438572491, 0.0101287640813915,
mdp,  csma.150.2.2.time_max,                     0.091104219286251,  0.0569248763879216,
mdp,  csma.225.2.6.time_max,                     0.136656328929362,  0.109921763054416,
mdp,  csma.225.3.4.time_max,                     0.136656328929362,  0.0928736407695809,
mdp,  csma.250.4.4.time_max,                     0.151840365477085,  0.0979821317887598,
mdp,  eajs.50.2.100.5.ExpUtil,                   0.0303680730954152, 0.0163164186002176,
mdp,  eajs.80.3.150.7.ExpUtil,                   0.0485889169526644, 0.0258152019624144,
mdp,  eajs.160.5.250.11.ExpUtil,                 0.0971778339053286, 0.0508270620460755,
mdp,  eajs.200.6.300.13.ExpUtil,                 0.121472292381661,  0.0629809113071296,
mdp,  firewire.600.true.3.400.time_max,          0.364416877144976,  0.332632718468176,
mdp,  firewire.600.true.36.200.time_max,         0.364416877144976,  0.332604558527518,
mdp,  firewire.600.true.36.600.time_max,         0.364416877144976,  0.332608778530434,
mdp,  firewire.600.true.3.400.time_sending,      0.364416877144983,  0.18758861022571,
mdp,  firewire.600.true.36.200.time_sending,     0.364416877144983,  0.189519214532999,
mdp,  firewire.600.true.36.600.time_sending,     0.364416877144983,  0.189518074532741,
mdp,  resource-gathering.0.200.15.15.expgold,    0.121472292381661,  0.0619887090739123,
mdp,  resource-gathering.0.400.30.15.expgold,    0.242944584763322,  0.123521996747668,
mdp,  resource-gathering.0.1300.100.100.expgold, 0.789569900480797,  0.399264980564367,
mdp,  wlan.150.0.0.time_min,                     4.55521096431221,   3.15130354267953,
mdp,  wlan.250.2.0.time_min,                     7.59201827385368,   5.20671332478901,
mdp,  wlan.350.4.0.time_min,                     10.6288255833952,   7.19998123015262,
mdp,  wlan.450.6.0.time_min,                     13.6656328929371,   8.79686572759365,
mdp,  wlan.150.0.0.num_collisions,               0.0911042192862458, 0.0465071718365052,
mdp,  wlan.250.2.0.num_collisions,               0.151840365477076,  0.0768391049323388,
mdp,  wlan.350.4.0.num_collisions,               0.212576511667907,  0.107215558024493,
mdp,  wlan.450.6.0.num_collisions,               0.273312657858737,  0.137585551103147,
}{\HoeffdingDKWdataLarge}

\pgfplotstableread[col sep=comma]{
category,          instance,                     x,                  y,
mdp,  wlan.150.0.0.cost_max,                     18544.2638357153,   9277.62297092179,
mdp,  wlan.250.2.0.cost_max,                     30907.1063928589,   15785.5013895411,
mdp,  wlan.350.4.0.cost_max,                     43269.9489500024,   21970.7118804045,
mdp,  wlan.450.6.0.cost_max,                     55632.791507146,    28153.8906514648,
}{\HoeffdingDKWdataSmall}

\begingroup
\def\xORIGIN{-1.3}
\def\yORIGIN{-1.4}
\def\largeplotMIN{0.005}
\def\largeplotMAX{30}
\def\smallplotMIN{4999}
\def\smallplotMAX{100000}
\begin{tikzpicture}[overlay]
  \node[anchor=north west] (largeplot) at (\xORIGIN,\yORIGIN) {%
  \begin{axis}[
    width=\plotsizeSoundPerformace,
    height=\plotsizeSoundPerformace,
    axis equal image,
    xmin={\largeplotMIN}, ymin={\largeplotMIN}, xmax={\largeplotMAX}, ymax={\largeplotMAX},
    xmode=log, ymode=log,
    grid=major,
    grid style={color=black!11},
    axis x line=bottom,
    axis y line=left,
    xlabel={\strut{}Hoeffding},
    xlabel style={font=\scriptsize,yshift=5pt},
    ylabel={\strut{}DKW},
    ylabel style={font=\scriptsize,yshift=-13pt,xshift=9pt},
    yticklabel style={font=\scriptsize},
    xticklabel style={font=\scriptsize},
    legend pos=north east,
    legend columns=-1,
    legend style={nodes={scale=0.75, transform shape},anchor=north east,inner sep=1.5pt,yshift=16pt,xshift=3.5pt},
    legend cell align={left},
    extra x ticks = {\largeplotMAX},
    extra x tick labels = {},  
    extra x tick style = {grid=major,grid style={very thin,color=gray}},
    extra y ticks = {\largeplotMAX},
    extra y tick labels = {},  
    extra y tick style = {grid=major,grid style={very thin,color=gray}},
  ]
  \addplot[
    scatter,
    only marks,
    scatter/classes={dtmc={mark=square*,color1,mark size=1.35},
                     mdp={mark=*,color5,mark size=1.25}},
    scatter src=explicit symbolic
  ]%
  table[x={x}, y={y}, meta={category}]{\HoeffdingDKWdataLarge};
  
  \addplot[no marks] coordinates {(\largeplotMIN,\largeplotMIN) (\largeplotMAX,\largeplotMAX)};
  \addplot[no marks, densely dotted] coordinates {(0.5*\largeplotMIN,\largeplotMIN) (0.5*\largeplotMAX,\largeplotMAX)};
  \addplot[no marks, densely dotted] coordinates {(\largeplotMIN,0.5*\largeplotMIN) (\largeplotMAX,0.5*\largeplotMAX)};
  \end{axis}
  };
  \node[anchor=north west] (smallplot) at
	([xshift=-.05\plotsizeSoundPerformace,yshift=.4\plotsizeSoundPerformace]largeplot.north west) {%
  \begin{axis}[
    width=.55\plotsizeSoundPerformace,
    height=.55\plotsizeSoundPerformace,
    axis equal image,
    xmin={\smallplotMIN}, ymin={\smallplotMIN}, xmax={\smallplotMAX}, ymax={\smallplotMAX},
    xmode=log, ymode=log,
    grid=major,
    grid style={color=black!9},
    axis x line*={top,very thin},
    axis y line*={right,very thin},
    yticklabel style={font=\fontsize{4.5}{7.2},xshift=-2pt},
    xticklabel style={font=\fontsize{4.5}{7.2},yshift=-1pt},
    axis background/.style={fill=white},
  ]
  \addplot[
	color=red,
    scatter,
    only marks,
    scatter/classes={dtmc={mark=square*,color1,mark size=1.25},
                     mdp={mark=*,color5,mark size=1.25}},
    scatter src=explicit symbolic
  ]%
  table[x={x}, y={y}, meta={category}]{\HoeffdingDKWdataSmall};
  
  \addplot[no marks, thin] coordinates {(\smallplotMIN,\smallplotMIN) (\smallplotMAX,\smallplotMAX)};
  \addplot[no marks, very thin] coordinates {(\smallplotMIN,\smallplotMIN) (\smallplotMAX,\smallplotMIN)};
  \addplot[no marks, very thin] coordinates {(\smallplotMIN,\smallplotMIN) (\smallplotMIN,\smallplotMAX)};
  \addplot[no marks, thin, densely dotted] coordinates {(0.5*\smallplotMIN,\smallplotMIN) (0.5*\smallplotMAX,\smallplotMAX)};
  \addplot[no marks, thin, densely dotted] coordinates {(\smallplotMIN,0.5*\smallplotMIN) (\smallplotMAX,0.5*\smallplotMAX)};
  \end{axis}
  };
\end{tikzpicture}%
\endgroup

%% file: experiments/plots/scatter-DKW-lower-bound.tex
\newlength{\lboundplotsize}
\setlength{\lboundplotsize}{1.05\linewidth}

\definecolor{color1}{RGB}{55,126,184}
\definecolor{color2}{RGB}{77,175,74}
\definecolor{color3}{RGB}{228,26,28}
\definecolor{color4}{RGB}{152,78,163}
\definecolor{color5}{RGB}{255,127,0}

\begingroup
\def\xORIGIN{-.4}
\def\yORIGIN{-1.45}
\def\xMIN{40}
\def\xMAX{4000000}
\def\yMIN{0}
\def\yMAX{14}
\begin{tikzpicture}[overlay]
  %
  \node at (\xORIGIN,\yORIGIN) {%
  \begin{axis}[
    width=\lboundplotsize,
    height=.8\lboundplotsize,
    xmin={\xMIN}, ymin={\yMIN}, xmax={\xMAX}, ymax={\yMAX},
    xmode=log, 
    grid=major,
    grid style={color=black!11},
    axis x line=bottom,
    axis y line=left,
    xtick={100,1000,10000,100000,1000000},
    xlabel={\strut{}$k$},
    xlabel style={font=\scriptsize,yshift=13pt,xshift=.36\lboundplotsize},
    ylabel={\strut{}\rotatebox{270}{exp.\ reward}},
    ylabel style={font=\scriptsize,yshift=-50pt,xshift=.28\lboundplotsize},
    yticklabel style={font=\scriptsize},
    xticklabel style={font=\scriptsize},
    subtickwidth=0pt,
    extra y ticks={2.5908,    
                   7.737,     
                   12.8184},  
    extra y tick labels={\textcolor{color3}{2.6},   
                         \textcolor{color2}{7.7},   
                         \textcolor{color4}{12.8}}, 
    extra y tick style={font={\bfseries}},
  ]
  \addplot[color3] coordinates { (\xMIN,2.5908) (\xMAX,2.5908) };
  \addplot[
    only marks,
    color3,
    mark=triangle*,
    mark size=1.85,
  ] coordinates {
    (100,0.994759393703753) 
    (1000,1.36241223666131) 
    (10000,2.08460040886108) 
    (100000,2.43434293074541)  
    (1000000,2.53579184087476)
  };
  \addplot[color2] coordinates { (\xMIN,7.7370) (\xMAX,7.7370) };
  \addplot[
    only marks,
    color2,
    mark=square*,
    mark size=1.25,
  ] coordinates {
    (100,6.10273645394457) 
    (1000,7.21801548631783) 
    (10000,7.54912680295348) 
    (100000,7.67690628283304)  
    (1000000,7.71554847724557)
  };
  \addplot[color4] coordinates { (\xMIN,12.8184) (\xMAX,12.8184) };
  \addplot[
    only marks,
    color4,
    mark=diamond*,
    mark size=1.85,
  ] coordinates {
    (100,9.91406539292614) 
    (1000,11.798735466138) 
    (10000,12.4356736059066) 
    (100000,12.672069177495)  
    (1000000,12.7672945484988)
  };

  \end{axis}
  };
\end{tikzpicture}%
\endgroup

%% file: table_models.tex
\begingroup
\definecolor{colorDTMC}{RGB}{55,126,184}
\definecolor{colorMDP}{RGB}{255,127,0}
\def\DTMC{\textcolor{colorDTMC}{\hfill\texttt{\smaller DTMC}}}
\def\MDP{\textcolor{colorMDP}{\hfill\texttt{\smaller MDP}}}
\def\DoMe{\textcolor{RubineRed!44}{\emph{??}}}
\def\firstColWidth{.155\linewidth}
\def\lastColWidth{.25\linewidth}
\setlength{\tabcolsep}{3.5pt}
\setlength{\extrarowheight}{.3ex}
\smaller[.5]
\begin{tabular}{ll@{~\:}c@{~\:}ccccc@{~\,}l}
	\toprule
	\multicolumn{7}{c}{\bfseries Models}
	& \multicolumn{2}{c}{\bfseries Properties}
	\\
	\cmidrule(l{\tabcolsep}r{.8\tabcolsep}){1-7}
	\cmidrule(l{.8\tabcolsep}r{\tabcolsep}){8-9}
	ID
	& Instance
	& $c$  
	& $|S|$
	& $\overline{|S|}$
	& $\overline{r}_\mathit{\!max}$
	& $\underline{p}_\mathit{min}$
	& Value
	& ID
	\\
	\cmidrule(l{\tabcolsep}r{.8\tabcolsep}){1-7}
	\cmidrule(l{.8\tabcolsep}r{\tabcolsep}){8-9}
	\arrayrulecolor{black}
	& \ttfamily 5,2,5
	& \slshape 200
	& $5.4\textsc{e}3$
	& $1.2\textsc{e}4$
	& 1
	& .200
	& 5.96
	& 
	\\
	& \ttfamily 7,3,5
	& \slshape 300
	& $3.4\textsc{e}5$
	& $8.8\textsc{e}5$
	& 1
	& .143
	& 6.38
	& 
	\\
	& \ttfamily 9,4,5
	& \slshape 350
	& $2.8\textsc{e}7$
	& $8.4\textsc{e}7$
	& 1
	& .111
	& 6.74
	& 
	\\
	\multirow{-4}{*}{\parbox{\firstColWidth}{%
		coupon \DTMC\\[1ex]
		\smaller\ttfamily N,DRAWS,B}}
	& \ttfamily 15,4,5
	& \slshape 850
	& $1.7\textsc{e}10$
	& $6.1\textsc{e}10$
	& 1
	& .067
	& 12.82
	& \multirow{-4}{*}{\parbox{\lastColWidth}{\RaggedRight%
		\texttt{exp\_draws}\\[.5ex]
		$E(\#\mathit{draw})$ to collect $\mathtt{N}$ unique coupons,
		drawing $\mathtt{DRAWS}$ coupons per $\mathit{draw}$}}
	\\
	\arrayrulecolor{black}
	\cmidrule(l{\tabcolsep}r{.8\tabcolsep}){1-7}
	\cmidrule(l{.8\tabcolsep}r{\tabcolsep}){8-9}
	& \ttfamily 5,2
	& \slshape 100
	& $3.4\textsc{e}4$
	& $1.5\textsc{e}40$
	& 1
	& .500
	& 1.15
	& 
	\\
	& \ttfamily 5,8
	& \slshape 200
	& $1.6\textsc{e}5$
	& $8.7\textsc{e}78$
	& 1
	& .500
	& 2.06
	& 
	\\
	& \ttfamily 10,2
	& \slshape 225
	& $6.6\textsc{e}7$
	& $3.0\textsc{e}40$
	& 1
	& .500
	& 1.01
	& 
	\\
	& \ttfamily 10,6
	& \slshape 250
	& $2.3\textsc{e}8$
	& $2.4\textsc{e}70$
	& 1
	& .500
	& 1.05
	& \multirow{-4}{*}{\parbox{\lastColWidth}{\RaggedRight%
		\texttt{messagesA}\\[.5ex]
		$E(\#\mathit{msgs})$ from $A$ that $B$ needs to know}}
	\\
	\arrayrulecolor{gray}
	\cmidrule[.2pt](l{\tabcolsep}r{.8\tabcolsep}){2-7}
	\cmidrule[.2pt](l{.8\tabcolsep}r{\tabcolsep}){8-9}
	& \ttfamily 5,2
	& \slshape 100
	& $3.4\textsc{e}4$
	& $1.5\textsc{e}40$
	& 1
	& .500
	& 1.68
	& 
	\\
	& \ttfamily 5,8
	& \slshape 200
	& $1.6\textsc{e}5$
	& $8.7\textsc{e}78$
	& 1
	& .500
	& 2.59
	& 
	\\
	& \ttfamily 10,2
	& \slshape 225
	& $6.6\textsc{e}7$
	& $3.0\textsc{e}40$
	& 1
	& .500
	& 4.01
	& 
	\\
	\multirow{-8}{*}{\parbox{\firstColWidth}{%
		egl \DTMC\\[1ex]
		\smaller\ttfamily N,L}}
	& \ttfamily 10,6
	& \slshape 250
	& $2.3\textsc{e}8$
	& $2.4\textsc{e}70$
	& 1
	& .500
	& 4.05
	& \multirow{-4}{*}{\parbox{\lastColWidth}{\RaggedRight%
		\texttt{messagesB}\\[.5ex]
		$E(\#\mathit{msgs})$ from $B$ that $A$ needs to know}}
	\\
	\arrayrulecolor{black}
	\cmidrule(l{\tabcolsep}r{.8\tabcolsep}){1-7}
	\cmidrule(l{.8\tabcolsep}r{\tabcolsep}){8-9}
	& \ttfamily 3,2
	& \slshape 20
	& $2.6\textsc{e}1$
	& $4.2\textsc{e}6$
	& 1
	& .500
	& 1.33
	& 
	\\
	& \ttfamily 4,3
	& \slshape 25
	& $2.7\textsc{e}2$
	& $2.1\textsc{e}10$
	& 1
	& .333
	& 1.35
	& 
	\\
	\multirow{-3}{*}{\parbox{\firstColWidth}{%
		leader\texttt{\textunderscore}sync\DTMC\!\\[1ex]
		\smaller\ttfamily N,K}}
	& \ttfamily 10,6
	& \slshape 30
	& $4.2\textsc{e}3$
	& $1.4\textsc{e}14$
	& 1
	& .250
	& 1.14
	& \multirow{-3}{*}{\parbox{\lastColWidth}{\RaggedRight%
		\texttt{time}\\[.5ex]
		$E(\#\mathit{rounds})$ to elect a leader}}
	\\
	\arrayrulecolor{black}
	\cmidrule(l{\tabcolsep}r{.8\tabcolsep}){1-7}
	\cmidrule(l{.8\tabcolsep}r{\tabcolsep}){8-9}
	& \ttfamily 2,2
	& \slshape 150
	& $1.0\textsc{e}3$
	& $1.3\textsc{e}8$
	& 1
	& .250
	& 68.83
	& 
	\\
	& \ttfamily 2,6
	& \slshape 225
	& $6.7\textsc{e}4$
	& $1.9\textsc{e}11$
	& 1
	& .016
	& 86.93
	& 
	\\
	& \ttfamily 3,4
	& \slshape 225
	& $1.5\textsc{e}6$
	& $8.9\textsc{e}13$
	& 1
	& .063
	& 111.1
	& 
	\\
	\multirow{-4}{*}{\parbox{\firstColWidth}{%
		csma \MDP\\[1ex]
		\smaller\ttfamily N,K}}
	& \ttfamily 4,4
	& \slshape 250
	& $1.3\textsc{e}8$
	& $1.3\textsc{e}18$
	& 1
	& .063
	& 147.4
	& \multirow{-4}{*}{\parbox{\lastColWidth}{\RaggedRight%
		\texttt{time\_max}\\[.5ex]
		$E(\mathrm{unif}(\mathit{time}))$ to send all messages}}
	\\
	\arrayrulecolor{black}
	\cmidrule(l{\tabcolsep}r{.8\tabcolsep}){1-7}
	\cmidrule(l{.8\tabcolsep}r{\tabcolsep}){8-9}
	& \ttfamily 2,5,100
	& \slshape 50
	& $1.3\textsc{e}4$
	& $2.9\textsc{e}7$
	& 1
	& .333
	& 3.08
	& 
	\\
	& \ttfamily 3,7,150
	& \slshape 80
	& $1.4\textsc{e}5$
	& $7.0\textsc{e}9$
	& 1
	& .333
	& 3.98
	& 
	\\
	& \ttfamily 5,11,250
	& \slshape 160
	& $3.0\textsc{e}6$
	& $2.1\textsc{e}14$
	& 1
	& .333
	& 2.33
	& 
	\\
	\multirow{-4}{*}{\parbox{\firstColWidth}{%
		eajs \MDP\\[1ex]
		\smaller\ttfamily N,B, energy\_capacity}}
	& \ttfamily 6,13,300
	& \slshape 200
	& $7.9\textsc{e}6$
	& $3.1\textsc{e}16$
	& 1
	& .333
	& 1.98
	& \multirow{-4}{*}{\parbox{\lastColWidth}{\RaggedRight%
		\texttt{ExpUtil}\\[.5ex]
		$E(\mathrm{unif}(\mathit{utility}))$ obtained until
		the battery becomes empty}}
	\\
	\arrayrulecolor{black}
	\cmidrule(l{\tabcolsep}r{.8\tabcolsep}){1-7}
	\cmidrule(l{.8\tabcolsep}r{\tabcolsep}){8-9}
	& \ttfamily true,3,400
	& \slshape 600
	& $4.3\textsc{e}5$
	& $1.2\textsc{e}14$
	& 1
	& .500
	& 185.7
	& 
	\\
	& \ttfamily true,36,200
	& \slshape 600
	& $6.9\textsc{e}6$
	& $3.9\textsc{e}17$
	& 1
	& .500
	& 185.6
	& 
	\\
	& \ttfamily true,36,600
	& \slshape 600
	& $8.7\textsc{e}7$
	& $3.9\textsc{e}17$
	& 1
	& .500
	& 185.6
	& \multirow{-3}{*}{\parbox{\lastColWidth}{\RaggedRight%
		\texttt{time\_max}\\[.5ex]
		$E(\mathrm{unif}(\mathit{time}))$ to elect a leader}}
	\\
	\arrayrulecolor{gray}
	\cmidrule[.2pt](l{\tabcolsep}r{.8\tabcolsep}){2-7}
	\cmidrule[.2pt](l{.8\tabcolsep}r{\tabcolsep}){8-9}
	& \ttfamily true,3,400
	& \slshape 600
	& $4.3\textsc{e}5$
	& $1.2\textsc{e}14$
	& 1
	& .500
	& 3.45
	& 
	\\
	& \ttfamily true,36,200
	& \slshape 600
	& $6.9\textsc{e}6$
	& $3.9\textsc{e}17$
	& 1
	& .500
	& 3.97
	& 
	\\
	\multirow{-6}{*}{\parbox{\firstColWidth}{%
		firewire \MDP\\[1ex]
		\smaller\ttfamily explicit\_timer, delay,deadline}}
	& \ttfamily true,36,600
	& \slshape 600
	& $8.7\textsc{e}7$
	& $3.9\textsc{e}17$
	& 1
	& .500
	& 3.97
	& \multirow{-3}{*}{\parbox{\lastColWidth}{\RaggedRight%
		\texttt{time\_sending}\\[.5ex]
		$E(\mathrm{unif}(\mathit{time}))$ spent sending to
		elect a leader}}
	\\
	\arrayrulecolor{black}
	\cmidrule(l{\tabcolsep}r{.8\tabcolsep}){1-7}
	\cmidrule(l{.8\tabcolsep}r{\tabcolsep}){8-9}
	& \ttfamily 200,15,15
	& \slshape --
	& $2.4\textsc{e}4$
	& $5.1\textsc{e}4$
	& 1
	& .100
	& 1.02
	& 
	\\
	& \ttfamily 400,30,30
	& \slshape --
	& $9.0\textsc{e}4$
	& $1.9\textsc{e}5$
	& 1
	& .100
	& 2.24
	& 
	\\
	\multirow{-3}{*}{\parbox{\firstColWidth}{%
		resources\texttt{\textunderscore}\MDP\\[-.3ex]gathering\\[.7ex]
		\smaller\ttfamily B,GOLD\_TO\_COLLECT, GEM\_TO\_COLLECT}}
	& \ttfamily 1300,100,100
	& \slshape --
	& $9.6\textsc{e}5$
	& $2.0\textsc{e}6$
	& 1
	& .100
	& 7.74
	& \multirow{-3}{*}{\parbox{\lastColWidth}{\RaggedRight%
		\texttt{expgold}\\[.5ex]
		$E(\mathrm{unif}(\mathit{gold}))$ collectable in \texttt{B} steps}}
	\\
	\arrayrulecolor{black}
	\cmidrule(l{\tabcolsep}r{.8\tabcolsep}){1-7}
	\cmidrule(l{.8\tabcolsep}r{\tabcolsep}){8-9}
	& \ttfamily 0,0
	& \slshape 150
	& $3.0\textsc{e}3$
	& $2.8\textsc{e}10$
	& $2.04\textsc{e}5$
	& .063
	& 13469
	& 
	\\
	& \ttfamily 2,0
	& \slshape 250
	& $2.8\textsc{e}4$
	& $9.9\textsc{e}11$
	& $2.04\textsc{e}5$
	& .063
	& 35914
	& 
	\\
	& \ttfamily 4,0
	& \slshape 350
	& $3.5\textsc{e}5$
	& $4.4\textsc{e}13$
	& $2.04\textsc{e}5$
	& .063 
	& 35914
	& 
	\\
	& \ttfamily 6,0
	& \slshape 450
	& $5.0\textsc{e}6$
	& $1.4\textsc{e}15$
	& $2.04\textsc{e}5$
	& .016
	& 35914
	& \multirow{-4}{*}{\parbox{\lastColWidth}{\RaggedRight%
		\texttt{cost\_max}\\[.5ex]
		$E(\mathrm{unif}(\mathit{cost}))$ for both stations to send correctly}}
	\\
	\arrayrulecolor{gray}
	\cmidrule[.2pt](l{\tabcolsep}r{.8\tabcolsep}){2-7}
	\cmidrule[.2pt](l{.8\tabcolsep}r{\tabcolsep}){8-9}
	& \ttfamily 0,0
	& \slshape 150
	& $3.0\textsc{e}3$
	& $2.8\textsc{e}10$
	& 50
	& .063
	& 2214
	& 
	\\
	& \ttfamily 2,0
	& \slshape 250
	& $2.8\textsc{e}4$
	& $9.9\textsc{e}11$
	& 50
	& .063
	& 2241
	& 
	\\
	& \ttfamily 4,0
	& \slshape 350
	& $3.5\textsc{e}5$
	& $4.4\textsc{e}13$
	& 50
	& .063 
	& 2241
	& 
	\\
	& \ttfamily 6,0
	& \slshape 450
	& $5.0\textsc{e}6$
	& $1.4\textsc{e}15$
	& 50
	& .016
	& 2241
	& \multirow{-4}{*}{\parbox{\lastColWidth}{\RaggedRight%
		\texttt{time\_min}\\[.5ex]
		$E(\mathrm{unif}(\mathit{time}))$ for both stations to send correctly}}
	\\
	\arrayrulecolor{gray}
	\cmidrule[.2pt](l{\tabcolsep}r{.8\tabcolsep}){2-7}
	\cmidrule[.2pt](l{.8\tabcolsep}r{\tabcolsep}){8-9}
	& \ttfamily 0,0
	& \slshape 150
	& $3.0\textsc{e}3$
	& $2.8\textsc{e}10$
	& 1
	& .063
	& 0.80
	& 
	\\
	& \ttfamily 2,0
	& \slshape 250
	& $2.8\textsc{e}4$
	& $9.9\textsc{e}11$
	& 1
	& .063
	& 0.80
	& 
	\\
	& \ttfamily 4,0
	& \slshape 350
	& $3.5\textsc{e}5$
	& $4.4\textsc{e}13$
	& 1
	& .063 
	& 0.80
	& 
	\\
	\multirow{-12}{*}{\parbox{\firstColWidth}{%
		wlan \MDP\\[1ex]
		\smaller\ttfamily MAX\_BACKOFF,COL}}
	& \ttfamily 6,0
	& \slshape 450
	& $5.0\textsc{e}6$
	& $1.4\textsc{e}15$
	& 1
	& .016
	& 0.80
	& \multirow{-4}{*}{\parbox{\lastColWidth}{\RaggedRight%
		\texttt{num\_collisions}\\[.5ex]
		$E(\#(\mathit{collisions}))$ until both stations send correctly}}
	\\
	\arrayrulecolor{black}
	\bottomrule
\end{tabular}
\endgroup

%% file: experiments/plots/scatter-coverage-normal-hoeffding.tex
\pgfplotstableread[col sep=comma]{
category,  instance,                                  x,                 y,
dtmc,      coupon.200.5.2.5.exp_draws,                8.58938816693475,  0.150414442010658,	 
dtmc,      coupon.300.7.3.5.exp_draws,                12.8840822504021,  0.170229326100802,	 
dtmc,      coupon.350.9.4.5.exp_draws,                15.0314292921358,  0.154932005953482,	 
dtmc,      coupon.850.15.4.5.exp_draws,               36.5048997094727,  0.276743082721351,	 
dtmc,      egl.100.5.2.messagesA,                     4.29469408346738,  0.119874908974928,	 
incorrect, egl.200.5.8.messagesA,                     8.58938816693475,  0.454003171268597,	 
dtmc,      egl.225.10.2.messagesA,                    9.66306168780159,  0.0837853639382766, 
incorrect, egl.250.10.6.messagesA,                    10.7367352086684,  0.0786819134767031, 
dtmc,      egl.100.5.2.messagesB,                     4.29469408346738,  0.125391581275442,	 
dtmc,      egl.200.5.8.messagesB,                     8.58938816693475,  0.448640561039747,	 
dtmc,      egl.225.10.2.messagesB,                    9.66306168780159,  0.254834334690913,	 
dtmc,      egl.250.10.6.messagesB,                    10.7367352086684,  0.255063085818741,	 
dtmc,      leader_sync.20.3.2.time,                   0.858938816693475, 0.0408550487304114, 
dtmc,      leader_sync.25.4.3.time,                   1.07367352086684,  0.0426704450916353, 
dtmc,      leader_sync.30.5.4.time,                   1.28840822504021,  0.0268939830568464, 
mdp,       csma.150.2.2.time_max,                     6.44204112520106,  0.333755317243686,	 
mdp,       csma.225.2.6.time_max,                     9.6630616878016,   1.36992482625401,	 
incorrect, csma.225.3.4.time_max,                     9.6630616878016,   0.657615592662168,	 
mdp,       csma.250.4.4.time_max,                     10.7367352086684,  0.73370987206718,	 
mdp,       eajs.50.2.100.5.ExpUtil,                   2.14734704173369,  0.0376123302591171, 
mdp,       eajs.80.3.150.7.ExpUtil,                   3.4357552667739,   0.0539803235742864, 
mdp,       eajs.160.5.250.11.ExpUtil,                 6.8715105335478,   0.0830478023687276, 
mdp,       eajs.200.6.300.13.ExpUtil,                 8.58938816693475,  0.081014882458132,	 
incorrect, firewire.600.true.3.400.time_max,          25.7681645008043,  4.82636420645645,	 
mdp,       firewire.600.true.36.200.time_max,         25.7681645008043,  5.06208463676549,	 
mdp,       firewire.600.true.36.600.time_max,         25.7681645008043,  5.00699212354471,	 
mdp,       firewire.600.true.3.400.time_sending,      25.7681645008043,  0.168435801494745,	 
incorrect, firewire.600.true.36.200.time_sending,     25.7681645008043,  0.200431341804805,	 
mdp,       firewire.600.true.36.600.time_sending,     25.7681645008043,  0.205344488979442,	 
mdp,       resource-gathering.0.200.15.15.expgold,    8.58938816693475,  0.0519705841536617, 
mdp,       resource-gathering.0.400.30.15.expgold,    17.1787763338695,  0.0763125257295201, 
mdp,       resource-gathering.0.1300.100.100.expgold, 55.8310230850759,  0.130871148801164,	 
incorrect, wlan.150.0.0.time_min,                     322.102056260053,  30.393898980079,	 
incorrect, wlan.250.2.0.time_min,                     536.836760433422,  35.0996216997701,	 
incorrect, wlan.350.4.0.time_min,                     751.571464606791,  35.8133983154285,	 
mdp,       wlan.450.6.0.time_min,                     966.306168780159,  33.8458169412484,	 
mdp,       wlan.150.0.0.num_collisions,               6.44204112520106,  0.0335028994869795, 
mdp,       wlan.250.2.0.num_collisions,               10.7367352086684,  0.0316919241477416, 
mdp,       wlan.350.4.0.num_collisions,               15.0314292921358,  0.0318175590239919, 
mdp,       wlan.450.6.0.num_collisions,               19.3261233756032,  0.0320583573931118, 
}{\HoeffdingvsNormalLarge}

\pgfplotstableread[col sep=comma]{
category,  instance,                                  x,                 y,
incorrect, wlan.150.0.0.cost_max,                     1311277.47103468,  183.72470894397,	 
mdp,       wlan.250.2.0.cost_max,                     2185462.45172446,  7108.63276269723,	 
incorrect, wlan.350.4.0.cost_max,                     3059647.43241424,  7284.40600736209,	 
mdp,       wlan.450.6.0.cost_max,                     3933832.41310403,  6510.40032633235,	 
}{\HoeffdingvsNormalSmall}

\begin{tikzpicture}
  \node[anchor=north west] (largeplot) at (0,0) {%
  \begin{axis}[
    width=\scatterPlotsSize,
    height=\scatterPlotsSize,
    axis equal image,
    xmin={\largeplotMIN}, ymin={\largeplotMIN}, xmax={\largeplotMAX}, ymax={\largeplotMAX},
    xmode=log, ymode=log,
    grid=major,
    grid style={color=black!11},
    axis x line=bottom,
    axis y line=left,
    xlabel={\strut{}Hoeffding},
    xlabel style={font=\scriptsize,yshift=5pt},
    ylabel={\strut{}Normal},
    ylabel style={font=\scriptsize,yshift=-13pt,xshift=9pt},
    yticklabel style={font=\scriptsize},
    xticklabel style={font=\scriptsize},
    legend pos=north east,
    legend columns=-1,
    legend style={nodes={scale=0.75, transform shape},anchor=north east,inner sep=1.5pt,yshift=16pt,xshift=3.5pt},
    legend cell align={left},
    extra x ticks = {\largeplotMAX},
    extra x tick labels = {},  
    extra x tick style = {grid=major,grid style={very thin,color=gray}},
    extra y ticks = {\largeplotMAX},
    extra y tick labels = {},  
    extra y tick style = {grid=major,grid style={very thin,color=gray}},
  ]
  \addplot[
    scatter,
    only marks,
    scatter/classes={dtmc={mark=square*,color1,mark size=.95},
                     mdp={mark=*,color5,mark size=.95},
                     incorrect={mark=10-pointed star,color4,mark size=3}},
    scatter src=explicit symbolic
  ]%
  table[x={x}, y={y}, meta={category}]{\HoeffdingvsNormalLarge};
  
  \addplot[no marks] coordinates {(\largeplotMIN,\largeplotMIN) (\largeplotMAX,\largeplotMAX)};
  \addplot[no marks, densely dotted] coordinates {(0.5*\largeplotMIN,\largeplotMIN) (0.5*\largeplotMAX,\largeplotMAX)};
  \addplot[no marks, densely dotted] coordinates {(\largeplotMIN,0.5*\largeplotMIN) (\largeplotMAX,0.5*\largeplotMAX)};
  \end{axis}
  };
  \node[anchor=north west] (smallplot) at
	([xshift=-.082\scatterPlotsSize,yshift=.35\scatterPlotsSize]largeplot.north west) {%
  \begin{axis}[
    width=.59\scatterPlotsSize,
    height=.59\scatterPlotsSize,
    axis equal image,
    xmin={\smallplotMIN}, ymin={\smallplotMIN}, xmax={\smallplotMAX}, ymax={\smallplotMAX},
    xmode=log, ymode=log,
    grid=major,
    grid style={color=black!9},
    axis x line*={top,very thin},
    axis y line*={right,very thin},
    yticklabel style={font=\fontsize{4.5}{7.2},xshift=-2pt},
    xticklabel style={font=\fontsize{4.5}{7.2},yshift=-1pt},
    axis background/.style={fill=white},
  ]
  \addplot[
	color=red,
    scatter,
    only marks,
    scatter/classes={dtmc={mark=square*,color1,mark size=.95},
                     mdp={mark=*,color5,mark size=.95},
                     incorrect={mark=10-pointed star,color4,mark size=3}},
    scatter src=explicit symbolic
  ]%
  table[x={x}, y={y}, meta={category}]{\HoeffdingvsNormalSmall};
  
  \addplot[no marks, thin] coordinates {(\smallplotMIN,\smallplotMIN) (\smallplotMAX,\smallplotMAX)};
  \addplot[no marks, very thin] coordinates {(\smallplotMIN,\smallplotMIN) (\smallplotMAX,\smallplotMIN)};
  \addplot[no marks, very thin] coordinates {(\smallplotMIN,\smallplotMIN) (\smallplotMIN,\smallplotMAX)};
  \addplot[no marks, thin, densely dotted] coordinates {(0.5*\smallplotMIN,\smallplotMIN) (0.5*\smallplotMAX,\smallplotMAX)};
  \addplot[no marks, thin, densely dotted] coordinates {(\smallplotMIN,0.5*\smallplotMIN) (\smallplotMAX,0.5*\smallplotMAX)};
  \end{axis}
  };
\end{tikzpicture}

%% file: experiments/plots/scatter-coverage-normal-dkw.tex
\pgfplotstableread[col sep=comma]{
category,  instance,                                  x,                 y,
dtmc,      coupon.200.5.2.5.exp_draws,                4.53198184680609,  0.150414442010658,	 
dtmc,      coupon.300.7.3.5.exp_draws,                6.67630235895711,  0.170229326100802,	 
dtmc,      coupon.350.9.4.5.exp_draws,                7.74000240940674,  0.154932005953482,	 
dtmc,      coupon.850.15.4.5.exp_draws,               18.6690519109967,  0.276743082721351,	 
dtmc,      egl.100.5.2.messagesA,                     2.3227143938204,   0.119874908974928,	 
incorrect, egl.200.5.8.messagesA,                     5.14023901849148,  0.454003171268597,	 
dtmc,      egl.225.10.2.messagesA,                    4.94242472557013,  0.0837853639382766, 
incorrect, egl.250.10.6.messagesA,                    5.47926148600354,  0.0786819134767031, 
dtmc,      egl.100.5.2.messagesB,                     2.27724092340299,  0.125391581275442,	 
dtmc,      egl.200.5.8.messagesB,                     4.80908796513669,  0.448640561039747,	 
dtmc,      egl.225.10.2.messagesB,                    5.02479207765683,  0.254834334690913,	 
dtmc,      egl.250.10.6.messagesB,                    5.56162883809026,  0.255063085818741,	 
dtmc,      leader_sync.20.3.2.time,                   0.486916349181372, 0.0408550487304114, 
dtmc,      leader_sync.25.4.3.time,                   0.588783701268057, 0.0426704450916353, 
dtmc,      leader_sync.30.5.4.time,                   0.673177582937396, 0.0268939830568464, 
mdp,       csma.150.2.2.time_max,                     3.70554303011235,  0.333755317243686,	 
mdp,       csma.225.2.6.time_max,                     6.75830642102725,  1.36992482625401,	 
incorrect, csma.225.3.4.time_max,                     5.81604924934203,  0.657615592662168,	 
mdp,       csma.250.4.4.time_max,                     6.41430642102738,  0.73370987206718,	 
mdp,       eajs.50.2.100.5.ExpUtil,                   1.11812046170156,  0.0376123302591171, 
mdp,       eajs.80.3.150.7.ExpUtil,                   1.78532457422154,  0.0539803235742864, 
mdp,       eajs.160.5.250.11.ExpUtil,                 3.55112261886062,  0.0830478023687276, 
mdp,       eajs.200.6.300.13.ExpUtil,                 4.41056143555409,  0.081014882458132,	 
incorrect, firewire.600.true.3.400.time_max,          20.7531458431148,  4.82636420645645,	 
mdp,       firewire.600.true.36.200.time_max,         20.6724763843185,  5.06208463676549,	 
mdp,       firewire.600.true.36.600.time_max,         21.3265662543669,  5.00699212354471,	 
mdp,       firewire.600.true.3.400.time_sending,      13.1168434841583,  0.168435801494745,	 
incorrect, firewire.600.true.36.200.time_sending,     13.1652904249929,  0.200431341804805,	 
mdp,       firewire.600.true.36.600.time_sending,     13.1882638954102,  0.205344488979442,	 
mdp,       resource-gathering.0.200.15.15.expgold,    4.36011449471932,  0.0519705841536617, 
mdp,       resource-gathering.0.400.30.15.expgold,    8.69728204860414,  0.0763125257295201, 
mdp,       resource-gathering.0.1300.100.100.expgold, 28.1147993058767,  0.130871148801164,	 
incorrect, wlan.150.0.0.time_min,                     207.679601360369,  30.393898980079,	 
incorrect, wlan.250.2.0.time_min,                     335.327056260052,  35.0996216997701,	 
incorrect, wlan.350.4.0.time_min,                     437.185326096334,  35.8133983154285,	 
mdp,       wlan.450.6.0.time_min,                     544.804004662155,  33.8458169412484,	 
mdp,       wlan.150.0.0.num_collisions,               3.26746750343532,  0.0335028994869795, 
mdp,       wlan.250.2.0.num_collisions,               5.41281454516871,  0.0316919241477416, 
mdp,       wlan.350.4.0.num_collisions,               7.5591615869024,   0.0318175590239919, 
mdp,       wlan.450.6.0.num_collisions,               9.70850862863667,  0.0320583573931118, 
}{\DKWvsNormalLarge}

\pgfplotstableread[col sep=comma]{
category,  instance,                                  x,                 y,
incorrect, wlan.150.0.0.cost_max,                     655911.109609679,  183.72470894397,	 
mdp,       wlan.250.2.0.cost_max,                     1104103.56658949,  7108.63276269723,	 
incorrect, wlan.350.4.0.cost_max,                     1540110.40825825,  7284.40600736209,	 
mdp,       wlan.450.6.0.cost_max,                     1976931.42579945,  6510.40032633235,	 
}{\DKWvsNormalSmall}

\begin{tikzpicture}
  \node[anchor=north west] (largeplot) at (0,0) {%
  \begin{axis}[
    width=\scatterPlotsSize,
    height=\scatterPlotsSize,
    axis equal image,
    xmin={\largeplotMIN}, ymin={\largeplotMIN}, xmax={\largeplotMAX}, ymax={\largeplotMAX},
    xmode=log, ymode=log,
    grid=major,
    grid style={color=black!11},
    axis x line=bottom,
    axis y line=left,
    xlabel={\strut{}DKW},
    xlabel style={font=\scriptsize,yshift=5pt},
    ylabel={\strut{}Normal},
    ylabel style={font=\scriptsize,yshift=-13pt,xshift=9pt},
    yticklabel style={font=\scriptsize},
    xticklabel style={font=\scriptsize},
    legend pos=north east,
    legend columns=-1,
    legend style={nodes={scale=0.75, transform shape},anchor=north east,inner sep=1.5pt,yshift=16pt,xshift=3.5pt},
    legend cell align={left},
    extra x ticks = {\largeplotMAX},
    extra x tick labels = {},  
    extra x tick style = {grid=major,grid style={very thin,color=gray}},
    extra y ticks = {\largeplotMAX},
    extra y tick labels = {},  
    extra y tick style = {grid=major,grid style={very thin,color=gray}},
  ]
  \addplot[
    scatter,
    only marks,
    scatter/classes={dtmc={mark=square*,color1,mark size=.95},
                     mdp={mark=*,color5,mark size=.95},
                     incorrect={mark=10-pointed star,color4,mark size=3}},
    scatter src=explicit symbolic
  ]%
  table[x={x}, y={y}, meta={category}]{\DKWvsNormalLarge};
  
  \addplot[no marks] coordinates {(\largeplotMIN,\largeplotMIN) (\largeplotMAX,\largeplotMAX)};
  \addplot[no marks, densely dotted] coordinates {(0.5*\largeplotMIN,\largeplotMIN) (0.5*\largeplotMAX,\largeplotMAX)};
  \addplot[no marks, densely dotted] coordinates {(\largeplotMIN,0.5*\largeplotMIN) (\largeplotMAX,0.5*\largeplotMAX)};
  \end{axis}
  };
  \node[anchor=north west] (smallplot) at
	([xshift=-.082\scatterPlotsSize,yshift=.35\scatterPlotsSize]largeplot.north west) {%
  \begin{axis}[
    width=.59\scatterPlotsSize,
    height=.59\scatterPlotsSize,
    axis equal image,
    xmin={\smallplotMIN}, ymin={\smallplotMIN}, xmax={\smallplotMAX}, ymax={\smallplotMAX},
    xmode=log, ymode=log,
    grid=major,
    grid style={color=black!9},
    axis x line*={top,very thin},
    axis y line*={right,very thin},
    yticklabel style={font=\fontsize{4.5}{7.2},xshift=-2pt},
    xticklabel style={font=\fontsize{4.5}{7.2},yshift=-1pt},
    axis background/.style={fill=white},
  ]
  \addplot[
	color=red,
    scatter,
    only marks,
    scatter/classes={dtmc={mark=square*,color1,mark size=.95},
                     mdp={mark=*,color5,mark size=.95},
                     incorrect={mark=10-pointed star,color4,mark size=3}},
    scatter src=explicit symbolic
  ]%
  table[x={x}, y={y}, meta={category}]{\DKWvsNormalSmall};
  
  \addplot[no marks, thin] coordinates {(\smallplotMIN,\smallplotMIN) (\smallplotMAX,\smallplotMAX)};
  \addplot[no marks, very thin] coordinates {(\smallplotMIN,\smallplotMIN) (\smallplotMAX,\smallplotMIN)};
  \addplot[no marks, very thin] coordinates {(\smallplotMIN,\smallplotMIN) (\smallplotMIN,\smallplotMAX)};
  \addplot[no marks, thin, densely dotted] coordinates {(0.5*\smallplotMIN,\smallplotMIN) (0.5*\smallplotMAX,\smallplotMAX)};
  \addplot[no marks, thin, densely dotted] coordinates {(\smallplotMIN,0.5*\smallplotMIN) (\smallplotMAX,0.5*\smallplotMAX)};
  \end{axis}
  };
\end{tikzpicture}

%% file: experiments/plots/scatter-coverage-hoeffding-dkw.tex
\pgfplotstableread[col sep=comma]{
category,  instance,                                  x,                 y,
dtmc,      coupon.200.5.2.5.exp_draws,                8.58938816693475,  4.53198184680609, 
dtmc,      coupon.300.7.3.5.exp_draws,                12.8840822504021,  6.67630235895711, 
dtmc,      coupon.350.9.4.5.exp_draws,                15.0314292921358,  7.74000240940674, 
dtmc,      coupon.850.15.4.5.exp_draws,               36.5048997094727,  18.6690519109967, 
dtmc,      egl.100.5.2.messagesA,                     4.29469408346738,  2.3227143938204,  
dtmc,      egl.200.5.8.messagesA,                     8.58938816693475,  5.14023901849148, 
dtmc,      egl.225.10.2.messagesA,                    9.66306168780159,  4.94242472557013, 
dtmc,      egl.250.10.6.messagesA,                    10.7367352086684,  5.47926148600354, 
dtmc,      egl.100.5.2.messagesB,                     4.29469408346738,  2.27724092340299, 
dtmc,      egl.200.5.8.messagesB,                     8.58938816693475,  4.80908796513669, 
dtmc,      egl.225.10.2.messagesB,                    9.66306168780159,  5.02479207765683, 
dtmc,      egl.250.10.6.messagesB,                    10.7367352086684,  5.56162883809026, 
dtmc,      leader_sync.20.3.2.time,                   0.858938816693475, 0.486916349181372,
dtmc,      leader_sync.25.4.3.time,                   1.07367352086684,  0.588783701268057,
dtmc,      leader_sync.30.5.4.time,                   1.28840822504021,  0.673177582937396,
mdp,       csma.150.2.2.time_max,                     6.44204112520106,  3.70554303011235, 
mdp,       csma.225.2.6.time_max,                     9.6630616878016,   6.75830642102725, 
mdp,       csma.225.3.4.time_max,                     9.6630616878016,   5.81604924934203, 
mdp,       csma.250.4.4.time_max,                     10.7367352086684,  6.41430642102738, 
mdp,       eajs.50.2.100.5.ExpUtil,                   2.14734704173369,  1.11812046170156, 
mdp,       eajs.80.3.150.7.ExpUtil,                   3.4357552667739,   1.78532457422154, 
mdp,       eajs.160.5.250.11.ExpUtil,                 6.8715105335478,   3.55112261886062, 
mdp,       eajs.200.6.300.13.ExpUtil,                 8.58938816693475,  4.41056143555409, 
mdp,       firewire.600.true.3.400.time_max,          25.7681645008043,  20.7531458431148, 
mdp,       firewire.600.true.36.200.time_max,         25.7681645008043,  20.6724763843185, 
mdp,       firewire.600.true.36.600.time_max,         25.7681645008043,  21.3265662543669, 
mdp,       firewire.600.true.3.400.time_sending,      25.7681645008043,  13.1168434841583, 
mdp,       firewire.600.true.36.200.time_sending,     25.7681645008043,  13.1652904249929, 
mdp,       firewire.600.true.36.600.time_sending,     25.7681645008043,  13.1882638954102, 
mdp,       resource-gathering.0.200.15.15.expgold,    8.58938816693475,  4.36011449471932, 
mdp,       resource-gathering.0.400.30.15.expgold,    17.1787763338695,  8.69728204860414, 
mdp,       resource-gathering.0.1300.100.100.expgold, 55.8310230850759,  28.1147993058767, 
mdp,       wlan.150.0.0.time_min,                     322.102056260053,  207.679601360369, 
mdp,       wlan.250.2.0.time_min,                     536.836760433422,  335.327056260052, 
mdp,       wlan.350.4.0.time_min,                     751.571464606791,  437.185326096334, 
mdp,       wlan.450.6.0.time_min,                     966.306168780159,  544.804004662155, 
mdp,       wlan.150.0.0.num_collisions,               6.44204112520106,  3.26746750343532, 
mdp,       wlan.250.2.0.num_collisions,               10.7367352086684,  5.41281454516871, 
mdp,       wlan.350.4.0.num_collisions,               15.0314292921358,  7.5591615869024,  
mdp,       wlan.450.6.0.num_collisions,               19.3261233756032,  9.70850862863667, 
}{\DKWvsHoeffdingLarge}

\pgfplotstableread[col sep=comma]{
category,  instance,                                  x,                 y,
mdp,       wlan.150.0.0.cost_max,                     1311277.47103468,  655911.109609679, 
mdp,       wlan.250.2.0.cost_max,                     2185462.45172446,  1104103.56658949, 
mdp,       wlan.350.4.0.cost_max,                     3059647.43241424,  1540110.40825825, 
mdp,       wlan.450.6.0.cost_max,                     3933832.41310403,  1976931.42579945, 
}{\DKWvsHoeffdingSmall}

\def\smallplotMIN{300000}
\def\smallplotMAX{10000000}

\begin{tikzpicture}
  \node[anchor=north west] (largeplot) at (0,0) {%
  \begin{axis}[
    width=\scatterPlotsSize,
    height=\scatterPlotsSize,
    axis equal image,
    xmin={\largeplotMIN}, ymin={\largeplotMIN}, xmax={\largeplotMAX}, ymax={\largeplotMAX},
    xmode=log, ymode=log,
    grid=major,
    grid style={color=black!11},
    axis x line=bottom,
    axis y line=left,
    xlabel={\strut{}Hoeffding},
    xlabel style={font=\scriptsize,yshift=5pt},
    ylabel={\strut{}DKW},
    ylabel style={font=\scriptsize,yshift=-13pt,xshift=9pt},
    yticklabel style={font=\scriptsize},
    xticklabel style={font=\scriptsize},
    legend pos=north east,
    legend columns=-1,
    legend style={nodes={scale=0.75, transform shape},anchor=north east,inner sep=1.5pt,yshift=16pt,xshift=3.5pt},
    legend cell align={left},
    extra x ticks = {\largeplotMAX},
    extra x tick labels = {},  
    extra x tick style = {grid=major,grid style={very thin,color=gray}},
    extra y ticks = {\largeplotMAX},
    extra y tick labels = {},  
    extra y tick style = {grid=major,grid style={very thin,color=gray}},
  ]
  \addplot[
    scatter,
    only marks,
    scatter/classes={dtmc={mark=square*,color1,mark size=.95},
                     mdp={mark=*,color5,mark size=.95},
                     incorrect={mark=10-pointed star,color4,mark size=3}},
    scatter src=explicit symbolic
  ]%
  table[x={x}, y={y}, meta={category}]{\DKWvsHoeffdingLarge};
  
  \addplot[no marks] coordinates {(\largeplotMIN,\largeplotMIN) (\largeplotMAX,\largeplotMAX)};
  \addplot[no marks, densely dotted] coordinates {(0.5*\largeplotMIN,\largeplotMIN) (0.5*\largeplotMAX,\largeplotMAX)};
  \addplot[no marks, densely dotted] coordinates {(\largeplotMIN,0.5*\largeplotMIN) (\largeplotMAX,0.5*\largeplotMAX)};
  \end{axis}
  };
  \node[anchor=north west] (smallplot) at
	([xshift=-.082\scatterPlotsSize,yshift=.35\scatterPlotsSize]largeplot.north west) {%
  \begin{axis}[
    width=.59\scatterPlotsSize,
    height=.59\scatterPlotsSize,
    axis equal image,
    xmin={\smallplotMIN}, ymin={\smallplotMIN}, xmax={\smallplotMAX}, ymax={\smallplotMAX},
    xmode=log, ymode=log,
    grid=major,
    grid style={color=black!9},
    axis x line*={top,very thin},
    axis y line*={right,very thin},
    yticklabel style={font=\fontsize{4.5}{7.2},xshift=-2pt},
    xticklabel style={font=\fontsize{4.5}{7.2},yshift=-1pt},
    axis background/.style={fill=white},
  ]
  \addplot[
	color=red,
    scatter,
    only marks,
    scatter/classes={dtmc={mark=square*,color1,mark size=.95},
                     mdp={mark=*,color5,mark size=.95},
                     incorrect={mark=10-pointed star,color4,mark size=3}},
    scatter src=explicit symbolic
  ]%
  table[x={x}, y={y}, meta={category}]{\DKWvsHoeffdingSmall};
  
  \addplot[no marks, thin] coordinates {(\smallplotMIN,\smallplotMIN) (\smallplotMAX,\smallplotMAX)};
  \addplot[no marks, very thin] coordinates {(\smallplotMIN,\smallplotMIN) (\smallplotMAX,\smallplotMIN)};
  \addplot[no marks, very thin] coordinates {(\smallplotMIN,\smallplotMIN) (\smallplotMIN,\smallplotMAX)};
  \addplot[no marks, thin, densely dotted] coordinates {(0.5*\smallplotMIN,\smallplotMIN) (0.5*\smallplotMAX,\smallplotMAX)};
  \addplot[no marks, thin, densely dotted] coordinates {(\smallplotMIN,0.5*\smallplotMIN) (\smallplotMAX,0.5*\smallplotMAX)};
  \end{axis}
  };
\end{tikzpicture}

%% file: paper.bbl
\begin{thebibliography}{10}
\providecommand{\url}[1]{\texttt{#1}}
\providecommand{\urlprefix}{URL }
\providecommand{\doi}[1]{https://doi.org/#1}

\bibitem{AGKM22}
Agarwal, C., Guha, S., Kret{\'{\i}}nsk{\'{y}}, J., Muruganandham, P.: {PAC}
  statistical model checking of mean payoff in discrete- and continuous-time
  {MDP}. In: {CAV} {(2)}. Lecture Notes in Computer Science, vol. 13372, pp.
  3--25. Springer (2022). \doi{10.1007/978-3-031-13188-2\_1}

\bibitem{AP18}
Agha, G., Palmskog, K.: A survey of statistical model checking. {ACM} Trans.
  Model. Comput. Simul.  \textbf{28}(1),  6:1--6:39 (2018).
  \doi{10.1145/3158668}

\bibitem{AC98}
Agresti, A., Coull, B.A.: Approximate is better than ``exact'' for interval
  estimation of binomial proportions. The American Statistician
  \textbf{52}(2),  119--126 (1998). \doi{10.2307/2685469}

\bibitem{And69}
Anderson, T.W.: Confidence limits for the value of an arbitrary bounded random
  variable with a continuous distribution function. Bulletin of The
  International and Statistical Institute  \textbf{43},  249--251 (1969)

\bibitem{ABHK18}
Ashok, P., Butkova, Y., Hermanns, H., Kret{\'{\i}}nsk{\'{y}}, J.:
  Continuous-time {M}arkov decisions based on partial exploration. In: Lahiri,
  S.K., Wang, C. (eds.) 16th International Symposium on Automated Technology
  for Verification and Analysis ({ATVA}). Lecture Notes in Computer Science,
  vol. 11138, pp. 317--334. Springer (2018). \doi{10.1007/978-3-030-01090-4_19}

\bibitem{ACDKM17}
Ashok, P., Chatterjee, K., Daca, P., Kret{\'{\i}}nsk{\'{y}}, J., Meggendorfer,
  T.: Value iteration for long-run average reward in {M}arkov decision
  processes. In: {CAV} {(1)}. Lecture Notes in Computer Science, vol. 10426,
  pp. 201--221. Springer (2017). \doi{10.1007/978-3-319-63387-9_10},
  \url{https://doi.org/10.1007/978-3-319-63387-9_10}

\bibitem{ADKW20}
Ashok, P., Daca, P., Kret{\'{\i}}nsk{\'{y}}, J., Weininger, M.: Statistical
  model checking: Black or white? In: Margaria, T., Steffen, B. (eds.) 9th
  International Symposium on Leveraging Applications of Formal Methods
  ({ISoLA}). Lecture Notes in Computer Science, vol. 12476, pp. 331--349.
  Springer (2020). \doi{10.1007/978-3-030-61362-4_19}

\bibitem{AKW19}
Ashok, P., Kret{\'{\i}}nsk{\'{y}}, J., Weininger, M.: {PAC} statistical model
  checking for markov decision processes and stochastic games. In: {CAV} {(1)}.
  Lecture Notes in Computer Science, vol. 11561, pp. 497--519. Springer (2019).
  \doi{10.1007/978-3-030-25540-4_29}

\bibitem{AEKSW22}
Azeem, M., Evangelidis, A., Kret{\'{\i}}nsk{\'{y}}, J., Slivinskiy, A.,
  Weininger, M.: Optimistic and topological value iteration for simple
  stochastic games. In: {ATVA}. Lecture Notes in Computer Science, vol. 13505,
  pp. 285--302. Springer (2022). \doi{10.1007/978-3-031-19992-9_18}

\bibitem{Bai16}
Baier, C.: Probabilistic model checking. In: Esparza, J., Grumberg, O.,
  Sickert, S. (eds.) Dependable Software Systems Engineering, {NATO} Science
  for Peace and Security Series -- {D}: Information and Communication Security,
  vol.~45, pp. 1--23. {IOS} Press (2016). \doi{10.3233/978-1-61499-627-9-1}

\bibitem{BAFK18}
Baier, C., de~Alfaro, L., Forejt, V., Kwiatkowska, M.: Model checking
  probabilistic systems. In: Clarke, E.M., Henzinger, T.A., Veith, H., Bloem,
  R. (eds.) Handbook of Model Checking, pp. 963--999. Springer (2018).
  \doi{10.1007/978-3-319-10575-8_28}

\bibitem{BHHK10}
Baier, C., Haverkort, B.R., Hermanns, H., Katoen, J.P.: Performance evaluation
  and model checking join forces. Commun. {ACM}  \textbf{53}(9),  76--85
  (2010). \doi{10.1145/1810891.1810912}

\bibitem{BK08}
Baier, C., Katoen, J.P.: Principles of model checking. {MIT} Press (2008)

\bibitem{BKLPW17}
Baier, C., Klein, J., Leuschner, L., Parker, D., Wunderlich, S.: Ensuring the
  reliability of your model checker: Interval iteration for {M}arkov decision
  processes. In: {CAV} {(1)}. Lecture Notes in Computer Science, vol. 10426,
  pp. 160--180. Springer (2017). \doi{10.1007/978-3-319-63387-9_8}

\bibitem{BBD+15}
Ballarini, P., Barbot, B., Duflot, M., Haddad, S., Pekergin, N.: {HASL}: A new
  approach for performance evaluation and model checking from concepts to
  experimentation. Performance Evaluation  \textbf{90},  53--77 (2015).
  \doi{10.1016/j.peva.2015.04.003}

\bibitem{Ben62}
Bennett, G.: Probability inequalities for the sum of independent random
  variables. Journal of the American Statistical Association  \textbf{57}(297),
   33--45 (1962). \doi{10.1080/01621459.1962.10482149}

\bibitem{Ber24}
Bernstein, S.: On a modification of {Chebyshev's} inequality and of the error
  formula of {Laplace}. Ann. Sci. Inst. Sav. Ukraine, Sect. Math
  \textbf{1}(4),  38--49 (1924)

\bibitem{Ber34}
Bernstein, S.: Theory of Probability. 2 edn. (1934)

\bibitem{BDHK06}
Bohnenkamp, H.C., D'Argenio, P.R., Hermanns, H., Katoen, J.P.: {MoDeST}: A
  compositional modeling formalism for hard and softly timed systems. {IEEE}
  Trans. Software Eng.  \textbf{32}(10),  812--830 (2006).
  \doi{10.1109/TSE.2006.104}

\bibitem{BCCFKKPU14}
Br{\'{a}}zdil, T., Chatterjee, K., Chmelik, M., Forejt, V.,
  Kret{\'{\i}}nsk{\'{y}}, J., Kwiatkowska, M.Z., Parker, D., Ujma, M.:
  Verification of {M}arkov decision processes using learning algorithms. In:
  Cassez, F., Raskin, J.F. (eds.) 12th International Symposium on Automated
  Technology for Verification and Analysis ({ATVA}). Lecture Notes in Computer
  Science, vol.~8837, pp. 98--114. Springer (2014).
  \doi{10.1007/978-3-319-11936-6_8}

\bibitem{BS24}
Bu, H., Sun, M.: Clopper-pearson algorithms for efficient statistical model
  checking estimation. IEEE Transactions on Software Engineering (01),  1--20
  (2024). \doi{10.1109/TSE.2024.3392720}

\bibitem{Bud22a}
Budde, C.E.: {FIG}: the {F}inite {I}mprobability {G}enerator v1.3. {SIGMETRICS}
  Perform. Evaluation Rev.  \textbf{49}(4),  59--64 (2022).
  \doi{10.1145/3543146.3543160}

\bibitem{Bud22b}
Budde, C.E.: Using statistical model checking for cybersecurity analysis. In:
  Skarmeta, A.F., Canavese, D., Lioy, A., Matheu, S.N. (eds.) First
  International Workshop on Digital Sovereignty in Cyber Security: New
  Challenges in Future Vision ({CyberSec4Europe}). Communications in Computer
  and Information Science, vol.~1807, pp. 16--32. Springer (2022).
  \doi{10.1007/978-3-031-36096-1_2}

\bibitem{BDHS20}
Budde, C.E., D'Argenio, P.R., Hartmanns, A., Sedwards, S.: An efficient
  statistical model checker for nondeterminism and rare events. Int. J. Softw.
  Tools Technol. Transf.  \textbf{22}(6),  759--780 (2020).
  \doi{10.1007/S10009-020-00563-2}

\bibitem{BDHHJT17}
Budde, C.E., Dehnert, C., Hahn, E.M., Hartmanns, A., Junges, S., Turrini, A.:
  {JANI}: Quantitative model and tool interaction. In: Legay, A., Margaria, T.
  (eds.) 23rd International Conference on Tools and Algorithms for the
  Construction and Analysis of Systems ({TACAS}). Lecture Notes in Computer
  Science, vol. 10206, pp. 151--168 (2017). \doi{10.1007/978-3-662-54580-5_9}

\bibitem{BHKKPQTZ20}
Budde, C.E., Hartmanns, A., Klauck, M., Kret{\'{\i}}nsk{\'{y}}, J., Parker, D.,
  Quatmann, T., Turrini, A., Zhang, Z.: On correctness, precision, and
  performance in quantitative verification -- qcomp 2020 competition report.
  In: Margaria, T., Steffen, B. (eds.) 9th International Symposium on
  Leveraging Applications of Formal Methods ({ISoLA}). Lecture Notes in
  Computer Science, vol. 12479, pp. 216--241. Springer (2020).
  \doi{10.1007/978-3-030-83723-5_15}

\bibitem{publishedversion}
Budde, C.E., Hartmanns, A., Meggendorfer, T., Weininger, M., Wienh{\"{o}}ft,
  P.: Sound statistical model checking for probabilities and expected rewards.
  In: Gurfinkel, A., Heule, M. (eds.) 31st International Conference on Tools
  and Algorithms for the Construction and Analysis of Systems ({TACAS}).
  Lecture Notes in Computer Science, vol. 15696, pp. 167--190. Springer (2025).
  \doi{10.1007/978-3-031-90643-5\_9}

\bibitem{BHMWW24-artifact}
Budde, C.E., Hartmanns, A., Meggendorfer, T., Weininger, M., Wienhöft, P.:
  Sound statistical model checking for probabilities and expected rewards
  (experimental reproduction package) (2025). \doi{10.5281/zenodo.14602067}

\bibitem{Cas86}
Casella, G.: Refining binomial confidence intervals. Canadian Journal of
  Statistics  \textbf{14}(2),  113--129 (1986).
  \doi{https://doi.org/10.2307/3314658}

\bibitem{Che15}
Chen, J.: Properties of a new adaptive sampling method with applications to
  scalable learning. Web Intell.  \textbf{13}(4),  215--227 (2015).
  \doi{10.3233/WEB-150322}

\bibitem{CDGL21}
Chenoy, A., Duchene, F., Given-Wilson, T., Legay, A.: {C-SMC}: A hybrid
  statistical model checking and concrete runtime engine for analyzing {C}
  programs. In: Laarman, A., Sokolova, A. (eds.) 27th International Symposium
  on Model Checking Software ({SPIN}). Lecture Notes in Computer Science, vol.
  12864, pp. 101--119. Springer (2021). \doi{10.1007/978-3-030-84629-9_6}

\bibitem{Che52}
Chernoff, H.: A measure of asymptotic efficiency for tests of a hypothesis
  based on the sum of observations. The Annals of Mathematical Statistics
  \textbf{23}(4),  493--507 (1952). \doi{10.1214/aoms/1177729330}

\bibitem{CR65}
Chow, Y.S., Robbins, H.: {On the Asymptotic Theory of Fixed-Width Sequential
  Confidence Intervals for the Mean}. The Annals of Mathematical Statistics
  \textbf{36}(2),  457--462 (1965). \doi{10.1214/aoms/1177700156}

\bibitem{CP34}
Clopper, C., Pearson, E.: The use of confidence or fiducial limits illustrated
  in the case of the binomial. Biometrika  \textbf{26}(4),  404--413 (1934).
  \doi{10.1093/biomet/26.4.404}

\bibitem{DHKP17}
Daca, P., Henzinger, T.A., Kret{\'{\i}}nsk{\'{y}}, J., Petrov, T.: Faster
  statistical model checking for unbounded temporal properties. {ACM} Trans.
  Comput. Log.  \textbf{18}(2),  12:1--12:25 (2017). \doi{10.1145/3060139}

\bibitem{DM18}
D'Argenio, P.R., Monti, R.E.: Input/output stochastic automata with urgency:
  Confluence and weak determinism. In: Fischer, B., Uustalu, T. (eds.) 15th
  International Colloquium on Theoretical Aspects of Computing ({ICTAC}).
  Lecture Notes in Computer Science, vol. 11187, pp. 132--152. Springer (2018).
  \doi{10.1007/978-3-030-02508-3_8}

\bibitem{DLLMP15}
David, A., Larsen, K.G., Legay, A., Mikucionis, M., Poulsen, D.B.: {U}ppaal
  {SMC} tutorial. Int. J. Softw. Tools Technol. Transf.  \textbf{17}(4),
  397--415 (2015). \doi{10.1007/S10009-014-0361-Y}

\bibitem{DLL+15}
David, A., Larsen, K.G., Legay, A., Miku{\v{c}}ionis, M., Poulsen, D.B.:
  {Uppaal} {SMC} tutorial. Int. J. Softw. Tools Technol. Transf.
  \textbf{17}(4),  397--415 (2015). \doi{10.1007/s10009-014-0361-y}

\bibitem{DGW02}
Domingo, C., Gavald{\`{a}}, R., Watanabe, O.: Adaptive sampling methods for
  scaling up knowledge discovery algorithms. In: Discovery Science. Lecture
  Notes in Computer Science, vol.~1721, pp. 172--183. Springer (1999).
  \doi{10.1007/3-540-46846-3_16}

\bibitem{DKW56}
Dvoretzky, A., Kiefer, J., Wolfowitz, J.: Asymptotic minimax character of the
  sample distribution function and of the classical multinomial estimator. The
  Annals of Mathematical Statistics  \textbf{27}(3),  642--669 (1956).
  \doi{10.1214/aoms/1177728174}

\bibitem{EKKW22}
Eisentraut, J., Kelmendi, E., Kret{\'{\i}}nsk{\'{y}}, J., Weininger, M.: Value
  iteration for simple stochastic games: {S}topping criterion and learning
  algorithm. Inf. Comput.  \textbf{285}(Part),  104886 (2022).
  \doi{10.1016/j.ic.2022.104886},
  \url{https://doi.org/10.1016/j.ic.2022.104886}

\bibitem{FKNP11}
Forejt, V., Kwiatkowska, M.Z., Norman, G., Parker, D.: Automated verification
  techniques for probabilistic systems. In: Bernardo, M., Issarny, V. (eds.)
  11th International School on Formal Methods for the Design of Computer,
  Communication and Software Systems ({SFM}). Lecture Notes in Computer
  Science, vol.~6659, pp. 53--113. Springer (2011).
  \doi{10.1007/978-3-642-21455-4_3}

\bibitem{Fre10}
Frey, J.: Fixed-width sequential confidence intervals for a proportion. The
  American Statistician  \textbf{64}(3),  242--249 (2010),
  \url{http://www.jstor.org/stable/20799919}

\bibitem{GRV17}
Gilmore, S., Reijsbergen, D., Vandin, A.: Transient and steady-state
  statistical analysis for discrete event simulators. In: {IFM}. pp. 145--160.
  Springer (2017). \doi{10.1007/978-3-319-66845-1_10}

\bibitem{HM18}
Haddad, S., Monmege, B.: Interval iteration algorithm for mdps and imdps.
  Theor. Comput. Sci.  \textbf{735},  111--131 (2018).
  \doi{10.1016/J.TCS.2016.12.003}

\bibitem{HHHK13}
Hahn, E.M., Hartmanns, A., Hermanns, H., Katoen, J.P.: A compositional
  modelling and analysis framework for stochastic hybrid systems. Formal
  Methods Syst. Des.  \textbf{43}(2),  191--232 (2013).
  \doi{10.1007/S10703-012-0167-Z}

\bibitem{HH14}
Hartmanns, A., Hermanns, H.: The {M}odest {T}oolset: An integrated environment
  for quantitative modelling and verification. In: {\'{A}}brah{\'{a}}m, E.,
  Havelund, K. (eds.) 20th International Conference on Tools and Algorithms for
  the Construction and Analysis of Systems ({TACAS}). Lecture Notes in Computer
  Science, vol.~8413, pp. 593--598. Springer (2014).
  \doi{10.1007/978-3-642-54862-8_51}

\bibitem{HK20}
Hartmanns, A., Kaminski, B.L.: Optimistic value iteration. In: Lahiri, S.K.,
  Wang, C. (eds.) 32nd International Conference on Computer Aided Verification
  ({CAV}). Lecture Notes in Computer Science, vol. 12225, pp. 488--511.
  Springer (2020). \doi{10.1007/978-3-030-53291-8_26}

\bibitem{HKPQR19}
Hartmanns, A., Klauck, M., Parker, D., Quatmann, T., Ruijters, E.: The
  quantitative verification benchmark set. In: Vojnar, T., Zhang, L. (eds.)
  25th International Conference on Tools and Algorithms for the Construction
  and Analysis of Systems ({TACAS}). Lecture Notes in Computer Science, vol.
  11427, pp. 344--350. Springer (2019). \doi{10.1007/978-3-030-17462-0_20}

\bibitem{HJKQV22}
Hensel, C., Junges, S., Katoen, J.P., Quatmann, T., Volk, M.: The probabilistic
  model checker {S}torm. Int. J. Softw. Tools Technol. Transf.  \textbf{24}(4),
   589--610 (2022). \doi{10.1007/S10009-021-00633-Z}

\bibitem{HLMP04}
H{\'{e}}rault, T., Lassaigne, R., Magniette, F., Peyronnet, S.: Approximate
  probabilistic model checking. In: Steffen, B., Levi, G. (eds.) 5th
  International Conference on Verification, Model Checking, and Abstract
  Interpretation ({VMCAI}). Lecture Notes in Computer Science, vol.~2937, pp.
  73--84. Springer (2004). \doi{10.1007/978-3-540-24622-0_8}

\bibitem{Hoe63}
Hoeffding, W.: Probability inequalities for sums of bounded random variables.
  Journal of the American Statistical Association  \textbf{58}(301),  13--30
  (1963). \doi{10.1080/01621459.1963.10500830}

\bibitem{JSD19}
J{\'{e}}gourel, C., Sun, J., Dong, J.S.: Sequential schemes for frequentist
  estimation of properties in statistical model checking. {ACM} Trans. Model.
  Comput. Simul.  \textbf{29}(4),  25:1--25:22 (2019). \doi{10.1145/3310226}

\bibitem{Kre16}
Kret{\'{\i}}nsk{\'{y}}, J.: Survey of statistical verification of linear
  unbounded properties: Model checking and distances. In: ISoLA {(1)}. Lecture
  Notes in Computer Science, vol.~9952, pp. 27--45 (2016).
  \doi{10.1007/978-3-319-47166-2_3}

\bibitem{KM20}
Kret{\'{\i}}nsk{\'{y}}, J., Meggendorfer, T.: Of cores: A partial-exploration
  framework for {M}arkov decision processes. Log. Methods Comput. Sci.
  \textbf{16}(4) (2020), \url{https://lmcs.episciences.org/6833}

\bibitem{DBLP:conf/lics/KretinskyMW23}
Kret{\'{\i}}nsk{\'{y}}, J., Meggendorfer, T., Weininger, M.: Stopping criteria
  for value iteration on stochastic games with quantitative objectives. In:
  38th Annual {ACM/IEEE} Symposium on Logic in Computer Science, {LICS} 2023,
  Boston, MA, USA, June 26-29, 2023. pp. 1--14. {IEEE} (2023).
  \doi{10.1109/LICS56636.2023.10175771}

\bibitem{KNP11}
Kwiatkowska, M.Z., Norman, G., Parker, D.: {PRISM} 4.0: Verification of
  probabilistic real-time systems. In: Gopalakrishnan, G., Qadeer, S. (eds.)
  23rd International Conference on Computer Aided Verification ({CAV}). Lecture
  Notes in Computer Science, vol.~6806, pp. 585--591. Springer (2011).
  \doi{10.1007/978-3-642-22110-1_47}

\bibitem{LKP24}
Lampacrescia, M., Klauck, M., Palmas, M.: Towards verifying robotic systems
  using statistical model checking in {STORM}. In: Steffen, B. (ed.) Second
  International Conference on Bridging the Gap Between {AI} and Reality
  ({AISoLA}). Lecture Notes in Computer Science, vol. 15217, pp. 446--467.
  Springer (2024). \doi{10.1007/978-3-031-75434-0\_28}

\bibitem{LMZ23}
Lanotte, R., Merro, M., Zannone, N.: Impact analysis of coordinated
  cyber-physical attacks via statistical model checking: A case study. In:
  Huisman, M., Ravara, A. (eds.) 43rd {IFIP} {WG} 6.1 International Conference
  on Formal Techniques for Distributed Objects, Components, and Systems
  ({FORTE}). Lecture Notes in Computer Science, vol. 13910, pp. 75--94.
  Springer (2023). \doi{10.1007/978-3-031-35355-0_6}

\bibitem{LT19}
Learned-Miller, E.G., Thomas, P.S.: A new confidence interval for the mean of a
  bounded random variable. CoRR  \textbf{abs/1905.06208} (2019),
  \url{http://arxiv.org/abs/1905.06208}

\bibitem{LLTYSG19}
Legay, A., Lukina, A., Traonouez, L.M., Yang, J., Smolka, S.A., Grosu, R.:
  Statistical model checking. In: Steffen, B., Woeginger, G.J. (eds.) Computing
  and Software Science -- State of the Art and Perspectives, Lecture Notes in
  Computer Science, vol. 10000, pp. 478--504. Springer (2019).
  \doi{10.1007/978-3-319-91908-9_23}

\bibitem{LST16}
Legay, A., Sedwards, S., Traonouez, L.M.: {P}lasma {L}ab: A modular statistical
  model checking platform. In: Margaria, T., Steffen, B. (eds.) 7th
  International Symposium on Leveraging Applications of Formal Methods
  ({ISoLA}). Lecture Notes in Computer Science, vol.~9952, pp. 77--93 (2016).
  \doi{10.1007/978-3-319-47166-2_6}

\bibitem{LP17}
Levin, D.A., Peres, Y.: Markov chains and mixing times, vol.~107. American
  Mathematical Soc. (2017)

\bibitem{Mas90}
Massart, P.: The tight constant in the {D}voretzky-{K}iefer-{W}olfowitz
  inequality. The Annals of Probability  \textbf{18}(3),  1269--1283 (1990).
  \doi{10.1214/aop/1176990746}

\bibitem{MTWPS23}
Mazurek, F., Tschand, A., Wang, Y., Pajic, M., Sorin, D.J.: Rigorous evaluation
  of computer processors with statistical model checking. In: 56th Annual
  {IEEE/ACM} International Symposium on Microarchitecture ({MICRO}). pp.
  1242--1254. {ACM} (2023). \doi{10.1145/3613424.3623785}

\bibitem{DBLP:conf/cav/MeggendorferW24}
Meggendorfer, T., Weininger, M.: Playing games with your {PET}: Extending the
  partial exploration tool to stochastic games. In: Gurfinkel, A., Ganesh, V.
  (eds.) Computer Aided Verification - 36th International Conference, {CAV}
  2024, Montreal, QC, Canada, July 24-27, 2024, Proceedings, Part {III}.
  Lecture Notes in Computer Science, vol. 14683, pp. 359--372. Springer (2024).
  \doi{10.1007/978-3-031-65633-0\_16}

\bibitem{MWW24}
Meggendorfer, T., Weininger, M., Wienh{\"o}ft, P.: What are the odds?
  {I}mproving the foundations of statistical model checking. CoRR
  \textbf{abs/2404.05424} (2024). \doi{10.48550/arXiv.2404.05424}

\bibitem{MSA08}
Mnih, V., Szepesv{\'{a}}ri, C., Audibert, J.Y.: Empirical {B}ernstein stopping.
  In: Cohen, W.W., McCallum, A., Roweis, S.T. (eds.) 25th International
  Conference on Machine Learning ({ICML}). {ACM} International Conference
  Proceeding Series, vol.~307, pp. 672--679. {ACM} (2008).
  \doi{10.1145/1390156.1390241}

\bibitem{New98}
Newcombe, R.G.: Two-sided confidence intervals for the single proportion:
  Comparison of seven methods. Statistics in medicine  \textbf{17}(8),
  857--872 (1998)

\bibitem{NPR20}
Niehage, M., Pilch, C., Remke, A.: Simulating hybrid {P}etri nets with general
  transitions and non-linear differential equations. In: 13th {EAI}
  International Conference on Performance Evaluation Methodologies and Tools
  ({VALUETOOLS}). pp. 88--95. {ACM} (2020). \doi{10.1145/3388831.3388842}

\bibitem{NMB+18}
Nouri, A., Mediouni, B.L., Bozga, M., Combaz, J., Bensalem, S., Legay, A.:
  Performance evaluation of stochastic real-time systems with the
  {$\pazocal{S}$BIP} framework. International Journal of Critical
  Computer-Based Systems  \textbf{8}(3-4),  340--370 (2018).
  \doi{10.1504/IJCCBS.2018.096439}

\bibitem{Oka59}
Okamoto, M.: Some inequalities relating to the partial sum of binomial
  probabilities. Annals of the Institute of Statistical Mathematics
  \textbf{10}(1),  29--35 (1959)

\bibitem{PL24}
Parmentier, M., Legay, A.: Adaptive stopping algorithms based on concentration
  inequalities. In: Steffen, B. (ed.) Second International Conference on
  Bridging the Gap Between {AI} and Reality ({AISoLA}). Lecture Notes in
  Computer Science, vol. 15217, pp. 336--353. Springer (2024).
  \doi{10.1007/978-3-031-75434-0\_23}

\bibitem{PTL21}
Phan, M., Thomas, P.S., Learned-Miller, E.G.: Towards practical mean bounds for
  small samples. In: Meila, M., Zhang, T. (eds.) 38th International Conference
  on Machine Learning ({ICML}). Proceedings of Machine Learning Research,
  vol.~139, pp. 8567--8576. {PMLR} (2021),
  \url{http://proceedings.mlr.press/v139/phan21a.html}

\bibitem{PER17}
Pilch, C., Edenfeld, F., Remke, A.: {HYPEG}: Statistical model checking for
  hybrid {P}etri nets: Tool paper. In: Marin, A., Houdt, B.V., Casale, G.,
  Petriu, D.C., Rossi, S. (eds.) 11th {EAI} International Conference on
  Performance Evaluation Methodologies and Tools ({VALUETOOLS}). pp. 186--191.
  {ACM} (2017). \doi{10.1145/3150928.3150956}

\bibitem{PR17}
Pilch, C., Remke, A.: Statistical model checking for hybrid petri nets with
  multiple general transitions. In: {DSN}. pp. 475--486. {IEEE} Computer
  Society (2017). \doi{10.1109/DSN.2017.41}

\bibitem{QK18}
Quatmann, T., Katoen, J.P.: Sound value iteration. In: Chockler, H.,
  Weissenbacher, G. (eds.) 30th International Conference on Computer Aided
  Verification ({CAV}). Lecture Notes in Computer Science, vol. 10981, pp.
  643--661. Springer (2018). \doi{10.1007/978-3-319-96145-3_37}

\bibitem{RBSH15}
Reijsbergen, D., de~Boer, P., Scheinhardt, W.R.W., Haverkort, B.R.: On
  hypothesis testing for statistical model checking. Int. J. Softw. Tools
  Technol. Transf.  \textbf{17}(4),  377--395 (2015).
  \doi{10.1007/S10009-014-0350-1}

\bibitem{RLHBRCZ21}
Roberts, R., Lewis, B., Hartmanns, A., Basu, P., Roy, S., Chakraborty, K.,
  Zhang, Z.: Probabilistic verification for reliability of a two-by-two
  network-on-chip system. In: Lluch-Lafuente, A., Mavridou, A. (eds.) 26th
  International Conference on Formal Methods for Industrial Critical Systems
  ({FMICS}). Lecture Notes in Computer Science, vol. 12863, pp. 232--248.
  Springer (2021). \doi{10.1007/978-3-030-85248-1_16}

\bibitem{RT09}
Rubino, G., Tuffin, B. (eds.): Rare Event Simulation using {M}onte {C}arlo
  Methods. Wiley (2009). \doi{10.1002/9780470745403}

\bibitem{SV13}
Sebastio, S., Vandin, A.: {M}ulti{V}e{S}t{A}: statistical model checking for
  discrete event simulators. In: Horv{\'{a}}th, A., Buchholz, P., Cortellessa,
  V., Muscariello, L., Squillante, M.S. (eds.) 7th International Conference on
  Performance Evaluation Methodologies and Tools ({VALUETOOLS}). pp. 310--315.
  {ICST/ACM} (2013). \doi{10.4108/ICST.VALUETOOLS.2013.254377}

\bibitem{Wal45}
Wald, A.: Sequential tests of statistical hypotheses. The Annals of
  Mathematical Statistics  \textbf{16}(2),  117--186 (1945).
  \doi{10.1214/aoms/1177731118}

\bibitem{Wan14}
Wang, W.: An iterative construction of confidence intervals for a proportion.
  Statistica Sinica  \textbf{24}(3),  1389--1410 (2014),
  \url{http://www.jstor.org/stable/24310993}

\bibitem{YKNP06}
Younes, H.L.S., Kwiatkowska, M.Z., Norman, G., Parker, D.: Numerical vs.
  statistical probabilistic model checking. Int. J. Softw. Tools Technol.
  Transf.  \textbf{8}(3),  216--228 (2006). \doi{10.1007/S10009-005-0187-8}

\bibitem{YS02}
Younes, H.L.S., Simmons, R.G.: Probabilistic verification of discrete event
  systems using acceptance sampling. In: Brinksma, E., Larsen, K.G. (eds.) 14th
  International Conference on Computer Aided Verification ({CAV}). Lecture
  Notes in Computer Science, vol.~2404, pp. 223--235. Springer (2002).
  \doi{10.1007/3-540-45657-0_17}

\bibitem{Zul15}
Zuliani, P.: Statistical model checking for biological applications. Int. J.
  Softw. Tools Technol. Transf.  \textbf{17}(4),  527--536 (2015).
  \doi{10.1007/S10009-014-0343-0}

\end{thebibliography}
